\documentclass{article}

\usepackage{booktabs} % For formal tables
\usepackage{enumitem} % For better itemize env
\usepackage{fullpage}
\usepackage{natbib}
\usepackage{amsmath,amsthm,amssymb}

\usepackage{soul}
\usepackage[noend,linesnumbered,noline,nofillcomment]{algorithm2e}

\usepackage{tikz}
\usetikzlibrary{positioning,shadows,shapes,backgrounds,fit}

\usepackage{subcaption}

\newcommand{\out}{\mathcal S}
\newcommand{\vett}[1]{\mathbf{#1} }
\renewcommand{\a}{{\vett a}}
\newcommand{\ai}{{\vett a}_{-i}}
\renewcommand{\b}{\vett b}
\newcommand{\bi}{\b_{-i}}
\renewcommand{\xi}{\vett{x}_{-i}}

\newcommand{\x}{\vett{x}}
\newcommand{\y}{\vett{y}}
\newcommand{\z}{\vett{z}}
\newcommand{\w}{\vett{w}}
\newcommand{\p}{\mathbf{p}}

\newcommand{\real}{\mbox{$\mathbb{R}$}}
\newcommand{\M}{\mathcal{M}}
\newcommand{\T}{\mathcal{T}}
\newcommand{\vgraph}{OSP-graph}
\newcommand{\ver}{{\mathcal O}_{i}^{\mathcal{T}}}

\newcommand{\A}{\mathcal{A}}
\newcommand{\Sm}{\mathcal{M}_\mathsf{set}^\mathsf{opt}}
\newcommand{\OPT}{\mathcal{OPT}}
\newcommand{\notP}{\bar{P}}
\newcommand{\shp}{{\textsc{SP}}}

\newtheorem{theorem}{Theorem}
\newtheorem{lemma}[theorem]{Lemma}
\newtheorem{proposition}[theorem]{Proposition}
\newtheorem{definition}[theorem]{Definition}
\newtheorem{corollary}[theorem]{Corollary}
\newtheorem{obs}[theorem]{Observation}
\newtheorem{fact}[theorem]{Fact}
\newtheorem{example}[theorem]{Example}

\usepackage{todonotes} %% Disable ToDo Notes for final version: write "\usepackage[disable]{todonotes}"
\DeclareFontFamily{U}{mathx}{\hyphenchar\font45}
\DeclareFontShape{U}{mathx}{m}{n}{
	<5> <6> <7> <8> <9> <10>
	<10.95> <12> <14.4> <17.28> <20.74> <24.88>
	mathx10
}{}
\DeclareSymbolFont{mathx}{U}{mathx}{m}{n}
\DeclareMathSymbol{\bigtimes}{1}{mathx}{"91}

\newcommand{\block}[1]{
	\begin{center}
		{
			\fbox{
				\begin{minipage}[t]{0.95\textwidth}
					#1
				\end{minipage}
			}
		}
	\end{center}
}

\begin{document}

% Title portion. Note the short title for running heads 
\title{On the approximation guarantee of obviously strategyproof mechanisms} %for single-parameter settings}  
\author{Diodato Ferraioli\thanks{Universit\`a di Salerno, Email: {\tt dferraioli@unisa.it}.} \and Adrian Meier\thanks{ETH Zurich, Email: {\tt meiera@student.ethz.ch}, {\tt paolo.penna@inf.ethz.ch}.} \and Paolo Penna\footnotemark[2] \and Carmine Ventre\thanks{University of Essex, Email: {\tt c.ventre@essex.ac.uk}.}}
\date{}

\maketitle

% note that the abstract must come before \maketitle
\begin{abstract}
Catering to the incentives of people with limited rationality is a challenging research direction that requires novel paradigms to design mechanisms and approximation algorithms. Obviously strategyproof (OSP) mechanisms have recently emerged as the concept of interest to this research agenda. However, the majority of the literature in the area has either highlighted the shortcomings of OSP or focused on the ``right'' definition rather than on the construction of these mechanisms. 

We here give the first set of \emph{tight} results on the approximation guarantee of OSP mechanisms for scheduling related machines and a \emph{characterization} of optimal OSP mechanisms for set system problems. By extending the well-known cycle monotonicity technique, we are able to concentrate on the algorithmic component of OSP mechanisms and provide some novel paradigms for their design. 
\end{abstract}

% Renew this after \maketitle if the default list of authors is too long for headers
%\renewcommand{\shortauthors}{W.\ Vickrey et.\ al.}

\section{Introduction}
Mechanism design has been a very active research area that aims to develop algorithms that align the objectives of the designer (e.g., optimality of the solution) with the incentives of self-interested agents (e.g., maximize their own utility). One of the main obstacles to its application in real settings is the assumption of full rationality. Where theory predicts that people should not strategize,
lab experiments in the field show that they do (to their own disadvantage):
this is, for example, the case for the renown Vickrey's second-price auction;
proved to be truthful and yet bidders lie when submitting sealed bids.
Interestingly, however, lies are less frequent when the very same mechanism is implemented via an ascending auction \citep{kagel1987information}.

A vague explanation of this phenomenon is that, from the point of view of a bidder, the truthfulness (a.k.a. strategyproofness) of an ascending price auction is \emph{easier to grasp} than the truthfulness of the second-price sealed bid auction \citep{ausubel2004efficient}. The %main 
key difference here is on how these auctions are \emph{implemented}:
\begin{itemize}
	\item In the second-price sealed-bid auction (direct-revelation implementation), each bidder  submits her own bid \emph{once} (either her true valuation or a different value). This mechanism  is truthful meaning that truth-telling is a dominant strategy: \emph{for every report of the other bidders}, the utility when truth-telling is not worse than the utility when bidding untruthfully.
	\item In the ascending price auction  (extensive-form implementation), each bidder is repeatedly offered some price which she can accept (stay in the auction) or reject (leave the auction). In this auction,  momentarily \emph{accepting a good price} guarantees a non-negative utility, while \emph{rejecting a good price} or \emph{accepting a bad price} yield non-positive utility. Here good price refers to the private valuation of the bidder and, intuitively, truth-telling in this auction means accepting prices as long as they are not above the true valuation. 
\end{itemize}
Intuitively speaking, in the second type of auction, it is \emph{obvious} for a bidder to decide her strategy, because the utility for the \emph{worst scenario} when truth-telling is at least as good as that of the \emph{best scenario} when cheating. The recent definition of \emph{obviously strategyproof (OSP)} mechanisms \citep{li2015obviously} formalizes this argument:
ascending auctions are OSP mechanisms, while sealed-bid auctions are not. Interestingly, \citet{li2015obviously} proves that a mechanism is OSP \emph{if and only if} truth-telling is dominant even for bidders who lack contingent reasoning skills.

As being OSP is \emph{stronger} than being strategyproof,  it is natural to ask if this has an impact on what can be done by such mechanisms. For instance, the so-called \emph{deferred-acceptance} (DA) mechanisms \citep{Milgrom2014DeferredacceptanceAA} are OSP (as they essentially are ascending price auctions), but unfortunately their performance (approximation guarantee) for several optimization problems is quite poor compared to what strategyproof mechanisms can do \citep{DGR17}. Whether this is an inherent limitation of OSP mechanisms or just of this technique is not clear.

One of the reasons behind this open question might be the absence of a general technique for designing OSP mechanisms and the lack of an algorithmic understanding of OSP mechanisms. Specifically, it is well known that strategyproofness is equivalent to certain \emph{monotonicity} conditions of the \emph{algorithm} used by the mechanism for computing the solution (being it an allocation of goods or a path in a network with self-interested agents). Therefore, one can essentially focus on the algorithmic part and study questions regarding the approximation and the complexity. The same type of questions seem much more challenging for OSP mechanisms, as such characterizations are not known. 
Recent work in the area \citep{BadeG17,pycia2016obvious,mackenzie2017revelation}
aims at simplifying the notion of OSP, by looking at versions of the revelation principle
for OSP mechanisms. This, for example, allows to think, without loss of generality,
at deterministic (rather than randomized) extensive-form mechanisms
where each agent moves sequentially (rather than concurrently).
Nevertheless, this line of work does not seem to help much in thinking algorithmically.

%\subsection{An example}
%\label{sec:path}
\paragraph{An illustrative example (path auctions).}
To better exemplify this conundrum, let us consider path auctions introduced by \cite{NisRon99}. In this problem, 
each edge corresponds to a link that is owned by a selfish agent, the cost for using link $i$ is some private nonnegative value $t_i$ which is known only to
agent $i$, and the goal is to pick a path of minimal total cost (shortest path) between two given nodes $s$ and  $t$. Since the standard VCG mechanism yields an (exact) strategyproof mechanism for this problem \citep{NisRon99}, one might consider the following natural question:  
\begin{quote}
	\emph{Can we compute the shortest path whilst guaranteeing OSP?}
\end{quote}
 For a graph consisting of \emph{parallel links}, we know from  \citet{li2015obviously} that the answer is yes via a simple %in the case in which $V=\{s,t\}$ and $E$ is a set of parallel edges. \citet{li2015obviously} shows that we can run a 
descending auction to select the cheaper edge. %and make truthtelling an obviously dominant strategy 
%(for more details, see Sect. \ref{sec:warmup2}). 
%$G=(V,E)$. Each edge corresponds to a router that is owned by a selfish agent. The cost for using link $i$ is known only to agent $i$ and is equal to some nonnegative value $t_i$. Can we compute the shortest path whilst guaranteeing OSP? We know from Li's original paper \citep{li2015obviously} that the answer is yes in the case in which $V=\{s,t\}$ and $E$ is a set of parallel edges. \citet{li2015obviously} shows that we can run a descending auction to select the cheaper edge and make truthtelling an obviously dominant strategy (for more details, see Sect. \ref{sec:warmup2}).
Already for \emph{slightly} more general graphs, the answer is unclear. Consider, for example, the graph in Figure~\ref{graph21}. To make things even simpler let us restrict to a two-value domain $\{L,H\}$,
i.e., edges cost either $L$ or $H>L$.
In this setting, a simple OSP mechanism can be designed by querying the agents according to the \emph{implementation tree} (i.e., a querying protocol where different actions are taken according to the answers received) in Figure~\ref{tree21}. This algorithm is augmented with the following payments: $H$ for edges in the selected path, $0$ otherwise.
It is not hard to see that, for every edge $e$,
it is not possible that $e$ is selected when she declares that her type is $H$ and it is not selected when she says $L$. In particular, edge $(s,t)$ is always selected when she says $L$, while the remaining edges are never selected when they declare $H$. Then, if $e$ declares her true type, she receives a utility of $H-L$ if the true type is $L$ and $e$ is selected, and $0$ otherwise; by inspection, she would receive at most the same utility when cheating. It turns out that this argument is enough to prove that the mechanism is indeed OSP. %Since this argument is completely independent from the declaration of the other edges, we can conclude that the mechanism is OSP.

Does the same approach work, for example, on the \emph{slightly} more general graph in Figure~\ref{graph22}?
%Unfortunately, this is not the case. 
%Consider, indeed, the graph in Figure~\ref{graph22}.
Consider an edge $e$ that is queried before the type of the remaining edges is known
(that is, the first edge to be queried in a sequential mechanism,
or an arbitrary edge in a direct revelation mechanism).
Suppose that the type of this edge is $L$.
If she declares her type truthfully, then the worst that may occur is
that the corresponding path is not selected
(that occurs when this path costs $H+L$ and the alternative path costs $2L$),
and thus $e$ receives utility $0$.
If this edge, instead, cheats and declares $H$, then 
it is possible that the corresponding path is selected
(if it costs $H+L$ and the alternative path costs $2H$)
and $e$ receives utility $H-L$.
Thus, it is not obvious for an edge $e$ lacking contingent reasoning skills,
to understand that revealing the type truthfully is dominant.
This raises the following questions:
Is there a different payment rule that enables to design
an optimal OSP mechanism in the latter case?
%If yes, how can we find this payment rule? 
If not, how good an approximation can we compute?

\begin{figure}[htbp]
	\centering
	\begin{subfigure}[b]{0.25\linewidth}
		\centering
		\begin{tikzpicture}[shorten >=1pt,node distance=1.5cm]
		\tikzstyle{place}=[circle,draw]
		\node[place] (x)  		     {$s$};
		\node[place] (y)   [above right of=x]  {$u$};
		\node[place] (z) [below right of=y]  {$t$};
		\path[-] (x) edge (y);
		\path[-] (y) edge (z);
		\path[-] (x) edge (z);
		\end{tikzpicture}
		\caption{\label{graph21}}
	\end{subfigure}%
	\begin{subfigure}[b]{0.5\linewidth}
		\centering
		\begin{tikzpicture}[shorten >=1pt,node distance=1.5cm,inner sep=1pt]
		\tikzstyle{place}=[circle,draw]
		\tikzstyle{placew}=[rectangle]
		\node[place] (x)  		     {$(s,t)$};
		\node[place] (y)   [below left of=x]  {$(s,u)$};
		\node[placew] (a)   [below right of=x]  {$(s,t)$};
		\node[place] (z) [below left of=y]  {$(u,t)$};
		\node[placew] (b)   [below right of=y]  {$(s,t)$};
		\node[placew] (c)   [below left of=z]  {$(s,u,t)$};
		\node[placew] (d)   [below right of=z]  {$(s,t)$};
		\path[-] (x) edge node [label={$H$}] {} (y);
		\path[-] (x) edge node [label={$L$}] {} (a);
		\path[-] (y) edge node [label={$L$}] {} (z);
		\path[-] (y) edge node [label={$H$}] {} (b);
		\path[-] (z) edge node [label={$L$}] {} (c);
		\path[-] (z) edge node [label={$H$}] {} (d);
		\end{tikzpicture}
		\caption{\label{tree21}}
	\end{subfigure}%
	\begin{subfigure}[b]{0.25\linewidth}
		\centering
		\begin{tikzpicture}[shorten >=1pt,node distance=1.5cm]
		\tikzstyle{place}=[circle,draw]
		\node[place] (x)  		     {$s$};
		\node[place] (y)   [above right of=x]  {$u$};
		\node[place] (z) [below right of=y]  {$t$};
		\node[place] (w) [below right of=x]  {$v$};
		\path[-] (x) edge (y);
		\path[-] (y) edge (z);
		\path[-] (x) edge (w);
		\path[-] (w) edge (z);
		\end{tikzpicture}
		\caption{\label{graph22}}
	\end{subfigure}
	\caption{OSP mechanisms for path auctions}
\end{figure}
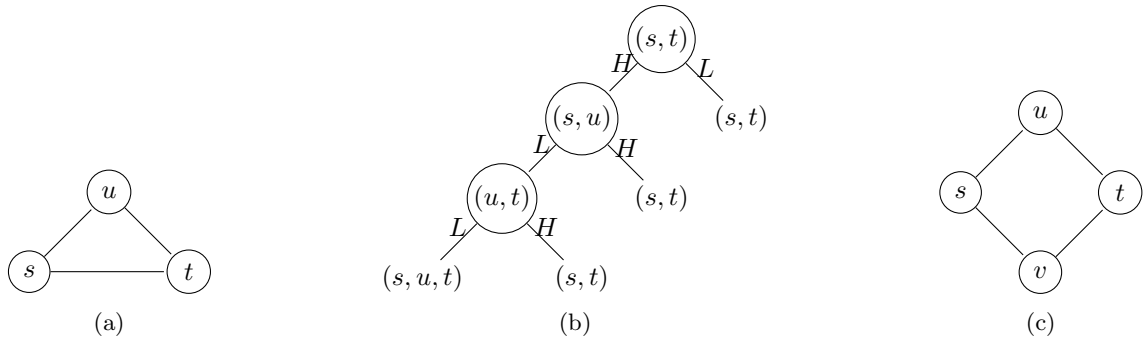

% 
% This simple idea, unfortunately, breaks down already when we add another vertex to the graph. For $V=\{s,u,v,t\}$ and $E=\{(s,u), (u,t), (s,v), (s,t)\}$, in fact, should we query $(s,u)$, with type $L$, first then her worst utility when truthtelling would be $0$ (this is when the alternative path would cost $2L$ and $(u,t)$ would cost $H$) whilst the best utility when lying would be $H-L>0$ (this is when $(u,t)$ costs $L$ and both $(s,v)$ and $(s,t)$ cost $H$).
% (Figure \ref{fig:sp} summarizes these examples.)
% 
% 
% \begin{figure}[htb]
% 	\includegraphics[scale=0.35]{21v22.png}
% 	\caption{Studying OSP optimal mechanisms for shortest path on simple graphs}
% 	\label{fig:sp}
% \end{figure}
% 
% 
%From here on it's straightforward \LaTeX. You should know that the style uses Natbib author--year citations.  That means citations that work like footnotes use ``cite'' \citep{Abril07} and citations that work like a part of speech are cited as we do \citet{Abril07}.
\bigskip
The goal of this work is  build the foundations to reason about OSP
algorithmically and answer this kind of questions.  In particular, we  advance
the state of the art by providing an algorithmic characterization of OSP mechanisms. Among others, our results show why deferred acceptance (DA) auctions \citep{Milgrom2014DeferredacceptanceAA} -- essentially the only known technique to design OSP mechanisms -- do not fully capture the power of a ``generic'' OSP mechanism, as the latter may exploit some aspects of the implementation (extensive-form game) in a crucial way.
\subsection{Our contribution}
%\smallskip \noindent {\bf Our contribution.}
To give an algorithmic characterization of OSP mechanisms,
we extend the well known cycle-monotonicity (CMON) technique.
This approach allows to abstract the truthfulness of an algorithm in terms of non-negative weight cycles on suitably defined graphs. %In its more general form, the graph for truthfulness of agent $i$ is complete and has a vertex for each possible outcome; edge weights account for the incentives of agent $i$ in forcing a different outcome by lying. 
We show that non-negative weight cycles continue to characterize OSP when the graph of interest is carefully defined. Our main conceptual contribution is here a way to accommodate the OSP constraints, which \emph{depend on the particular extensive-form implementation} of the mechanism, in the machinery of CMON, which is designed to focus on the algorithmic output of mechanism. %We have here a graph for each agent $i$ as well, but there is a node for each strategy profile and we add an edge between two vertices only when there is an OSP constraint for $i$ involving those profiles; the OSP constraints \emph{depend on the particular extensive-form implementation} of the mechanism using the algorithm at hand.  We prove that payments exist that can be paired to the algorithm, implemented in extensive-form, to make an OSP mechanism if and only if this graph does not contain negative cycles. 
Interestingly, our technique shows the interplay between algorithms (which outcome/solution to return) and how this is implemented as an extensive-form game (the implementation tree). Roughly speaking, our characterization says which algorithms can be used 
for \emph{any} choice of the implementation tree. The ability to choose between different implementation trees is what gives extra power to the designer: for example, the construction of OSP mechanisms based on  DA auctions \citep{Milgrom2014DeferredacceptanceAA} uses always \emph{the same} fixed tree for all problems and instances. Though this yields a simple algorithmic condition, it can be wasteful in terms of optimality (approximation guarantee)  as we show herein. 
In fact, for our results, we will use CMON \emph{two ways} to characterize both algorithmic properties (having fixed an implementation tree) and implementation properties (having fixed the approximation guarantee we want to achieve). 
%The construction based on DA auctions \citep{Milgrom2014DeferredacceptanceAA}, on the contrary, \emph{fix} the implementation tree a priori, which makes  to make OSP incentives ``easy'' and concentrate on the definition of a suitable algorithm; this choice is rather wasteful in terms of approximation guarantee as we show herein. 

%allows to take the payments out of the OSP equation, and to focus solely on the algorithmic aspects related to the design of OSP mechanisms.

Armed with the OSP CMON technique, we are able to give \emph{tight} bounds on the approximation ratio of OSP mechanisms for scheduling related machines (for identical jobs) and a \emph{characterization} of optimal OSP mechanisms for set system problems (which include path auctions as a special case). We study mechanisms for two- and three-value domains, as we prove that these are the only cases in which non-negative two-cycles are necessary and sufficient. 

%Whereas 
We show that the optimum for machine scheduling can be implemented OSP-ly when the agents' domains have size two. We prove that given a ``balanced'' optimum (i.e., a greedy allocation of jobs to machines) we can always find an implementation tree for which OSP is guaranteed. The mechanism directly asks the queried agents to reveal their type; given that the domain only contains two values, this is basically a descending/ascending auction. For domains of size three, instead, we give a lower bound of $\sqrt{n}$ and an essentially tight upper bound of $\lceil \sqrt{n} \rceil$. Interestingly, the upper bound is proved with two different OSP mechanisms -- one assuming more than  $\lceil \sqrt{n} \rceil^2$ number of jobs and the second under the hypothesis that there are less than that. On the technical level, these results are shown by using our approach of CMON two ways. We prove that any better than  $\sqrt{n}$-approximate OSP mechanism must have the following structure: for a number of rounds, the mechanism must (i) separate,  in its implementation tree, largest and second largest value in the domain; (ii) assign nothing to agents who have maximum value in the domain. The former property restricts the family of implementation trees we can use, whilst the latter restricts the algorithmic output. Our lower bound shows that there is nothing in this intersection. Our matching upper bounds need to find \emph{both} implementation tree and algorithm satisfying OSP and approximation guarantee. 

\bigskip
\centerline{
\begin{minipage}[c]{0.8\textwidth}
	{\bf Main Theorem on Scheduling (informal).} {\em The tight approximation guarantee of OSP mechanisms is $\sqrt{n}$. The OSP mechanisms use a descending auction (to find the $\lceil \sqrt{n} \rceil$ slowest machines) followed by an ascending auction (to find the fastest machine(s)).}
\end{minipage}
}
\bigskip
\noindent While the general idea of the implementation is that of a descending auction followed by an ascending auction independently of the number of jobs, we need to tailor the design of the mechanisms (namely, their ascending phase) according to the number of jobs to achieve OSP and desired approximation simultaneously. This proves two important points. On one hand, the design of OSP mechanisms is challenging yet interesting as one needs to carefully balance algorithms and their implementation. On the other hand, it proves why fixing the implementation, as in DA auctions, might be the wrong choice. We in fact extend and adapt our analysis to prove that any ascending and descending (thus including DA) auction has an approximation of $n$ (cf. Appendix \ref{apx:asc_desc_sched}). 

For set systems, we study under what conditions it is possible to design optimal OSP mechanisms. We begin by considering a restricted version of set systems, shortest path on graphs comprised of two parallel paths, whose edges have two-value domains. We show how the topology (i.e., number of edges of either path) and values in the domains can change the OSP-implementability of optimal algorithms. Specifically, we show that our observation for shortest path on the graph in Figure \ref{graph21} is not an accident as for all the graphs where a path is a direct edge, we can design an optimum OSP mechanism no matter what the alternative path looks like. Similarly, we prove that there are no optimal OSP mechanisms for the graph in Figure \ref{graph22} and all the graphs where the two paths are composed of the same number (larger than one) of edges. As for the graphs where neither path is direct and each has a different number of edges, the existence of an OSP mechanism returning the shortest path depends on the values in the domains. Given the simple setting considered for shortest paths, the mechanisms that we design are not very complex and it actually turns out that any implementation tree would lead to optimal OSP mechanisms, when these are feasible. We then generalize the setting to any set system problem, wherein agents have three-value domains and fully characterize the properties needed to design OSP mechanisms. 

\bigskip
\centerline{
	\begin{minipage}[c]{0.8\textwidth}
		{\bf Main Theorem on Set Systems (informal).} {\em There is an optimal OSP mechanism if and only if the set of feasible solutions are ``aligned'' for every agent's subdomain. The optimal OSP mechanism, if any, combines ascending and descending auctions depending on the structure of the feasible solution set.}
	\end{minipage}
}
\bigskip
\noindent 
The intuition behind the characterization is simple. From OSP CMON, we know that if an OSP mechanism selects an agent $e$ when she has a ``high'' cost then it must select $e$ when she has a ``low'' cost (akin to monotonicity for strategyproofness). Therefore, to design an optimal OSP mechanism we need to define an implementation tree which satisfies this property. 
%At the node of the tree in which $e$ is asked to separate a high cost from a low cost, the domain of the other agents is restricted to a particular subdomain of the agents' domain (depending on the particular history). We then need the feasible solutions to be ``aligned'' for this subdomain, which roughly means that it should never be the case that an optimal solution, including $e$ with high cost and the costs of the other agents in the subdomain, becomes non-optimal when $e$ has a low cost and the costs of the others are chosen from the subdomain. 
At each node of the tree, the domain of the agents is restricted to a particular subdomain of the agents' domain, depending  on the particular history. When $e$ is asked to separate a high cost from a low cost at node $u$ of the tree, we then need the feasible solutions to be ``aligned'' for the subdomain at $u$, which roughly means that it should never be the case that there are two bid profiles in this subdomain for which $e$ belongs to an optimal solution when she has a high cost and is not part of an optimal solution when she has a low cost.  
The somehow surprising extra aspect is that even if the feasible solutions were not aligned for one subdomain then there would be no way to design an implementation tree to bypass this misalignment, since, roughly speaking, any implementation tree for an optimal algorithm must explore all subdomains to find an optimal solution. The technical definition of alignment has some nuisance to do with the particular ways in which the OSP monotonicity can be broken, but on the positive side, rather immediately suggests how to cleverly interleave ascending and descending phases to design an optimal OSP mechanism. 
%This characterization precisely shows how OSP needs to look at the quality of solutions among set of instances (encoded by agent subdomains) rather than just the single instance and how this is needed to inform the shape of the implementation tree.

%To facilitate the use of our results for set systems, we then give a \emph{testing algorithm}, running in time polynomial in the size of the set system, which flags whether an optimum OSP mechanism for the set system problem at hand is possible or not. This coupled with our mechanism gives a sort of automated mechanism design result in that the designer has a blackbox, comprised of testing algorithm, and possibly our mechanism, to implement the optimal solution in an OSP way. 

%For set systems,
%%\dio{Removed reference to tie-breaking rule here}
%% we begin by observing that the tie-breaking rule is fundamental to ensure OSP (this holds true also for the mechanisms given by \citet{li2015obviously}) and then
%we fully characterize under which conditions it is possible to design a $\rho$-approximate OSP mechanism. This is done by again characterizing which properties the OSP implementation tree of a $\rho$-approximate algorithm must satisfy and show, very generally, how to guarantee those no matter the set system and the approximation ratio of interest. We apply this result to characterize when an optimum OSP mechanism exists for path auction on simple graphs.    

We remark that our mechanisms are, to the best of our knowledge,
the first examples of OSP mechanisms with money that do not follow a clock or a posted price auction format
(other mechanisms that do not follow these formats have been proposed only for setting without money,
namely matching and voting \citep{li2015obviously,ashlagi2015no,BadeG17,pycia2016obvious}). One of the main messages of our work is exactly that it is possible to combine ascending and descending phases for the implementation trees of algorithms with good approximation guarantees and obtain OSP mechanisms.

\subsection{Related work}

%\smallskip \noindent {\bf Related work.}
The notion of OSP mechanism has been introduced recently by \citep{li2015obviously} and has received a lot of attention in the community. 
As mentioned above, the class of deferred-acceptance (DA) auctions \citep{Milgrom2014DeferredacceptanceAA} yields OSP mechanisms since every DA auction can be implemented as a (suitable) ascending price auction. One of the main advantages of DA auctions is that the construction boils down to the problem of defining a suitable \emph{scoring function} for the bidders \citep{Milgrom2014DeferredacceptanceAA}. \cite{DGR17} studied the approximability of DA auctions for several optimization problems, and showed that in some cases DA auctions must have an approximation guarantee significantly worse than the best strategyproof mechanism; \citep{DGR17,Kim15} provide a number of positive results where DA auctions are instead optimal. 

Several works have focused on understanding better the notion of OSP mechanism, and to apply it mainly to settings without money, namely matching and voting.
In particular \citet{ashlagi2015no,BadeG17,mackenzie2017revelation} mainly attempt to simplify the notion,
whilst \citet{pycia2016obvious,zhang2017partition} define, among other results, stronger and weaker versions of OSP, respectively. A recent work related to ours is \citep{FerraioliV17} where the authors consider OSP mechanisms with and without money for machine scheduling
and facility location. Their lower bound for machine scheduling is \emph{constant} and uses a particular definition of payments,
while here we prove a $\sqrt{n}$ lower bound that follows from the CMON characterization of OSP; %that does not need any assumption on payments;
their upper bounds instead use monitoring, a model wherein agents pay their reported costs whenever they overbid (instead of their actual cost). 

Research  in algorithmic mechanism design \citep{sandholm2003sequences,hartline2009simple,chawla2010multi}
has  suggested to focus on ``simple'' mechanisms to deal with bounded rationality.
For example, posted-price mechanisms received huge attention very recently and have been applied to many different settings:
from revenue maximization \citep{babaioff2014simple} to makespan minimization \citep{feldman2017makespan},
for buyers with correlated valuations \citep{adamczyk2015sequential}, or with complements \citep{eden2017simple}, or for streams of buyers \citep{correa2017posted}.
In these mechanisms one's own bid is immaterial for the price paid to get some goods of interest
-- this should immediately suggest that trying to play the mechanism is worthless no matter the cognitive abilities of the agents.
However, posted price mechanisms do not fully capture the concept of simple mechanisms:
e.g., ascending price auctions are not posted price mechanisms and still turn out to be ``simple''.
An orthogonal approach is that of verifiably truthful mechanisms \citep{BP15},
wherein agents can run some algorithm to effectively check that the mechanism is incentive compatible.
Nevertheless, these verification algorithms can run for long (i.e., time exponential in the input size) and are so far known only for quite limited scenarios.
Importantly, moreover, they seem to transfer the question from the mechanism itself to the verification algorithm.
Another line of work \citep{hartline2009simple,Eden} focused on extending to more complex settings the well-known result of \citet{bulow1996auctions},
that states that the revenue-maximizing single item auction, i.e., the Myerson auction,
performs worse than the ``simpler'' social welfare maximizing auction, i.e., the Vickrey auction, with one extra bidder.
However, simplicity in this research is vaguely stated, since the focus is more on the complexity of designing a mechanism,
than on the complexity of understanding it.

Another paper relevant in the context of ``simple'' mechanisms is  \citep{GR96}, which shares some of the motivation and
philosophy behind OSP. Specifically, the authors notice that the implementation of a
social choice function via a normal-form game is
simpler/more obvious if there exists a ``simple guide''
that can always be used to discover the strategy
combination which survives successive elimination of
dominated strategies. They showed that backwards induction
implementation via an extensive game is equivalent to
such a ``guided'' implementation via a normal-form game.

CMON is a widely used technique in mechanism design that dates back to \citet{Roc85} -- a general treatment is given in \citep{vohraNet,GuiMulVoh04}.
This method has been used quite extensively to prove truthfulness of mechanisms in the classical setting, cf., e.g., \citep{NisMulti,LavSwa09}
and when some form of verification can be adopted, see \citep{Ven14,KV15}.
Particularly relevant for our work is the research which shows that in order to establish truthfulness it is sufficient to study cycles of length $2$ as in \citep{SakYu05}.

\section{Preliminaries}
A mechanism design setting is defined by a set of $n$ \emph{selfish agents} and a set of allowed \emph{outcomes} $\out$.
Each agent $i$ has a \emph{type} $t_i \in D_i$, where $D_i$ is called the \emph{domain} of $i$.
The type $t_i$ is usually assumed to be \emph{private knowledge} of agent $i$.
We will let $t_i(X) \in \real$ denote the \emph{cost} (opposite of the valuation) of agent $i$ with type $t_i$ for the outcome $X \in \out$. In our applications, we will assume that costs are non-negative; however, our OSP CMON framework and characterizations hold in general no matter the sign (and thus also for valuations).

%Observe, moreover, that we are not making any assumption on the sign of the cost of the agents (i.e., we might have valuations as negative costs, although in our applications agents will have non-negative costs).

A \emph{mechanism} is a process for selecting an outcome $X \in \out$.
To this aim, the mechanism interacts with agents.
Specifically, agent $i$ is observed to take \emph{actions} (e.g., saying yes/no)
that may depend on her presumed type $b_i \in D_i$
(e.g., saying yes could ``signal'' that the presumed type has some properties that $b_i$ alone might enjoy).
We say that agent $i$ takes \emph{actions compatible with (or according to) $b_i$} to stress this.
We highlight that the presumed type $b_i$ can be different from the real type $t_i$.

For a mechanism $\M$, we let $\M(\b)$ denote the outcome returned by the mechanism
when agents take actions according to their presumed types $\b = (b_1, \ldots, b_n)$.
In our context, this outcome is given by a pair $(f,\p)$,
where $f = f(\b)$ (termed \emph{social choice function} or, simply,  algorithm) maps the actions taken by the agents according to $\b$ to a feasible solution in $\out$,
and $\p=\p(\b)=(p_1(\b),\ldots,p_n(\b)) \in \real^n$ maps the actions taken by the agents according to $\b$ to \emph{payments} from the mechanism to the agents. %$i$.
%Note that the $p_i$'s can be positive (meaning that the mechanism will pay the agents) or negative (meaning that the agents will pay the mechanism).

Each selfish agent $i$ is equipped with a \emph{utility function} $u_i \colon D_i \times \out \rightarrow \real$.
For $t_i \in D_i$ and for an outcome $X = (f,p) \in \out$ returned by a mechanism $\M$,
$u_i(t_i, X)$ is the utility that agent $i$ has for outcome $X$ when her type is $t_i$.
We define utility as a quasi-linear combination of payments and costs, i.e., 
$
u_i(t_i, \M(b_i, \bi)) = p_i(b_i, \bi) - t_i(f(b_i, \bi)).
$

A mechanism $\M$ is \emph{strategy-proof} if $u_i(t_i,\M(t_i, \bi)) \geq u_i(t_i,\M(b_i, \bi))$
for every $i$, every $\bi=(b_1, \ldots, b_{i-1}, b_{i+1}, \ldots, b_n)$
and every $b_i \in D_i$, with $t_i$ being the true type of $i$.
That is, in a strategy-proof mechanism the actions taken according to the true type are dominant for each agent.

For our applications, we will be focusing on \emph{single-parameter} settings, that is, the case in which
the private information of each bidder $i$ is a single real number $t_i$ and $t_i(X)$ can be expressed as $t_i w_i(X)$ for some publicly known function $w_i$.
To simplify the notation, we will write $t_i f_i(\b)$ when we want to express the cost of a single-parameter agent $i$ of type $t_i$
for the output of social choice function $f$ on input the actions corresponding to a bid vector $\b$.

 %\subsection{Obvious Strategyproofness.}
\smallskip \noindent {\bf Obvious Strategyproofness.}
We now formally define the concept of obviously strategy-proof mechanism.
This concept has been introduced in~\citep{li2015obviously}.
However, our definition is built on the more accessible ones given by \citet{ashlagi2015no} and \citet{FerraioliV17}. As shown in \citep{BadeG17,mackenzie2017revelation}, our definition is equivalent to Li's.\footnote{In particular, we highlight that the absence of information sets and chance nodes -- i.e., concurrent moves of the agents and randomized choices of the mechanism -- is without loss of generality \citep{BadeG17}. As for the latter, it is worth noting that this is because in Li's definition, a randomized OSP mechanism turns out to be a probability distribution over deterministic OSP mechanisms; we refer the interested reader to \citep{OSP2} for a notion of OSP in expectation.}% mechanisms.}.

% Note that we focus on deterministic mechanisms only.

Let us first formally model how a mechanism works.
An \emph{extensive-form mechanism} $\M$ is defined by a directed tree $\T=(V,E)$, called the \emph{implementation tree}, such that:
\begin{itemize}[noitemsep,leftmargin=10pt]
	\item every leaf $\ell$ of the tree is labeled with a possible outcome $X(\ell) \in \out$ of the mechanism;
	\item every internal vertex $u \in V$ is labeled by an agent $S(u) \in [n]$;
	\item every edge $e=(u,v) \in E$ is labeled by a subset $T(e) \subseteq D = \times_i D_i$ of type profiles such that:
	\begin{itemize}[noitemsep,leftmargin=10pt]
		\item the subsets of profiles that label the edges outgoing from the same vertex $u$ are disjoint,
		i.e., for every triple of vertices $u, v, v'$ such that $(u,v) \in E$ and $(u,v') \in E$, we have that $T(u,v) \cap T(u,v') = \emptyset$;
		\item the union of the subsets of profiles that label the edges outgoing from a non-root vertex $u$
		is equal to the subset of profiles that label the edge going in $u$,
		i.e., $\bigcup_{v \colon (u,v) \in E} T(u,v) = T(\phi(u),u)$, where $\phi(u)$ is the parent of $u$ in $\T$; %along the path between the root and $u$;
		\item the union of the subsets of profiles that label the edges outgoing from the root vertex $r$
		is equal to the set of all profiles, i.e., $\bigcup_{v \colon (r,v) \in E} T(r,v) = D$;
		\item for every $u, v$ such that $(u,v) \in E$ and for every two profiles $\b, \b' \in T(\phi(u),u)$ such that $b_i = b'_i$, $i = S(u)$,
		if $\b$ belongs to $T(u,v)$, then $\b'$ must belong to $T(u,v)$ also.
	\end{itemize}
\end{itemize}

Roughly speaking, the tree represents the steps of the execution of the mechanism.
As long as the current visited vertex $u$ is not a leaf, the mechanism interacts with the agent $S(u)$.
Different edges outgoing from vertex $u$ are used for modeling the different actions that agents can take during this interaction with the mechanism.
In particular, each possible action is assigned to an edge outgoing from $u$.
As suggested above, the action that agent $i$ takes may depend on her presumed type $b_i \in D_i$.
That is, different presumed types may correspond to taking different actions, and thus to different edges.
The label $T(e)$ on edge $e=(u,v)$ then lists the type profiles that enable the agent $S(u)$ to take those actions that have been assigned to $e$.
In other words, when the agent takes the actions assigned to edge $e$, 
then the mechanism (and the other agents) can infer that the type profile must be contained in $T(e)$.
The constraints on the edges' label can be then explained as follows:
first we can safely assume that different actions must correspond to different type profiles
(indeed, if two different actions are enabled by the same profiles we can consider them as a single action);
second, we can safely assume that each action must correspond to at least one type profile
that has not been excluded yet by actions taken before node $u$ was visited
(otherwise, we could have excluded this type profile earlier);
third, we have that the action taken by agent $S(u)$ can only inform about her types and not about the type of the remaining agents.
The execution ends when we reach a leaf $\ell$ of the tree. In this case, the mechanism returns the outcome that labels $\ell$.

Observe that, according to the definition above, for every profile $\b$ there is only one leaf $\ell = \ell(\b)$
such that $\b$ belongs to $T(\phi(\ell),\ell)$. For this reason we say that $\M(\b) = X(\ell)$.
Moreover, for every type profile $\b$ and every node $u \in V$, we say that $\b$ is \emph{compatible} with $u$ if $\b \in T(\phi(u),u)$.
Finally, two profiles $\b$, $\b'$ are said to \emph{diverge} at vertex $u$ if there are two vertices $v, v'$
such that $(u,v) \in E$, $(u,v') \in E$ and $\b \in T(u,v)$, whereas $\b' \in T(u,v')$.
For every node $u$ in a mechanism $\M$ such that there are two profiles $\b, \b'$ that diverge at $u$, we say that $u$ is a \emph{divergent node},
and $i = S(u)$ the corresponding \emph{divergent agent}.
For each agent $i$, we define the \emph{current domain} at node $u$, denoted $D_i(u)$,
such that $D_i(r) = D_i$ for the root $r$ and $D_i(u) = \cup_{\b \in T(\phi(u),u)} b_i$.
In words, this is the set of types of $i$ that are compatible with the actions that $i$ took during the execution of the mechanism until node $u$ is reached.
Indeed, according to the definition above, at each node $u$ in which $i$ diverges,
$\M$ partitions $D_i(u)$ in $k$ subsets, where $k$ is the number of children of $u$,
and where for every child $v$ of $u$, $D_i(v) \subset D_i(u)$ contains the types of bidder $i$ compatible with the action
that she takes when interacting with the mechanism at node $u$.

We are now ready to define obvious strategyproofness.
An extensive-form mechanism $\M$ is \emph{obviously strategy-proof (OSP)} if for every agent $i$ with real type $t_i$,
for every vertex $u$ such that $i = S(u)$, for every $\bi, \bi'$ (with $\bi'$ not necessarily different from $\bi$),
and for every $b_i \in D_i$, with $b_i \neq t_i$,
such that $(t_i, \bi)$ and $(b_i, \bi')$ are compatible with $u$, but diverge at $u$,
it holds that $u_i(t_i, \M(t_i, \bi)) \geq u_i(t_i,\M(b_i, \bi'))$. 
Roughly speaking, an obvious strategy-proof mechanism requires that, at each time step
agent $i$ is asked to take a decision that depends on her type, the worst utility that
she can get if at this time step she behaves according to her true type
is at least the best utility achievable by behaving as she had a different type. 
Hence, if a mechanism is obviously strategy-proof, then it is also strategy-proof.
Indeed, the latter requires that truthful behavior is a dominant strategy when agents know the entire type profile,
whereas the former requires that it continues to be a dominant strategy
even if agents have only a partial knowledge of profiles
limited to what they observed in the mechanism up to the time they are called to take their choices.

We say that an extensive-form mechanism is \emph{trivial} if for every vertex $u \in V$ and for every two type profiles $\b,\b'$,
it holds that $\b$ and $\b'$ do \emph{not} diverge at $u$.
That is, a mechanism is trivial if it never requires that agents take actions that depend on their type.
If a mechanism is not trivial, then there is at least one divergent node.
On the other hand, every execution of a mechanism (i.e., every path from the root to a leaf in the mechanism implementation tree)
may go through at most $\sum_i (|D_i|-1)$ divergent nodes,
the upper bound being the case in which at each divergent node $u$, the agent $i = S(u)$ separates $D_i(u)$,
in $D_i(u) \setminus \{b\}$ and $\{b\}$ for some $b \in D_i(u)$.

% Let us state two further properties of obvious strategyproofness, that turn out to be very useful in the rest of the paper.
% First, it is not hard to see that if $\M$ is OSP when the type profile is taken from $D=\times_i D_i$,
% then it continues to be OSP even if the types are only allowed to be selected
% from $D' = D'_1 \times \cdots \times D'_n$, where $D'_i \subseteq D_i$.
% Moreover, let us define $\M'$ obtained from $\M$ by \emph{pruning} the paths involving actions corresponding to types in $D \setminus D'$.
% If $\M$ is OSP, then also $\M'$ enjoys this property~\citep{li2015obviously}.

 %\subsection{Machine scheduling}
\smallskip \noindent {\bf Machine scheduling.} Here, we are given a set of $m$ identical jobs to execute and the $n$ agents control related machines.
That is, agent $i$ has a job-independent processing time $t_i$ per unit of job
(equivalently, an execution speed $1/t_i$ that is independent from the actual jobs).
The social choice function $f$ must choose a possible schedule $f(\b) = (f_1(\b), \ldots, f_n(\b))$ of jobs to the machines,
where $f_i(\b)$ denotes the job load assigned to machine $i$ when agents take actions according to $\b$.
The cost that agent $i$ faces for the schedule $f(\b)$ is $t_i(f(\b))=t_i \cdot f_i(\b)$. We focus on social choice functions $f^*$ minimizing the \emph{makespan}, i.e.,
$
f^*(\b) \in \arg\min_{\x} \max_{i=1}^n b_i(\x).
$
We say that $f$ is $\alpha$-approximate if it returns a solution whose cost is a factor $\alpha$ away from the optimum.

%\subsection{Set systems}
\smallskip \noindent {\bf Set systems.}
In a \emph{set system} $(E, {\mathcal{F}})$ we are given a set $E$ of elements and a family $\mathcal{F} \subseteq 2^E$ of feasible subsets of $E$.
Each element $i \in E$ is controlled by a selfish agent, that is, the cost for using $i$ is known only to agent $i$ and is equal to some non-negative value $t_i$.
The social choice function $f$ must choose a feasible subset in $\mathcal{F}$; we can use the same notation used above for machine scheduling with the restriction that $f_i(\b) \in \{0,1\}$ to mean that the element controlled by agent $i$ is either chosen by $f$, with $f_i(\b)=1$, or not, with $f_i(\b)=0$.
Our objective is here social cost minimization, that is, 
$
f^*(\b) \in \arg\min_{\x} \sum_{i=1}^n b_i(\x).
$
%As above,  $f$ is $\alpha$-approximate if it returns a solution whose cost is a factor $\alpha$ away from the optimum. 
Several problems on graphs, such as the path auction discussed above, can be cast in this framework.

\section{Cycle-monotonicity for OSP mechanisms}\label{sec:cyclemon}
%In this section w
We  now show how to generalize the cycle-monotonicity technique to design OSP mechanisms. 

Let us fix an extensive-form mechanism $\M$ and the corresponding implementation tree $\T$.
For every $i$ and $a,b \in D_i$, we let $u_{a,b}^i$ %(simply $u$, when the rest is understood by the context) 
be a vertex $u$ in $\T$ such that $i = S(u)$ and $(a, \bi)$ and $(b, \bi')$ are compatible with $u$, but diverge at $u$, for some $\bi, \bi' \in D_{-i}(u) = \times_{j \neq i} D_j(u)$. Note that such a vertex might not be unique as player $i$ will be asked to separate $a$ from $b$ in different paths from the root to a leaf (but only once for every such path).
We call these vertices of $\T$ $ab$-separating for player $i$.

The algorithmic characterization of OSP we provide herein is based on the following 
% easy 
observation.
\begin{obs} \label{obs:def}
    A mechanism that uses a payment function $\p$ and an implementation tree $\T$ is an OSP mechanism for the social choice function $f$
%     is part of an OSP mechanism implemented with tree $\T$ 
    if and only if for all $i$, for all $a, b \in D_i$, $a \neq b$, for all vertices $u_{a,b}^i$ that are $ab$-separating for $i$, and $\ai, \bi \in D_{-i}(u_{a,b}^i)$, 
    \begin{equation}
     \label{eq:osp_constraint}
     p_i(\b) - p_i(\a) \leq a(f(b, \bi))-a(f(a, \ai)).
    \end{equation}
\end{obs}

\begin{definition}[\vgraph]\label{def:graph:verification:all}
	Let $f$ be a social choice function and $\T$ be an implementation tree.
	For every $i$, the {\em \vgraph} $\ver$ has a node for each type profile in $D$,
	and an edge $((a,\bi), (b, \bi'))$ for every $a, b \in D_i$, $a \neq b$, and $\bi,\bi' \in D_{-i}(u)$, where $u$ is an $ab$-separating vertex of $\mathcal{T}$.
	The weight of the edge is $a(f(b, \bi))-a(f(a, \bi'))$.
\end{definition}

\begin{theorem}\label{thm:cmon}
    A mechanism with implementation tree $\T$ is an OSP mechanism for a social function $f$ on finite domains
% 	A social choice function $f$ is part of an OSP mechanism with implementation tree $\mathcal{T}$ for finite domains
	if and only if, for all $i$, the graph $\ver$ does not have negative weight cycles.
\end{theorem}

The proof of the theorem follow standard arguments used %to prove cycle monotonicity for 
for the classical definition of strategyproofness
(a full proof is given in appendix for %sake of 
completeness).
We call the cycle monotonicity property of Theorem \ref{thm:cmon} OSP CMON; when this holds true for cycles of length 2 only then we talk about OSP two-cycle monotonicity. It is useful, for our applications, to recast the OSP two-cycle monotonicity between profiles $\a$ and $\b$ for a single-parameter agent $i$ as follows:
\[
-a_i f_i(\a) + a_i f_i(\b) - b_i f_i(\b) + b_i f_i(\a) \geq 0 \Longleftrightarrow (a_i-b_i) (f_i(\b)-f_i(\a)) \geq 0.
\]

\subsection{Warm-up I: using OSP CMON to bound approximation guarantee}
In this section and in the next one, we show some preliminary results about machine scheduling and set systems. These results, despite (or, rather, because of) their simplicity, show how  powerful a tool OSP CMON can be, when one wants to answer algorithmic questions,
such as approximability, about OSP.

Indeed, we give here a simple lower bound %achieved through OSP CMON 
for the machine scheduling problem.
\begin{proposition}\label{prop:scheduling:three-values:lb}
    For the machine scheduling problem, no OSP mechanism can be better than $2$-approximate, even for two jobs and two agents with three-value domains $D_i=\{L,M,H\}$, where $L<M<H$, with $M>2L$ and $H > 2M$.
%     \begin{align}\label{eq:3valdom-hard}
%         M> 2L && \text{ and } && H > 2M\ .
%     \end{align}
%     Thus,  only trivial $2$-approximations are possible  for two jobs in these domains.
\end{proposition}
\begin{proof}
    Assume by contradiction that there is an OSP mechanism $\M$
    that is better than $2$-approximate,
    and let $\T$ be its implementation tree.
    Since $M>2L$ and $H > 2M$, every trivial OSP mechanism
    must have approximation guarantee at least $2$.
    Hence $\M$ must be non trivial.
    Let $i$ be the first divergent agent of $\M$ implemented with $\mathcal T$, and let $u$ be the node where this agent diverges (such an agent exists because the mechanism is not trivial).
           
    If $i$ diverges at $u$ on $M$ and $H$, then
    consider $\x=(x_i,x_{-i})=(H,H)$ and $\y=(y_i,y_{-i})=(M,L)$.
    Since the mechanism is better than $2$-approximate,
    it must satisfy $f_i(\x)=1$ and $f_i(\y)=0$.
    Then, the cycle $(\y, \x, \y)$ in $\ver$
    has cost $M(1 - 0) + H(0 - 1) = M - H < 0$, thus a contradiction.
           
    If $i$ diverges at $u$ on $L$ and $M$, then
    consider $\x=(x_i,x_{-i})=(L,L)$ and $\y=(y_i,y_{-i})=(M,H)$.
    Since the mechanism is better than $2$-approximate,
    it must satisfy $f_i(\x)=1$ and $f_i(\y)=2$.
    Then, the cycle $(\x, \y, \x)$ in $\ver$
    has cost $L(2 - 1) + M(1 - 2) = L - M < 0$, thus a contradiction.
\end{proof}

In Appendix~\ref{apx:scheduling} we give a complete characterization of the approximability for two jobs and two agents in the three-value domains, showing that in this specific case the above bound is tight.
However, note that for this bound we require the domain to have at least three different values;
we will in fact prove in Section~\ref{sec:scheduling}
that we can design an optimal OSP mechanism for scheduling related machines
when $D_i = \{L_i, H_i\}$ for every $i$.
We will also show how to use a more involved argument to prove a substantially higher (and tight) bound of $\sqrt{n}$.

\subsection{Warm-up II: using OSP CMON to characterize OSP mechanisms} \label{sec:warmup2}

While the above example shows how OSP CMON can be used for giving approximation bounds,
we now give an %toy 
example that shows how OSP CMON helps to characterize
under which conditions a given mechanism is OSP. %turns out to be OSP or not.

Consider the path auction problem discussed in the introduction;
this is a special case of a set system problem
where the set of feasible solutions is the set of all the paths
between the source node $s$ and the destination node $t$ in a given graph $G$. Consider the case in which $G$ has two parallel paths from the source to the destination; the first is comprised of a set $T$ of $t$ edges, that we will sometimes call top edges, whilst the second is comprised of a set $B$ of $b$ edges, that we will call bottom edges. W.l.o.g., we assume that $t \geq b$. %We let $f_{e}(\b) \in \{0,1\}$ denote a binary variable indicating whether edge $e$ in either path is selected in the solution output by the mechanism on input (the actions according to) $\b$.  %Observe that the case in which $s=t=1$ is covered above for parallel links. 

\newcommand{\be}{\vett{b}_{-e}}

\begin{proposition}
	There is an optimal OSP mechanism for the shortest path problem on parallel paths and two-value domains $D=\{L,H\}^n$ if and only if either (1)~$b=1$ or (2)~$t > b > 1$ and $\frac{H}{L} \leq \frac{t-1}{b-1}$.
\end{proposition}
\begin{proof}
	Let us start by proving the sufficient condition; the mechanism shall return the bottom path in case of ties. The sufficiency is even stronger in that we will prove it no matter the implementation tree (thus including direct-revelation mechanisms). (Note that it will be enough to prove OSP $2$-cycle monotonicity since, as prove below in Theorem \ref{thm:2cycle3}, this is sufficient for two-value domains in this setting.)

	Consider first the case $t \geq b=1$. For a top edge $e$, we have $f_e(H, \be') =0$ no matter $\be'$, thus implying that OSP $2$-cycle monotonicity is satisfied since
	$
	(H-L)(f_e(L, \be)-f_e(H, \be')) \geq 0.
	$ 
	For a bottom edge $e$, instead, we can observe that $f_e(L, \be)=1$ for all $\be$ and therefore the OSP $2$-cycle monotonicity is satisfied again. %The theorem then follows in both cases from Theorem \ref{thm:2cycle3}. 
	
	Let us now consider  the case $t > b > 1$ and $\frac{H}{L} \leq \frac{t-1}{b-1}$. 
	%For $\tau,\beta\geq 0$, let a $\tau,\beta$-agent be an agent queried by the mechanism after $\tau$ top and $\beta$ bottom agents have been queried. 
	%With a slight abuse of notation, we let $\tau$ and $\beta$ also denote the set of agents already queried. 
	By contradiction, assume that the mechanism is not OSP for some agent $e$. Consider  the case that $e$ is a top edge (the proof is very similar for the case in which $e$ is a bottom edge). 
	%Assume by contradiction that the mechanism with implementation tree querying $e$ after a subset $\tau$ of top agents and $\beta$ bottom agents have been queried is not OSP for $e$. 
	Since the mechanism is not OSP for $e$, and we have two-value domains, Theorem~\ref{thm:cmon}  says that the OSP-graph for $e$ has a negative $2$-cycle. 
	%We distinguish two cases. 
	Let $\tau$ and $\beta$ denote the subset of top and bottom agents that have been queried before $e$ in the implementation tree. Since the OSP-graph for $e$ has a negative $2$-cycle, then for some $\b_{\tau}$ and $\b_{\beta}$ there exist $\b_{T'}$ and $\b_{B'}$, $T'=T\setminus (\tau \cup \{e\})$ and $B'=B \setminus \beta$ such that $f_e(H, (\b_{\tau},\b_{T'}, \b_{\beta}, \b_{B'}))=1$. Since ties are broken in favor of the bottom path, this means that
	$
	H+(t-1)L \leq H + \b_{\tau}+\b_{T'} < \b_{\beta} + \b_{B'} \leq bH,
	$ 
	thus implying  $\frac{H}{L} > \frac{t-1}{b-1}$, a contradiction.
	
	%Consider now the case that $e$ is a bottom edge. Since the OSP-graph for $e$ has a negative $2$-cycle, then for some $\b_{\tau}$ and $\b_{\beta}$ there exist $\b_{T'}$ and $\b_{B'}$, $T'=T\setminus \tau$ and $B'=B \setminus (\beta  \cup \{e\})$ such that $f_e(L, (\b_{\tau},\b_{T'}, \b_{\beta}, \b_{B'}))=0$. By the tie-breaking rule of the optimal mechanism, this means that
	%\[
	%L+(b-1)H \geq L + \b_{\beta} + \b_{B'} > \b_{\tau} + \b_{T'} \geq  tL,
	%\]
	%implying that $\frac{H}{L} > \frac{t-1}{s-1}$, a contradiction.
	
	We next prove the necessity and show that when either (1) $t=b>1$ or (2) $t > b > 1$ and $\frac{H}{L} > \frac{t-1}{b-1}$,  no optimal mechanism $\M$ can be OSP. Since $\M$ is optimal, it is not trivial and at some point it must separate $L$ from $H$ for at least one agent. We consider the first divergent agent $e$, and show a negative $2$-cycle in the OSP-graph for this agent $e$, thus implying that mechanism $\M$  is \emph{not} OSP (Theorem~\ref{thm:cmon}). Note that, since $e$ is the first divergent agent, in the implementation tree $\T$, there is a unique $LH$-separating node $u$ for agent $e$; moreover, $\be, \be'\in D_{-e}(u)$ for every $\b, \b' \in \{L,H\}^n$. Assume that $e$ is a top edge (the proof is exactly the same for the case in which $e$ is a bottom edge). For the case in which $t=b>1$, it is enough to consider $\be$ comprised of at least one $H$ for the top edges and all $L$ for all the bottom edges, so that $f_e(L, \be)=0$; $\be'$ is, instead, a bid vector where all the top edges have value $L$ and at least two of the bottom edges have bid $H$ so that $f_e(H, \be')=1$. Therefore, the $2$-cycle between $(L, \be)$ and $(H, \be')$ is negative. When $t > b > 1$ and $\frac{H}{L} > \frac{t-1}{b-1}$, consider instead the vectors $\be, \be'$ where $b_{e'}=H$,  $e' \in T \setminus \{e\}$, and $b_f=L$, $f \in B$, and $b'_{e'}=L$,  $e' \in T \setminus \{e\}$, and $b'_f=H$, $f \in B$. Since $L+(t-1)H > bL$ then $f_{e}(L,\be)=0$, while as $H+(t-1)L< bH$ then $f_{e}(H,\be')=1$. Therefore, the $2$-cycle between $(L, \be)$ and $(H, \be')$ is negative also in this case. 
\end{proof}

We shall extend this analysis to general set system problems,
and we shall characterize exactly  when optimal algorithms can be implemented in an OSP way.
 
\subsection{Two-cycles are sufficient for single-parameter domains of size at most three}

Two-cycle monotonicity is a property easier to work with than CMON.
We will now show that, for single parameter settings, these properties turns out to be equivalent if and only if 
$D_i = \{L_i, M_i, H_i\}$ for each $i$, with $L_i \leq M_i \leq H_i$. (For the full proofs see Appendix~\ref{app:CMON-characterizes-OSP}.)
\begin{theorem}
	\label{thm:2cycle3}
	Consider a single-parameter setting where $|D_i| \leq 3$ for each agent $i$.
	A mechanism with implementation tree $\T$ and social choice function $f$ is OSP 
	iff OSP two-cycle monotonicity holds.
\end{theorem}
\begin{proof}[Proof Sketch]
	One direction follows from Theorem~\ref{thm:cmon}. As for the other direction,
	we prove that OSP two-cycle monotonicity implies OSP CMON.
	
	Fix agent $i$, and consider a cycle $\mathcal{C} = (\x, \x', \x'', \ldots, \x)$ in the graph $\ver$.
	%  Since $\mathcal{C}$ is a cycle, we can assume without loss of generality that $x_i = L_i$
	%  if a profile exists in the cycle such that the type of $i$ is $L_i$, and $x_i = M_i$ otherwise.
	Observe that
	%  , since no edge exists among two profiles $\x$ and $\y$ such that $x_i = y_i$,
	%  then it must be the case that every cycle 
	$\mathcal{C}$ can be partitioned in paths $(\mathcal{P}_1, \ldots, \mathcal{P}_t)$ such that
	for every $j = 1, \ldots, t$, $\mathcal{P}_j$ is as one of the following:
	\begin{enumerate}[noitemsep,leftmargin=20pt]
		\item $\mathcal{P}_j = (\x^{j-1},\y,\x^{j})$, where $x^{j-1}_i = x^{j}_i < y_i$;
		\item $\mathcal{P}_j = (\x^{j-1},\y,\z,\x^{j})$, where $x^{j-1}_i = x^{j}_i = L_i$, but $y_i \notin \{x^j_i, z_i\}$ and $z_i \notin \{x^j_i, y_i\}$;
		\item $\mathcal{P}_j = (\x^{j-1}, \mathcal{P}_j^1, \ldots, \mathcal{P}_j^s, \x^{j})$, where $s \geq 1$,
		$x^{j-1}_i = x^{j}_i = L_i$ and $\mathcal{P}_j^h = (\y^{j,h-1}, \z^{j,h}, \y^{j,h})$, where either $y^{j,h-1}_i = y^{j,h}_i = M_i$ and $z^{j,h}_i = H_i$ for every $h \in \{1, \ldots, s\}$
		or $y^{j,h-1}_i = y^{j,h}_i = H_i$ and $z^{j,h}_i = M_i$ for every $h \in \{1, \ldots, s\}$;
		\item $\mathcal{P}_j = (\x^{j-1}, \mathcal{P}_j^1, \ldots, \mathcal{P}_j^s, \w^{j}, \x^{j})$, where $s \geq 1$,
		$x^{j-1}_i = x^{j}_i = L_i$, $w^{j}_i \neq L_i$ and, for every $h \in \{1, \ldots, s\}$, $\mathcal{P}_j^h = (\y^{j,h-1}, \z^{j,h}, \y^{j,h})$, where $y^{j,h-1}_i = y^{j,h}_i \notin \{L_i, w^{j}_i\}$ and $z^{j,h}_i = w^{j}_i$.
	\end{enumerate}
	%  Indeed, if no profile appears in the cycle such that the type of $i$ is $L_i$, then the only way of going from a profile $\x^{j-1}$ with $x^{j-1}_i = M_i$ to a profile $\x^j$ with $x^j_i = M_i$ is through a profile $\y$ such that $y_i = H_i$, and this is considered in case~\ref{item:len2}.
	%  As for cycles in which there is a profile such that the type of $i$ is $L_i$, then either we have a path of length 2 as considered in case~\ref{item:len2},
	%  or a path of length 3 as considered in case ~\ref{item:len3},
	%  or a path in which the internal profiles are such that the type of $i$ is alternatively $M_i$ and $H_i$.
	%  In turn, for this last case can can be distinguished two subcases:
	%  either the type of $i$ in the last internal profile is the same as in the first internal profile, as considered in case~\ref{item:leneven},
	%  or they are different, as considered instead in case~\ref{item:lenodd}.
	Since we are focusing on single-parameter settings, the cost of a path of
	%  type~\ref{item:len2} is
	%  $
	%   C(\mathcal{P}_j) = \Big(x^{j-1}_i \cdot f_i(\x^{j-1}) - x^{j-1}_i \cdot f_i(\y)\Big) + \Big(y_i \cdot f_i(\y) - y_i \cdot f_i(\x^j)\Big).
	%  $
	% %  for some $f_i \geq 0$.
	%  This cost can be written as follows
	%  \begin{align*}
	%   C(\mathcal{P}_j) & = y_i (f_i(\y) - f_i(\x^{j-1}) + f_i(\x^{j-1}) - f_i(\x^j)) - x^{j-1}_i (f_i(\y) - f_i(\x^{j-1}))\\
	%   & = (y_i-x^{j-1}_i) (f_i(\y) - f_i(\x^{j-1})) + y_i (f_i(\x^{j-1}) - f_i(\x^j))\\
	%   & \geq M_i (f_i(\x^{j-1}) - f_i(\x^j)),
	%  \end{align*}
	%  where we used that $y_i > x^{j-1}_i \geq L_i$ and that $f_i(\y) \geq f_i(\x^{j-1})$ by two-cycle monotonicity.
	% 
	%  Consider now paths of 
	type~\ref{item:len3} 
	% %  Since we are focusing on single-parameter settings, 
	%  The cost of this path 
	is
	$
	C(\mathcal{P}_j) = \Big(x^{j-1}_i \cdot f_i(\x^{j-1}) - x^{j-1}_i \cdot f_i(\y)\Big) + \Big(y_i \cdot f_i(\y) - y_i \cdot f_i(\z)\Big) + \Big(z_i \cdot f_i(\z) - z_i \cdot f_i(\x^j)\Big).
	$
	We can rewrite this as %cost as follows:
	{\allowdisplaybreaks
		\begin{align*}
		C(\mathcal{P}_j) & = - x^{j-1}_i(f_i(\y)-f_i(\x^{j-1})) + y_i(f_i(\y)-f_i(\z))\\
		& \qquad + z_i(f_i(\z) - f_i(\y) + f_i(\y) - f_i(\x^{j-1}) + f_i(\x^{j-1}) - f_i(\x^j))\\
		& \geq (z_i - x^{j-1}_i)(f_i(\y)-f_i(\x^{j-1})) + (z_i - y_i)(f_i(\z)-f_i(\y)) + z_i(f_i(\x^{j-1}) - f_i(\x^j))\\
		& \geq M_i (f_i(\x^{j-1}) - f_i(\x^j)),
		\end{align*}
	}
	where we used that $z_i \geq M_i > L_i = x^{j-1}_i$ and, by two-cycle monotonicity, $f_i(\y) \geq f_i(\x^{j-1})$ (since $y_i > x^{j-1}_i$),
	and either $f_i(\z) - f_i(\y) = 0$ or $\mathtt{sign}(f_i(\z) - f_i(\y)) = \mathtt{sign}(z_i - y_i)$.

	With similar arguments, we can prove that $C(\mathcal{P}_j) \geq M_i (f_i(\x^{j-1}) - f_i(\x^j))$ even for the other kinds of path.
	%  Finally, the cost of the cycle is the sum of the costs of paths in which this cycle has been partitioned.
	%  That is,
	Hence,
	$
	C(\mathcal{C}) = \sum_{j=1}^{t} C(\mathcal{P}_j) \geq M_i \sum_{j=1}^{t} (f_i(\x^{j-1}) - f_i(\x^j)) = M_i (f_i(\x^0) - f_i(\x^t)) = 0.
	$
	%The theorem then follows.
\end{proof}

Theorem~\ref{thm:2cycle3} is tight. Indeed, we next show that if there is at least one agent
whose type domain contains at least four values, then OSP two-cycle monotonicity does not imply OSP cycle-monotonicity (and thus OSP-ness).
Specifically, we have the following theorem.
\begin{theorem}
	\label{thm:neg_res}
	There exists a mechanism with implementation tree $\T$ such that $\ver$ is OSP two-cycle monotone for every $i$,
	but there is an agent for which the mechanism does not satisfy OSP CMON,
	whenever there is at least one agent $i$ such that $|D_i| \geq 4$.
	The claim holds even for a single-item auction setting and $D_j = D$ for every $j \neq i$.
\end{theorem}

\section{Scheduling related machines}
\label{sec:scheduling}

\newcommand{\stpath}[4]{
	\node[circle, fill=gray!50] (s) at (0,1) {$s$};
	\node[circle, fill=gray!50] (t) at (4,1) {$t$};
	\node[circle, fill=gray!50] (u) at (2,2) {\phantom{$u$} };
	\node[circle, fill=gray!50] (d) at (2,0) {\phantom{$d$}};
	\draw[thick] (s) -- (u) node[pos=.5,sloped,above] {$#1$} -- (t) node[pos=.5,sloped,above] {$#2$} -- (d) node[pos=.5,sloped,below] {$#4$} -- (s) node[pos=.5,sloped,below] {$#3$};
}

\newcommand{\upperpath}{\draw[line width=4pt] (s) -- (u) -- (t);}

\newcommand{\lowerpath}{\draw[line width=4pt] (s) -- (d) -- (t);}

\newcommand{\treeedges}[6]{
	\def\edgea{#1}
	\def\edgeb{#2}
	\def\edgec{#3}
	\def\edged{#4}
	\def\edgee{#5}
	\def\edgef{#6}
}

\newcommand{\tree}[4]{
	\node[circle, fill=gray!50] (LL) at (0,0) {$#1$};
	\node[circle, fill=gray!50] (LH) at (4,0) {$#2$};
	\node[circle, fill=gray!50] (HL) at (8,0) {$#3$};
	\node[circle, fill=gray!50] (HH) at (12,0) {$#4$};
	\node[circle, fill=gray!50] (L) at (2,2) {$2$};
	\node[circle, fill=gray!50] (H) at (10,2) {$2$};
	\node[circle, fill=gray!50] (R) at (6,4) {$1$};
	\draw[thick] (R) -- (L) node[pos=.5,sloped,above] {\edgea} -- (LL) node[pos=.5,sloped,above] {\edgec};
	\draw[thick] (R) -- (H) node[pos=.5,sloped,above] {\edgeb} -- (HH) node[pos=.5,sloped,above] {\edgef};
	\draw[thick] (L) -- (LH) node[pos=.5,sloped,above] {\edged};
	\draw[thick] (H) -- (HL) node[pos=.5,sloped,above] {\edgee};
}

\newcommand{\treetwovalues}[4]{\treeedges{L}{H}{L}{H}{L}{H}\tree{#1}{#2}{#3}{#4}}

\newcommand{\stpathup}[4]{
	\begin{tikzpicture}[scale=.8]
	\stpath{#1}{#2}{#3}{#4}
	\upperpath
	\end{tikzpicture}
}

In this section, we show how the domain structure impacts on the performance guarantee of OSP mechanisms, for the problem of scheduling related machines. Roughly speaking, the problem is easy for \emph{two-value} domains, while it becomes difficult already for \emph{three-value} domains and \emph{two} jobs.

\subsection{Two-value domains are easy}
OSP-CMON allows us to prove that an OSP optimal mechanism exists whenever each agent's domain has size two. Specifically, we have the following theorem.
\begin{theorem}\label{th:scheduling:many-agents:identical-jobs-}
	For the machine scheduling problem, there exists an optimal polynomial-time OSP mechanism for any number of agents with two-value domains $D_i=\{L_i,H_i\}$. 
	% 		for the case of identical jobs. 
\end{theorem}
The full proof of this theorem is given in Appendix~\ref{apx:2het}.
Here we instead show a proof of the above theorem for the simpler setting in which domains are homogeneous, i.e., $D_i = \{H, L\}$ for every $i$.
This case turns out to be interesting, 
since we can here use an explicit, simple payment function $P$ which rewards each agent an amount $p$, between $L$ and $H$, for each unit of allocated work,
% \begin{align}\label{eq:payments_two_val}
i.e., $p_i(\b) = p \cdot f_i(\b)$
% && \text{
for $L<p<H$.
% }\ . 
% \end{align}
We call such a mechanism a mechanism with \emph{proportional payments}.
The main intuition behind these payments is that agents with low cost have positive utility, while those with high cost have a negative utility.  In particular, agents with low cost aim to maximize their assigned work, and agents with high cost aim to minimize it. 
The following definition and theorem formalize this intuition.
\begin{definition}
	An algorithm $f$ is \emph{strongly monotone} for a two-value domain $D_i=\{L,H\}$ if
	% 	the following holds:
	% 	\begin{align}\label{eq:s-mon}
	$f_{i}(L,\b_{-i}) \geq f_{i}(H,\b'_{-i})$
	% 	\end{align}
	for all $i$ and for all $\b_{-i}, \b'_{-i} \in D_{-i}$.
\end{definition}

\begin{theorem}
	A mechanism $\M=(f,\p)$ with the proportional payments is OSP with direct-revelation implementation for a two-value domain  if and only if algorithm $f$ is strongly monotone for the corresponding two-value domain.
\end{theorem}
\begin{proof}
	Since $p_i(\b) = p f_i(\b)$, with $p>L$, and $f_{i}(L,\b_{-i}) \geq f_{i}(H,\b'_{-i})$
	% By simply combining \eqref{eq:payments_two_val} and \eqref{eq:s-mon}, since , 
	we have $u_i(L,\M(L,\b_{-i})) = (p -L)\cdot f_i(L,\b_{-i}) \geq (p -L)\cdot f_{i}(H,\b'_{-i})$. Moreover, since $p<H$,  we have $u_i(H,\M(H,\b_{-i})) = (p -H)\cdot f_i(H,\b_{-i}) \geq (p -H)\cdot f_{i}(L,\b'_{-i})$. As these two conditions hold for all $i$ and for all $\b_{-i}, \b'_{-i} \in D_{-i}$, we have shown that OSP holds.
\end{proof}

From this result we easily obtain an optimal OSP mechanisms for the scheduling problem.

\begin{proposition}\label{prop:scheduling:many-agents:identical-jobs}
	For the machine scheduling problem, there exists an optimal polynomial-time OSP mechanism for any number of agents with two-value domains $D_i=\{L,H\}$. 
	% 		for the case of identical jobs. 
\end{proposition}
\begin{proof}
	We show that there is an optimal allocation that is strongly monotone.
	Let $\w^{(\ell)}=(w^{\ell}_1,w^{\ell}_2,\ldots,w^{\ell}_n)$ denote some (suitably defined) optimal allocation for the case in which there are $\ell$ agents with cost $L$ and the remaining $n-\ell$ with cost $H$, with these workloads in non-decreasing order (i.e., $w_1^{(\ell)}\geq w_2^{(\ell)}\geq \cdots \geq w_n^{(\ell)}$). Note that $\w^{(n)}=\w^{(0)}$ is the allocation when all types are the same (all $L$ or all $H$). For a generic input $\b$ with $\ell$ types being $L$, we allocate jobs according to $\w^{(\ell)}$, following a fixed order of the agents among those with the same type: The $k$-th agent among those of type $L$ gets $w^{(\ell)}_k$ and the $j$-th agent of type $H$ gets $w^{(\ell)}_{\ell+j}$. Note that $w^{(\ell)}_\ell$ is the minimum load assigned to a machine declaring type $L$, while $w^{(\ell)}_{\ell+1}$ is the maximum load assigned to a machine that declare type $H$. Let us define $\w^{(n)}_{n+1}=0$. Then, if the following key inequalities hold for all $\ell$,
	% 		were :\footnote{For the sake of presentaion, .}
	\begin{align}\label{eq:vectors}
	w^{(\ell)}_\ell \geq w^{(n)}_\ell && w^{(n)}_n\geq w^{(\ell)}_{\ell+1}\ ,
	\end{align}
	then the allocation is strongly monotone. Indeed,
	$f_i(L,\b_{-i}) \geq w^{(\ell)}_{\ell}$, where $\ell$ is the number of machines of type $L$ in $(L,\b_{-i})$; similarly, $f_i(H,\b_{-i}') \leq w^{(\ell')}_{\ell'+1}$,  where $\ell'$ is the number of machines of type $L$ in $(H,\b_{-i}')$; finally, \eqref{eq:vectors} implies $w^{(\ell')}_{\ell'+1} \leq w^{(n)}_n \leq w^{(n)}_\ell \le w^{(\ell)}_\ell$. 
	
	To conclude the proof, we show how to guarantee \eqref{eq:vectors}. First, in the optimal allocation  $\w^{(n)}$, the load between machines of the same type differ by at most $1$. Then, every other optimal allocation $\w^{(\ell)}$ can be obtained starting from $\w^{(n)}$ by repeatedly reallocating one job in  the most loaded machine of type $H$ to the least loaded machine of type $L$. Every step increases the least loaded machine of type $L$ and decreases the most loaded machine of type $H$, and thus \eqref{eq:vectors} holds at every intermediate step until we obtain $\w^{(\ell)}$.
\end{proof}

\subsection{Three-value domains are hard}
In this section, we show that the problem becomes hard as soon as we have \emph{three-value} domains. Our contributions are a \emph{lower bound} of $\sqrt{n}$ on the approximation of any OSP mechanism, and a \emph{matching upper bound}.

%\subsection{Ternary domains} 

% \smallskip 
\noindent {\bf Lower bound.} 
% \subsubsection{Lower bound}
\label{sec:sched3_lb}
The $2$-inapproximability result in Proposition~\ref{prop:scheduling:three-values:lb} for three-value domains can be strengthen  to $\sqrt{n}$-inapproximability by considering a larger number of machines and jobs.
\begin{theorem}
	\label{thm:lower}
	For the machine scheduling problem, no OSP mechanism can be better than $\sqrt{n}$-approximate. This also holds for %identical jobs and 
	% certain
	three-value domains $D_i=\{L,M,H\}$.
\end{theorem}
In order to prove this lower bound, we consider  $m = n = c^2$, for some $c > 1$, and a three-value domain $D_i=\{L,M,H\}$ such that 
%\begin{align*}
$M \geq m \cdot L$ and %&&  \text{ and } && 
$H \geq m \sqrt{n} \cdot M.$ % \ .
%\end{align*}

First observe that, in such domains, every trivial mechanism must have an approximation ratio not lower than $\sqrt{n}$.
Consider then a non-trivial mechanism $\M$ and let $\T$ be its implementation tree. Let us rename the agents as follows: Agent $1$ is the $1$st agent in that diverges in $\T$; since the mechanism is not trivial agent $1$ exists. We now call agent $2$, the $2$nd agent that diverges in the subtree of $\T$ defined by agent $1$ taking an action compatible with type $H$; if no agent diverges in this subtree of $\T$ we simply call $2$ a random agent different from $1$. More generally, agent $i$ is the $i$th agent that diverges, if any, in the subtree of $\T$ that corresponds to the case that the actions taken by agents that previously diverged are compatible with their type being $H$. As above, if no agent diverges in the subtree of interest, we just let $i$ denote a random agent different from $1, 2, \ldots, i-1$. We denote with $u_i$ the node in which $i$ diverges in the subtree in which all the other agents have taken actions compatible with $H$; if $i$ does not diverge (i.e., got her id randomly) we denote with $u_i$ a dummy node in which we will say that $i$ does not diverge and $i$ takes an action compatible with every type in $D_i$.
% \pao{Why at least $n - \sqrt{n}+1$}
We then have the following lemma.
\begin{lemma}
	\label{lem:separ3}
	%For every $i \leq n - \sqrt{n} + 1$, and every OSP $k$-approximate mechanism $\M$, with $k < \sqrt{n}$ we have that at node $u_i$ $i$ diverges on $M$ and $H$.
	Any OSP mechanism $\M$ which is $k$-approximate, with $k<\sqrt{n}$, must satisfy the following conditions: 
	\begin{enumerate}
		\item For every $i \leq n - \sqrt{n} + 1$, if agent $i$ diverges at node $u_i$, it must diverge on $M$ and $H$.\label{itm:lem:separ3}
		\item For every $i \leq n - \sqrt{n}$, if agent $i$ at node $u_i$ takes an action compatible with her type being $H$, then $\M$ does not assign any job to $i$, regardless of the actions taken by the other agents. \label{itm:lem:notb3}
	\end{enumerate}
\end{lemma}
% The first part of the lemma is used to prove the second part, and the latter yields  the desired lower bound.
\begin{proof}
	Let us first prove part~\eqref{itm:lem:separ3}.
	Suppose that there is $i \leq n - \sqrt{n}+1$ such that at node $u_i$ $i$ diverges on $L$ and $\{M, H\}$.
	Consider the type profile $\x$ such that $x_i = M$, and $x_j = H$ for every $j \neq i$.
	Observe that $\x$ is compatible with node $u_i$.
	The optimal allocation for the type profile $\x$ assigns all jobs to machine $i$,
	with cost $OPT(\x) = m M$.
	Since $\M$ is $k$-approximate, then it also assigns all jobs to machine $i$.
	Indeed, if a job is assigned to a machine $j \neq i$, then the cost of the mechanism would be at least
	$H \geq \sqrt{n} \cdot m M > k \cdot OPT(\x)$, that contradicts the approximation bound.
	
	Consider now the profile $\y$ such that $y_i = L$, $y_j = H$ for every $j < i$, and $y_j = L$ for every $j > i$.
	Observe that also $\y$ is compatible with node $u_i$.
	It is not hard to see that $OPT(\y) \leq \left\lceil\frac{m}{n-i+1}\right\rceil \cdot L$.
	Since $\M$ is $k$-approximate, then it cannot assign all jobs to machine $i$.
	Indeed, in this case the cost of the mechanism contradicts the approximation bound, since it would be
	$m L \geq \sqrt{n} \left\lceil\frac{m}{n-i+1}\right\rceil L > k \cdot OPT(\y),$
	where we used that
	%  \begin{equation}
	%   \label{eq:sqrt_n_tom1}
	$\sqrt{n} \left\lceil\frac{m}{n-i+1}\right\rceil \leq \sqrt{n} \left\lceil\frac{m}{\sqrt{n}}\right\rceil = \sqrt{n} \left\lceil\frac{n}{\sqrt{n}}\right\rceil = \sqrt{n} \cdot \sqrt{n} = n = m$.
	%  \end{equation}
	
	Hence, we have that if $i$ takes actions compatible with $M$, then there exists a type profile compatible with $u_i$ such that $i$ receives $n$ jobs,
	whereas, if $i$ takes a different action compatible with a lower type, then there exists a type profile compatible with $u_i$ such that $i$ receives less than $n$ jobs.
	However, this contradicts the OSP CMON property.
	% \end{proof}
	
	%\begin{lemma}
	%\label{lem:notb3}
	%% For every $i \leq n - \sqrt{n}$, and every OSP $k$-approximate mechanism $\M$, with $k < \sqrt{n}$,
	%% if $i$ at node $u_i$ takes an action compatible with her type being $H$, then $\M$ does not assign any job to $i$,
	%% regardless of the actions taken by the other machines.
	%  Any OSP mechanism $\M$ which is $k$-approximate, with $k<\sqrt{n}$, must satisfy the following condition: For every $i \leq n - \sqrt{n} + 1$, if $i$ at node $u_i$ takes an action compatible with her type being $H$, then $\M$ does not assign any job to $i$,
	%  regardless of the actions taken by the other agents.
	%\end{lemma}
	% \begin{proof}[
	\smallskip
	Let us now prove part~\eqref{itm:lem:notb3}.
	%  of Lemma~\ref{lem:separ3})]
	Suppose that there is $i \leq n-\sqrt{n}$ and $\x_{-i}$ compatible with $u_i$ such that if $i$ takes an action compatible with type $H$,
	then $\M$ assigns at least a job to $i$.
	According to part~\eqref{itm:lem:separ3},
	%  of Lemma~\ref{lem:separ3}, 
	machine $i$ diverges at node $u_i$ on $H$ and $M$.
	
	Consider then the profile $\y$ such that $y_i = M$, $y_j = H$ for $j < i$, and $y_j = L$ for $j > i$.
	It is easy to see that the optimal allocation has cost $OPT(\y)=\left\lceil\frac{m}{n-i}\right\rceil \cdot L$.
	Since $\M$ is $k$-approximate, then it does not assign any job to machine $i$.
	Otherwise, the mechanism contradicts the approximation bound since his cost would be at least
	$
	M \geq m L \geq \sqrt{n} \left\lceil\frac{m}{n-i}\right\rceil L > k \cdot OPT(\x),
	$
	where we used that
	%  \begin{equation}
	%   \label{eq:sqrt_n_tom2}
	$\sqrt{n} \left\lceil\frac{m}{n-i}\right\rceil \leq \sqrt{n} \left\lceil\frac{m}{\sqrt{n}}\right\rceil = \sqrt{n} \left\lceil\frac{n}{\sqrt{n}}\right\rceil = \sqrt{n} \cdot \sqrt{n} = n = m$.
	%  \end{equation}
	
	Hence, we have that if $i$ takes actions compatible with $H$, then there exists a type profile compatible with $u_i$ such that $i$ receives one job,
	whereas, if $i$ takes a different action compatible with a lower type, then there exists a type profile compatible with $u_i$ such that $i$ receives zero jobs.
	However, this contradicts the OSP CMON property.
\end{proof}

We are now ready to prove our lower bound.
\begin{proof}[Proof of Theorem~\ref{thm:lower}]
	Suppose that there is an OSP $k$-approximate mechanism $\M$ for some $k < \sqrt{n}$.
	Consider $\x$ such that $x_i = H$ for every $i$.
	Observe that $\x$ is compatible with $u_i$ for every $i$.
	The optimal allocation consists in assigning a job to each machine,
	and has cost $OPT(\x) = H$.
	
	According to Part~\eqref{itm:lem:notb3} of Lemma~\ref{lem:separ3}, if machines take actions compatible with $\x$,
	then the mechanism $\M$ does not assign any job to machine $i$, for every $i \leq n-\sqrt{n}$.
	Hence, the best outcome that $\M$ can return for $\x$ consists in
	assigning $\sqrt{n}$ jobs to each of the other $\sqrt{n}$ machines.
	Therefore, the cost of $\M$ is at least $\sqrt{n} H > k OPT(\x)$,
	which contradicts the approximation ratio of $\M$.
\end{proof}

The arguments  above can be used to prove that ascending and descending auctions do not help in this setting.
Specifically, they cannot return an approximation better than $n$ (see Appendix~\ref{apx:asc_desc_sched}). % for the proof).

\smallskip \noindent{\bf Upper bound.} 
% \subsubsection{Upper bound}
%\input{sched3_ub}
We describe our mechanisms for a generic, not necessarily three-value, domain, as this turns out to be useful in the analysis.
% This will turn useful in the analysis of the following sections.
%The mechanisms make use of the following two functions:
%\begin{description}
%	\item[$\DesNext(A, p)$:] For any type value $p$ and any subset $A$ of machines,
%	$\DesNext(A, p)$ returns the pair $(i,t)$ such that $t$ is the \emph{largest} type value smaller than or equal to $p$
%	for which a machine $i \in A$ exists such that $t \in D_i$;
%	\item[$\AscNext(A, p)$] For any type value $p$ and any subset $A$ of machines, $\AscNext(A, p)$ returns the pair $(i,t)$ such that $t$ is the \emph{smallest} type value larger than or equal to $p$
%	for which a machine $i \in A$ exists such that $t \in D_i$.
%\end{description}
% $\DesNext$ and $\AscNext$ defined as follows:
%for every type value $p$ and every set $A$ of machines,
%$\DesNext(A, p)$ returns the pair $(i,t)$ such that $t$ is the largest type value smaller than or equal to $p$
%for which a machine $i \in A$ exists such that $t \in D_i$,
%whereas $\AscNext(A, p)$ returns the pair $(i,t)$ such that $t$ is the smallest type value larger than or equal to $p$
%for which a machine $i \in A$ exists such that $t \in D_i$.
%These functions break ties arbitrarily.
In what follows, the usual bold notation $\x$ denotes vectors of $n$ entries, while
a ``hat-bold'' notation $\hat{\x}$ denotes vectors of $\left\lceil\sqrt{n}\right\rceil$ entries only.

\smallskip \noindent \emph{A mechanism for many jobs (large $m$).}
% \paragraph{A mechanism for many jobs (large $m$)}
\newcommand{\Mmany}{\M_{many}}
We now introduce mechanism $\Mmany$ whose approximation ratio approaches $\left\lceil\sqrt{n}\right\rceil$,
whenever $m \gg \left\lceil\sqrt{n}\right\rceil$.
The mechanism consists of a Descending Phase (Algorithm~\ref{descending}) followed by  an Ascending Phase (Algorithm~\ref{ascending}).
\begin{algorithm}[htbp]
	\DontPrintSemicolon
	Set $A = [n]$, and $t_i = \max \{d \in D_i\}$\; 
	\While{$|A| > \left\lceil\sqrt{n}\right\rceil$}{
		Set $p = \max_{a\in A} \{t_a\}$ and $i = \min\{a\in A \colon t_a = p\}$\;
		%$(i,p) = \DesNext(A, (t_1,\ldots,t_n))$\;
		Ask machine $i$ if her type is equal to $p$\;
		\lIf{yes}{remove $i$ from $A$, and set $t_i = p$}
		\lElse{set $t_i = \max \{t \in D_i \colon t < p\}$}
	}
	%%% older version
	%Set $A = [n]$, $t_i = \max \{t \in D_i\}$; 
	%and $p = \max_i t_i$\;
	%\While{$|A| > \left\lceil\sqrt{n}\right\rceil$}{
	% $(i,p) = \DesNext(A, p)$\;
	% Ask to machine $i$ if her valuation is $p$\;
	% \lIf{yes}{remove $i$ from $A$, and set $t_i = p$}
	% \lElse{set $t_i = \max \{t \in D_i \colon t < p\}$}
	%}
	\caption{Descending Phase (for both mechanisms $\Mmany$ and $\M_{few}$).}
	\label{descending}
\end{algorithm}
\begin{algorithm}[htbp]
	\DontPrintSemicolon
	Set $s_i = \min \left\{d \in D_i\right\}$\;
	\While{$|A| > 0$}{
		Set $p =\min_{a\in A} \{s_a\}$ and $i = \min\{a\in A \colon s_a = p\}$\;
		%$(i,p) = \AscNext(A, (t_1,\ldots,t_n))$\;
		Ask machine $i$ if her type is equal to $p$\;
		\If{yes}{
			Consider the profile $\hat{\z}$ such that $\hat{z}_i = p$ and $\hat{z}_j = \min_{k \notin A} t_k$ for every $j \in A, j \neq i$ \nllabel{line:zeta}\;
			Let $f^\star(\hat{\z}) = \left(f^\star_i(\hat{\z})\right)_{i \in A}$ be the optimal assignment of jobs on input type profile $\hat{\z}$\;
			Assign $f^\star_j(\hat{\z})$ jobs to each machine $j \in A$\;
			Set $A = \emptyset$
		}\lElse{set $s_i = \min \{d \in D_i \colon d > p\}$}
	}
	%Set $p = \min_i \left\{t \in \bigcup_i D_i\right\}$\;
	%\While{$|A| > 0$}{
	% $(i,p) = \AscNext(A, p)$\;
	% Ask to machine $i$ if her valuation is $p$\;
	% \If{yes}{
	% Consider the profile $\hat{\z}$ such that $\hat{z}_i = p$ and $\hat{z}_j = \max_{k \in A} t_k$ for every $j \neq i$ \nllabel{line:zeta}\;
	% Let $f^\star(\hat{\z}) = \left(f^\star_i(\hat{\z})\right)_{i \in A}$ be the optimal assignment of jobs on input type profile $\hat{\z}$\;
	% Assign $f^\star_j(\hat{\z})$ jobs to each machine $j \in A$\;
	% Set $A = \emptyset$
	%}}
	\caption{Ascending Phase (for mechanism $\Mmany$).}
	\label{ascending}
\end{algorithm}

\begin{proposition}
	\label{prop:algo2mon}
	Mechanism $\Mmany$
	%  satisfies the OSP two-cycle monotonicity, and therefore it 
	is OSP for any three-value domain $D_i = \{L_i,M_i,H_i\}$.
\end{proposition}
\begin{proof}
	We prove that $\Mmany$ satisfies the OSP two-cycle monotonicity.
	The claim then follows from Theorem~\ref{thm:2cycle3}.
	% 	, it is enough to prove the two-cycle monotonicity: 
	Specifically, for each machine $i$, for each node $u$ in which the mechanism makes a query to $i$,
	for each pair of type profiles $\x, \y$ compatible with $u$ such that $i$ diverges at $u$ between $x_i$ and $y_i$,
	we need to prove that if $x_i > y_i$, then $f_i(\Mmany(\x)) \leq f_i(\Mmany(\y))$.
	
	Let us first consider a node $u$ corresponding to the descending phase of the mechanism.
	In this case, $x_i = p$, where $p$ is as at node $u$.
	Moreover, in all profiles compatible with $u$ there are at least $\left\lceil\sqrt{n}\right\rceil$ machines that either have a type lower than $p$,
	or they have type $p$ but are queried after $i$.
	However, for every $\x_{-i}$ satisfying this property, we have that $f_i(\Mmany(\x)) = 0$, which implies that these two-cycles are non-negative.
	
	Suppose now that node $u$ corresponds to the ascending phase of the mechanism.
	In this case, $y_i = p$, where $p$ is as at node $u$.
	Observe that for every $\y_{-i}$ compatible with node $u$,
	$f_i(\Mmany(\y)) = f^\star_i(y_i, \hat{\z}_{-i})$, where $f^\star_i(y_i, \hat{\z}_{-i})$ is the number of jobs assigned to machine $i$
	by the optimal outcome on input profile $(y_i, \hat{\z)_{-i}}$,
	$\hat{\z}_{-i}$ being such that $\hat{z}_j = \max_{k \in A} t_k$ for every $j \in A \setminus \{i\}$.
	Observe that for every $\x$ compatible with $u$, it must be the case that $x_j \geq y_i$ for every $j \in A$.
	Hence, we can distinguish two cases:
	if $\min_j x_j = x_i$, then $f_i(\Mmany(\x)) = f^\star_i(x_i,\hat{\z}_{-i}) \leq f^\star_i(y_i,\hat{\z}_{-i}) = f_i(\Mmany(\y))$;
	if instead $\min_j x_j = x_k$, for some $k \neq i$,
	then $f_i(\Mmany(\x)) = f^\star_i(x_k,\hat{\z}_{-k}) \leq f^\star_k(x_k,\hat{\z}_{-k}) \leq f^\star_i(y_i,\hat{\z}_{-i}) = f_i(\Mmany(\y))$,
	where we used that $\hat{\z}_{-k} = \hat{\z}_{-i}$ and the inequalities follow since:
	(i) in the optimal outcome the fastest machine must receive at least as many jobs as slower machines;
	(ii) the optimal outcome is monotone, (i.e., given the speeds of other machines, the number of jobs assigned to machine $i$ decreases as its speeds decreases).
\end{proof}

%Theorem~\ref{thm:algo2mon} and Theorem~\ref{thm:2cycle3} implies the following corollary.
%\begin{corollary}
% \label{cor:algoosp}
% Mechanism $\Mmany$ is OSP for  any three-value domain $D_i = \{L_i,M_i,H_i\}$.
%\end{corollary}

\begin{proposition}
	\label{prop:algoapprox}
	Mechanism $\Mmany$ is  $O\left(\sqrt{n} + \frac{n}{m}\right)$-approximate. %assignment, where 
	%$\alpha(n,m)=\frac{m+\left\lceil\sqrt{n}\right\rceil-1}{m} \cdot \left\lceil\sqrt{n}\right\rceil$.
\end{proposition}
\begin{proof}[Proof Sketch]
	We show that the mechanism returns an allocation which is $\left(\frac{m+\left\lceil\sqrt{n}\right\rceil-1}{m} \cdot \left\lceil\sqrt{n}\right\rceil\right)$-approx\-imate. % with respect to the input.
	We denote with $OPT(\x)$ the makespan of the optimal assignment when machines types are as in the input profile $\x$.
	We will use the same notation both if the optimal assignment is computed on a set of $n$ machines
	and if it is computed and on a set of $\left\lceil\sqrt{n}\right\rceil$ machines, since these two cases
	can be distinguished through the input profile.
	
	Fix a type profile $\x$. Let $OPT(\x)$ be the makespan of the optimal assignment on type profile $\x$.
	Let $A$ be as at the beginning of the ascending phase.
	Let $w$ be the last query done in this phase and let $i$ be the last agent queried, i.e., the one that answered yes.
	Moreover, let $t = \min_{j \notin A} t_j$.
	It is immediate to see that $OPT(\x) \geq OPT(\y)$ where $\y$ is such that $y_j = w$ for every $j \in A$, and $y_j = t$, otherwise.
	Moreover, let $\M(\x)$ be the makespan of the assignment returned by our mechanism on the same input.
	Then, $\M(\x)$ is equivalent to $OPT(\hat{\z})$, where $\hat{\z}$ is such that $\hat{z}_i = w$ and $\hat{z}_j = t$ for each remaining machine $j$.
	Hence, the theorem follows by proving that $\frac{OPT(\hat{\z})}{OPT(\y)} \leq \frac{m+\left\lceil\sqrt{n}\right\rceil-1}{m} \cdot \left\lceil\sqrt{n}\right\rceil$. (The full proof is given in appendix.)
\end{proof}

The next corollary follows by simple algebraic manipulations.
\begin{corollary}
	\label{cor:algoapprox}
	Mechanism $\Mmany$ is $\left(\left\lceil\sqrt{n}\right\rceil+1\right)$-approximate for $m > \left\lceil\sqrt{n}\right\rceil^2$.
\end{corollary}

% \paragraph{Remarks about extension to non-identical jobs.}
% Can the argument of Theorem~\ref{thm:algoapprox} be extended to prove a similar bound even for the case of non-identical jobs
% (maybe with the factor in front of $\left\lceil\sqrt{n}\right\rceil$ will depend on the maximum job weight)?
% 
% The main obstacle that I see in this case is that we do not have a well-defined formula for the optimum,
% since it may be impossible to evenly split jobs among machines with the same type.
% In particular, consider the case there are two large jobs.
% It may be the case that $\hat{\z}$ must assign one of these large jobs to a machine with type $t$,
% resulting in a very unbalanced assignment among these machine.

% \smallskip 
\noindent \emph{A mechanism for few jobs (small $m$).}
% \paragraph{A mechanism for few jobs (small $m$)}
\renewcommand{\A}{\M_{few}}
We now introduce a mechanism $\A$ which is OSP and $\left\lceil\sqrt{n}\right\rceil$-approximate whenever $m\leq \left\lceil\sqrt{n}\right\rceil^2$. 
Like $\Mmany$, also $\A$ consists of a descending phase followed by an ascending phase. The descending phase is exactly the same (Algorithm~\ref{descending}) with the difference that the ascending phase (Algorithm~\ref{ascending2}) does not need the information on the type of the machines that are not in $A$ at that point.
%\begin{algorithm}[htbp]
%\DontPrintSemicolon
%Set $A = [n]$, and $t_i = \max \{d \in D_i\}$\; 
%\While{$|A| > \left\lceil\sqrt{n}\right\rceil$}{
%	Set $p = \max_{a\in A} \{t_a\}$ and $i = \min\{a\in A \colon t_a = p\}$\;
%	%$(i,p) = \DesNext(A, (t_1,\ldots,t_n))$\;
%	Ask machine $i$ if her type is $p$\;
%	\lIf{yes}{remove $i$ from $A$}
%	\lElse{set $t_i = \max \{t \in D_i \colon t < p\}$}
%}
%%Set $A = [n]$, and $p = \max_i t_i$\;
%%\While{$|A| > \left\lceil\sqrt{n}\right\rceil$}{
%% $(i,p) = \DesNext(A, p)$\;
%% Ask to machine $i$ if her valuation is $p$\;
%% \lIf{yes}{remove $i$ from $A$}
%%}
% \caption{Descending Phase (for mechanism $\A$).}
% \label{descending2}
%\end{algorithm}

\begin{algorithm}[htbp]
	\DontPrintSemicolon
	Set $t_a = \min_i \left\{d \in D_i\right\}$ and $C = m$\;
	\While{$|A| > 0$}{
		Set $q =\min_{a\in A} \{t_a\}$ and $i = \min\{a\in A \colon t_a = q\}$\;
		%$(i,q) = \AscNext(A, q)$\;
		Ask machine $i$ if her type is $q$\;
		\If{yes}{
			Let $\zeta = \left\lceil\frac{C}{|A|}\right\rceil$\;
			%\lIf{$\zeta \cdot q > \left\lceil\sqrt{n}\right\rceil \cdot p$}{exit}
			%\Else{
			Let $z$ be the largest integer in $\left[\zeta, C\right]$ such that $z\cdot q \leq \left\lceil\sqrt{n}\right\rceil \cdot p$ \label{line:z}\; 
			Assign $z$ jobs to $i$ and set $C = C - z$\;
			Remove $i$ from $A$
			%}
		}\lElse{set $t_i = \min \{d \in D_i \colon d > p\}$}
	}
	\caption{Ascending Phase (for mechanism $\A$).}
	\label{ascending2}
\end{algorithm}

We first show that mechanism $\A$ is well defined under our assumption on the number of jobs.
\begin{lemma}
	If $m \leq \left\lceil\sqrt{n}\right\rceil^2$ then there exists a $z$ in line \ref{line:z} of Algorithm \ref{ascending2}.
\end{lemma}
\begin{proof}
	We next show that it never occurs during the ascending phase that $\zeta \cdot q > \left\lceil\sqrt{n}\right\rceil \cdot p$.
	Indeed, for the first machine to reveal the type during the ascending phase, we have that $|P_0| = m \leq \left\lceil\sqrt{n}\right\rceil^2$, and, thus $\zeta \leq \left\lceil\sqrt{n}\right\rceil$. Hence, $\zeta \cdot q \leq \left\lceil\sqrt{n}\right\rceil \cdot p$ since $q \leq p$.
	If a set $Q \subset A$ of machines has previously revealed the type during the ascending phase, and the execution of this phase has not been stopped,
	then these machines received at least $m' = \left\lfloor\frac{|Q| m}{|A|}\right\rfloor + \min\{|Q|, m \mod |A|\}$ jobs.
	Then $|P_0| = m - m' \leq (|A| - |Q|)\left\lceil\sqrt{n}\right\rceil$, and thus $\zeta \leq \left\lceil\sqrt{n}\right\rceil$, and,
	since $q \leq p$, $\zeta \cdot q \leq \left\lceil\sqrt{n}\right\rceil \cdot p$.
	%
	% The lemma follows since
	%  it is immediate to check that
	% in this case the mechanism correctly assigns every job to a machine.
\end{proof}

%We will denote with $\A(\x)$ the outcome of the mechanism on type profile $\x$
%and with $w_i(\A(\x))$ the number of jobs that this outcome assigns to machine $i$.
\begin{proposition}
	\label{prop:algo2mon2}
	Mechanism $\A$ %returns an assignment, then it 
	%  satisfies  two-cycle monotonicity, and therefore it
	is OSP for three-value domains $D_i = \{L_i,M_i,H_i\}$.
\end{proposition}
\begin{proof}
	We prove that $\Mmany$ satisfies the OSP two-cycle monotonicity.
	The claim then follows from Theorem~\ref{thm:2cycle3}.
	%By  Theorem~\ref{thm:2cycle3}, it is enough to prove the two-cycle monotonicity: 
	Specifically,
	for each machine $i$, for each node $u$ in which the mechanism makes a query to $i$,
	for each pair of type profiles $\x, \y$ compatible with $u$ such that $x_i$ and $y_i$ diverge at $u$,
	we need to prove that if $x_i > y_i$, then $f_i(\A(\x)) \leq f_i(\A(\y))$.
	
	Let us first consider a node $u$ corresponding to the descending phase of the mechanism.
	In this case, $x_i = p$, where $p$ is as at node $u$.
	Moreover, in all profiles compatible with $u$ there are at least $\left\lceil\sqrt{n}\right\rceil$ machines that either have a type lower than $p$,
	or they have type $p$ but they are queried after $i$.
	Hence, for every $\x_{-i}$ satisfying this property, we have that $f_i(\A(\x)) = 0$, that implies the claim.
	
	Suppose now that node $u$ corresponds to the ascending phase of the mechanism.
	Let $C(u)$, $A(u)$, $p(u)$ and $q(u)$ be the value of $C$, $A$, $p$ and $q$ at that node.
	Observe that for every profile compatible with $u$,
	%  we have that 
	the type of machines not in $A(u)$ is fixed,
	whereas for every machine in $A(u)$, the type is at least $q(u)$.
	Moreover,
	%  it must be 
	% %  the case 
	%  that 
	$y_i = q(u)$.
	Hence, $f_i(\A(\y))$ is the largest integer $z \leq C(u)$ such that $z\cdot y_i \leq \left\lceil\sqrt{n}\right\rceil \cdot p(u)$.
	On the other side, for every $x_i > y_i$, $f_i(\A(\x))$ is at most the largest integer $z' \leq C(u)$
	such that $z'\cdot x_i \leq \left\lceil\sqrt{n}\right\rceil \cdot p(u)$.
	Since $x_i > y_i$, then $z' \leq z$, and the lemma follows.
\end{proof}

%Theorem~\ref{thm:algo2mon2} and Theorem~\ref{thm:2cycle3} implies the following corollary.
%\begin{corollary}
% \label{cor:algoosp2}
% If mechanism $\A$ returns an assignment, then it is OSP whenever $|D_i| \leq 3$ for each agent.
%\end{corollary}

\begin{proposition}
	\label{prop:algoapprox2}
	Mechanism $\A$ %returns an assignment, then it 
	is $\left\lceil\sqrt{n}\right\rceil$-approximate.
\end{proposition}
\begin{proof}
	Let $\OPT(\x)$ be the optimal outcome on a type profile $\x$.
	Let $A(\A)$ be the set of alive machines at the end of the descending phase of mechanism $\A$.
	
	We first consider the case that the optimal mechanism assigns jobs to only the $\left\lceil\sqrt{n}\right\rceil$ machines with lower type.
	In particular, let $B$ be the set of machines that received a job in the optimal allocation,
	i.e., $B = \left\{i \colon f_i(\OPT(\x)) > 0 \right\}$.
	Note that we can assume that the machines that receive a job are in $A(\A)$,
	and they are exactly the first $|B|$ machines to be removed from this set during the descending phase of our mechanism.
	%  Not only, w
	We can also assume that assignment to these $|B|$ machines is monotone non-increasing,
	that is $f_i(\OPT(\x)) \leq f_j(\OPT(\x))$ if $j$ reveals her type before $i$.
	Indeed, these properties are satisfied by opportunely breaking ties among optimal outcomes.
	
	Observe that $f_i(\OPT(\x)) \cdot x_i \leq |B| \cdot p$,
	otherwise a lower makespan would be achieved by moving a job from every machine in $B$ to the one with type $p$,
	where $p$ is as at the end of the descending phase.
	However, this implies that the mechanism $\A$ also assigns at least $\min\left\{f_i(\OPT(\x)), |P_0|_i\right\}$ jobs to machine $i \in B$,
	where $|P_0|_i$ is the number of jobs that mechanism $\A$ has not assigned yet when the type of $i$ is revealed.
	
	Now, for every $i \in B$, let $B_i$ be the set of machines that revealed her type before $i$.
	Observe that, since the optimal outcome is monotone with respect to the order of revelation,
	then it must assign to $i$ at least $\frac{m-\sum_{j \in B_i} f_j(\OPT(\x))}{|B|-|B_i|}$ jobs.
	On the other side, our mechanism assigns to $i$ at most $m - \sum_{j \in B_i} f_j(\A(\x)) \leq m - \sum_{j \in B_i} f_j(\OPT(\x))$,
	where the inequality follows from the observation above.
	Hence, the approximation ratio is at most $|B|-|B_i|\leq |B| \leq \left\lceil\sqrt{n}\right\rceil$.
	
	Suppose now that optimal mechanism assigns jobs also to a machine that $\A$ remove during the descending phase.
	Observe that, since by hypothesis, $\A$ returns an outcome, then during the execution of the mechanism it always hold that $\zeta \cdot q \leq \left\lceil\sqrt{n}\right\rceil \cdot p$.
	Since $\OPT(\x)$ assigns at least one job to machines not in $A(\x)$
	In this case we have that $\max_i f_i(\OPT(\x)) \geq p$.
	By hypothesis, instead, we have that $f_i(\A(\x)) \leq \left\lceil\sqrt{n}\right\rceil \cdot p$ for every machine $i \in A(\x)$ that receives at least one job.
	The bound then immediately follows.
\end{proof}

\section{Set systems}
\label{sec:set_sys}
%In this section, we present our results on set systems. We begin with a warm-up on the conditions needed to compute optimum solutions with OSP mechanisms on binary domains; to introduce ideas and give intuitions on our techniques, we keep the presentation on shortest path. We then give our general results.
  
%\subsection{Optimum for shortest path on simple networks}
%\input{opt-sp-binary}

%\subsection{The general result} 
%\input{set_sys_opt}

%\section{Set systems}
In this section we characterize when optimal OSP mechanisms exist for set systems. We will formally define the concept of alignment introduced above, with a different and more technical terminology. The main message is that the feasibility of optimal OSP mechanisms depends on structural properties of the feasible solutions \emph{and} the values in the agents' domains.

\subsection{Key concepts}
\newcommand{\D}{\textsf{\textbf{D}}}
\newcommand{\tD}{\tilde{\D}}
\newcommand{\td}{\tilde{D}}
\newcommand{\ttd}{\dbtilde{D}}
Consider a set system problem $(E, {\mathcal{F}}, \D)$ where
there are $|E|$ elements controlled by the agents, having three-value domains. Specifically, each $e\in E$ is an agent whose domain is $D_e\subseteq \{L_e, M_e, H_e\}$ with  $L_e < M_e < H_e$;  $\mathcal{F}$ is the set of feasible solutions $P \subseteq E$, and $\D=(D_e)_{e \in E}$ denotes the domain. We assume that this triple is the input to our mechanism design problem; the implementation tree will query the agents in order to elicit the input instance $\b$ from $\D$ and return an optimal solution for $\b$ whilst guaranteeing OSP. 

We next define some useful concepts and notation, to state our characterization and mechanism. 
Consider an arbitrary \emph{subdomain} $\tD$ of $\D$, that is, a type domain $\tD = ({\td}_e)_{e \in E}$ such that $\td_e \subseteq D_e$ for all $e\in E$. We denote by $L(e,\tD)=\min\{t \in \td_e\}$ and $H(e,\tD)=\max\{t \in \td_e\}$ the lowest  and the highest type for $e$ according to the subdomain $\tD$. Similarly, for any $P\subseteq E$, we let  $L(P, \tD)$ and $H(P, \tD)$
be the lowest and the highest possible cost of $P$ according to subdomain $\tD$, i.e.,  
\begin{align*}
L(P, \tD) = \sum_{e \in P} L(e, \tD) && \text{ and }
&& H(P, \tD) = \sum_{e \in P} H(e,\tD). 
\end{align*}
(When clear from the context, we omit the reference to %$\vett{t}$ and 
$\tD$ in these notations.)
Finally, we let $\prec$ denote a total order among the feasible solutions in $\mathcal{F}$; this order will be used to select the optimal solution to return in case of ties. %; in other words, our mechanism will use a fixed tie-breaking rule. 

%Note that this tie-breaking rule is defined independently from the implementation tree used; we will indeed build a suitable implementation tree based on $\prec$.  
We now introduce some concepts, along with some properties of theirs, that relate  implementation trees of an extensive-form mechanism with the optimality of the solutions. Any implementation tree gradually shrinks $\D$ to subdomains $\tD$ by querying the agents. 
%
%The idea of our OSP mechanism is that of gradually and carefully shrinking $\D$ to a subdomain $\tD$ to which the input instance belongs. This will allow the mechanism to exclude feasible solutions that cannot be optimal (and then selected) for any instance in $\tD$. 
The concept of selectable solution for a certain subdomain $\tD$ captures the intuition that if the implementation tree has already shrunk the domain to $\tD$, the solution in question cannot be excluded a priori if we care about optimality (because, for some profile in $\tD$, it is either the unique optimum or the optimum preferred according to tie-breaking rule). % $\prec$ of the mechanism). %, meaning  that,  if  the mechanism has already shrunk the domain to $\tD$, the solution in question cannot be excluded a priori (because, for some profile in $\tD$, it is either  the unique optimum or the optimum preferred according to tie breaking rule $\prec$ of the mechanism).
\begin{definition}[selectable solution]\label{def:sel}
	A feasible solution $P\in \mathcal{F}$ is said \emph{selectable} for a subdomain $\tD$ if for every other $P' \in \mathcal{F}$ it holds that
	$$L(P \setminus P', \tD) < H(P' \setminus P, \tD)$$ 
	or
	$$L(P \setminus P', \tD) = H(P' \setminus P, \tD) \textrm { and } P \prec P'.$$
\end{definition}
%Note that $P$ is selectable for $\tD$ if there exists a type profile $\vett{t}$
%compatible with $\tD$ such that $P$ is an optimal solution for $\vett{t}$ and,
%if it is not the only one, then it is favored by the tie-breaking rule. 
%This means that given that the mechanism has shrunk the domain to $\tD$ then $P$ cannot be excluded as it can still be the optimal solution we aim to return.  
Observe  that there is at least one selectable solution for every subdomain $\tD$. 

Since the implementation tree may shrink the subdomain $\tD$ even further, it is useful to relate the optimality for the bigger domain with the optimality for smaller domains. The next lemma shows that an optimal mechanism can ``forget'' about solutions that are 
%In order to safely exclude solutions that are 
\emph{not} selectable for the current subdomain $\tD$. %, we need that this property is inherited by smaller domains. We show this in the next lemma.
\begin{lemma}\label{le:inheritance}
	If $P$ is not selectable for $\tD$, then it is not selectable for every subdomain $\tD'$ of $\tD$, that is, for every $\tD'$ such that  $\td'_e \subseteq \td_e$ for every $e \in E$. 
\end{lemma}
\begin{proof}
	By contradiction, assume that $P$ is selectable for $\tD'$ and not selectable for $\tD$. 
	Since $\td_e' \subseteq \td_e$ for every $e \in E$, we have $L(e, \tD') \geq L(e, \tD)$ and $H(e, \tD') \leq H(e, \tD)$. Since $P$ is selectable for $\tD'$, we have that for every other $P' \in \mathcal{F}$
	\begin{equation*}
	L(P \setminus P', \tD) \leq L(P \setminus P', \tD') \leq H(P' \setminus P, \tD') \leq H(P' \setminus P, \tD)\ , 
	\end{equation*}
	where the second inequality is either strict, or it holds with equality together with $P \prec P'$ (cf. Definition~\ref{def:sel}). This means that $P$ is selectable for $\tD$, thus a contradiction. 
\end{proof}

While the above concepts refer only to implementation tree and optimality, the next concept will turn out to be useful to study when there is a way to shrink $\D$ (i.e., an implementation tree) that returns an optimal solution but also that is compatible with OSP. %To this aim, we need to strengthen the concept of selectable solution as follows.  
\begin{definition}[strongly selectable solution]\label{def:strsel}
	A selectable solution $P$ is said \emph{strongly selectable} for a subdomain $\tD$ if, for all $e \in P$, it continues to be selectable even for the subdomain $(\tD_{-e}, H(e, \tD))$, where $\tD_{-f}=(\tilde{D}_e)_{e\neq f}$ and, with a slight abuse of notation,  $H(e, \tD)$ denotes $\{H(e, \tD)\}$. 
\end{definition}
In words, this means that %we still cannot exclude 
solution $P$ is still potentially optimum %also 
when any one of its elements has the largest possible cost $H(e, \tD)$ in the subdomain $\tD$ under consideration.  That is, should any of the elements of a strongly selectable solution reveal to the mechanism type $H(e, \tD)$, then optimality implies that we cannot ``forget'' about this solution. 
\begin{example}[path auctions, selectable vs strongly selectable solutions]
	To understand the difference between selectable and strongly selectable solutions, consider path auction on the graph in Figure~\ref{graph21} for $\td_e=\{L,H\}$, with $2L < H$, for all $e$. Both solutions are selectable for this subdomain $\tD$: the bottom path because $L <2 H$ and the top path since $2L<H$. 
	However, only the bottom path is strongly selectable since the top path cannot be the optimum as soon as one of its edges has cost $H$. On the contrary, for the graph in Figure \ref{graph22}, both paths are strongly selectable since, for either path, we can always set the cost of the alternative solution to $2H > H+L$. 
\end{example}

\subsection{Necessary condition}
Our next two lemmas identify \emph{necessary} conditions for the implementation tree of an optimal OSP mechanism for set systems. The first roughly says that if there exist a subdomain where elements of strongly selectable solutions can be excluded when they reveal their type to be as low as possible, then there is no implementation tree which yields an optimal OSP mechanism.

\begin{lemma}
	\label{lem:bin2ter_cond1}
	There is no optimal OSP mechanism for a set system problem
	if there is a subdomain $\tD$ of $\D$   such that the following properties are both satisfied:
	\begin{itemize}
		\item[(i)] the set $\cal S$ of strongly selectable solutions for $\tilde{\D}$ contains at least one $P$ 
		with $f \in P$ such that $|\tilde{D}_f| > 1$;
		%\item[(i)] the set $\cal S$ of strongly selectable solutions $P$ for $\tilde{\D}$
		%for which there is $f \in P$ such that $|\tilde{D}_f| > 1$ is not empty;\pao{Why do we need the first condition (i)?}
		\item[(ii)] for every $P \in \cal S$ and every $f \in P$ such that $|\tilde{D}_f| > 1$,
		there is $\notP_f \in \mathcal{S}$ with $f \not\in \notP_f$ such that $\notP_f$ remains selectable even for $(\tD_{-f}, L(f,\tilde{\D}))$.
		%	\item[(ii)] for every $P \in \cal S$ and $f \in P$ for which $|\tilde{D}_f| > 1$,
		%	there is $\notP_{f} \notni f$ that is strongly selectable for $\tilde{\D}$
		%	and remains selectable even for $(\tD_{-f}, L(f,\tilde{\D}))$.
	\end{itemize}
\end{lemma}
\begin{proof}
	Assume by contradiction that there is a subdomain $\tD$ for which the conditions above are satisfied and yet there is an optimal OSP mechanism $\M$;  let us denote with $\mathcal{T}$ its implementation tree. Fix the subdomain to be this particular $\tD$.
	
	\renewcommand{\L}{\vett{L}}
	\renewcommand{\H}{\vett{H}}
	
	Let $\mathcal{S}$ be the set of strongly selectable solutions defined in the statement, which is not empty by hypothesis.  Consider the first node $u \in \mathcal{T}$ in which
	an agent $f \in \bigcup_{P \in \mathcal{S}} P$
	diverges between $L(f,\tD)$ and $H(f,\tD)$
	in the subtree compatible with
	the type of every agent $e \in \bigcup_{P \in \mathcal{S}} P$ being in $\tilde{D}_e$ and the type of every remaining agent $e$ being $H(e,\tilde{\D})$.
	We first argue that, since the mechanism $\M$ is optimal, such a node $u$ must exist.
	%Assume the contrary and let $P \in \mathcal{S}$, with $f \in P$ such that $|\tilde{D}_f| > 1$, and $\notP_f$ be as in the statement. We have assumed that no element in $P$ can ever diverge in $\mathcal{T}$. Moreover, since $\notP_f$ is strongly selectable and as no  node $u$ exists, even if $H(e,\tD) \neq L(e,\tD)$, none of the elements $e \in \notP_f$ diverges between $H(e,\tD)$ and $L(e,\tD)$. So every bid profile where agents $e$ in $P \cup \notP_f$ have bid in $\{L(e,\tD), H(e,\tD)\}$ will be associated to the same leaf of $\mathcal{T}$ and will be mapped to the same solution. 
	Consider  the following two bid profiles $\b$ and $\bar\b$ in $\tD$ defined as
	\begin{align*}
	b_e = \begin{cases}
	L(e,\tD) & \text{if $e \in P$} \\
	H(e,\tD) & \text{if $e \not\in P$}
	\end{cases} && \text{ and } &&
	\bar b_e = \begin{cases}
	L(e,\tD) & \text{if $e \in \notP_f$} \\
	H(e,\tD) & \text{if $e \not\in \notP_f$}
	\end{cases}
	\end{align*}
	Since both $P$ and $\notP_f$ are (strongly) selectable for $\tD$, Definition~\ref{def:sel} says that the mechanism must return $P$ on input $\b$ and $\notP_f$ on input $\bar \b$. Therefore, $P = \M(\b) \neq \M(\bar \b)=\notP_f$. Since $\b$ and $\bar \b$ differ only in the bids of agents $e$ in $P \cup \notP_f$, there must be some node $u\in \T$ in which some agent $f \in P \cup \notP_f$ diverges   between $L(f,\tD)$ and $H(f,\tD)$.

	%	$\b$ such that $b_e =L(e,\tD)$ for every $e \in P$, and $b_e = H(e,\tD)$ for every $e \notin P$; and $\b'$ such that $b'_e = L(e,\tD)$ for every $e \in \notP_f$, $b'_e = H(e,\tD)$ for every $e \notin \notP_f$.
	%	
	%	 $\x$ such that $x_e = L(e)$ for every $e \in P$, and $x_e = H(e)$ for every $e \notin P$; and $\y$ such that $y_e = L(e)$ for every $e \in \notP_{f}$, $y_e = H(e)$ for every $e \notin \notP_{f}$. Since $P$ is (strongly) selectable for $\tD$ then the solution output of $\M(\x)$ should be $P$ (cf. Definition \ref{def:sel}); similarly, as $\notP_{f}$ is (strongly) selectable for $(\tD$ then the solution output of $\M(\y)$ should be $\notP_{f}$ --- a contradiction.
	
	Given the existence of $u$ and $f$ as above, we now apply the hypothesis (ii) as follows to show a negative OSP $2$-cycle. For the strongly selectable solution $P^*$ containing $f$ %(which is either $P$ or $\notP_f$ as above) 
	there is another solution $Q^*$ which does not contain $f$, and $Q^*$ is also selectable for $(\tD,L(f,\tD))$. We next define two profiles
	$\b^{(H)}$ and $\b^{(L)}$ where $f$ has high and low cost, respectively, and it is selected only for the high cost:
	\begin{align}
	%	b^{(H)}_g = H(g,\tD), && b^{(L)}_g = L(g,\tD), && \b_e^{(H)} = \b_e^{(L)} =L(e,\tD) \ \  \text{ for all } e \in P^* \cap Q^*
	%	\intertext{and}
	b_e^{(H)} = \begin{cases}
	H(f,\tD) & \text{for $e=f$} \\
	L(e,\tD) & \text{for $e\neq f$ and $e \in P^*$ } \\
	H(e,\tD) & \text{otherwise}
	\end{cases} &&
	b_e^{(L)} = \begin{cases}
	L(f,\tD) & \text{for $e=f$} \\
	L(e,\tD) & \text{for $e \in Q^*$ (NB $e \neq f$) } \\
	H(e,\tD) & \text{otherwise}
	\end{cases}
	\label{eq:bids:2cycle}
	\end{align} 
	It is not hard to see that $\b^{(H)}$ and $\b^{(L)}$ are compatible with node $u$, and thus the cycle between $\b^{(H)}$ and $\b^{(L)}$ exists. Since $P^*$ is strongly selectable,  and $Q^*$ is selectable for $(\tD,L(g,\tD))$, from \eqref{eq:bids:2cycle} we have $\M(\b^{(H)}) = P^*$ and $\M(\b^{(L)}) = Q^*$.  Therefore, the cycle between $\b^{(H)}$ and $\b^{(L)}$ is negative.
	%
	%	We are now going to show a negative OSP $2$-cycle
	%	for the agent $g$ that diverges at node $u$, reaching a contradiction. \pao{We did not prove this particular $f$ diverges}
	%	Consider profiles $\z$ and $\w$ such that
	%	$z_e = w_e = L(e,\tD)$ for every $e \in P \cap \notP_{g}$,
	%	$z_e = L(e,\tD)$ and $w_e = H(e,\tD)$ if $e \in \notP_{g} \setminus P$ or $e=g$,
	%	$z_e = H(e,\tD)$ and $w_e = L(e,\tD)$ if $e \in P \setminus \notP_{g}$,
	%	and $z_e = w_e = H(e,\tD)$ otherwise.
	%	By definition of $u$ we can conclude that both
	%	$\z$ and $\w$ are compatible with node $u$.
	%	However, in $\z$, where agent $f$ has type $L(f,\tD)$,
	%	an optimal mechanism selects $\notP_f$ (and thus it does not select $f$)
	%	since it is selectable even if the type of $f$ is $L(f,\tD)$.
	%	In $\w$, where agent $f$ has type $H(f,\tD) > L(f,\tD)$,
	%	instead mechanism $\M$ selects $P$ (and thus $f$),
	%	since $P$ is strongly selectable. Therefore, the cycle between $\z$ and $\w$ is negative.
\end{proof}

The second necessary property  regards  subdomains $\tD$ for which there are solutions that are selectable but \emph{not} strongly selectable. For each such solution $P$ 
there is an agent $w$, that we will call the \emph{witness} of $P$,
such that $P$ is no longer selectable for $\tD_{hw}=(\tD_{-w}, H(w, \tD))$. 
%If $w$ is a witness of $P$, then for every $\vett{t}$ compatible with $\tD'$ there is a solution $P'$ that is selected for $\vett{t}$. By Lemma \ref{le:inheritance}, $P'$ is selectable also for $\tD$. 
A couple of easy properties of witnesses are summarized in the next observation.

\begin{obs}\label{obs:witness}
	Let $P$ be a solution selectable but not strongly selectable for a subdomain $\tD$ and $w$ a witness of $P$. Then, we have, (i) $|\td_w|\geq 2$; (ii) there is $P' \in \mathcal{F}$ such that $w \in P \setminus P'$ and either $H(w) + L(P \setminus (P' \cup \{w\})) > H(P' \setminus P)$ or $H(w) + L(P \setminus (P' \cup \{w\})) = H(P' \setminus P)$ and 
	$P' \prec P$. 
\end{obs}
\begin{proof}
	The first observation is easy as if $|\td_w| = 1$ then $P$ would be selectable also for $\tD_{hw}$. Similarly, if $w \not\in P \setminus P'$ for every $P' \in \mathcal{F}$, then the conditions of Definition \ref{def:sel} would be satisfied for $\tD_{hw}$ since they are for $\tD$. But then, since $P$ is not selectable for $\tD_{hw}$, there is a feasible $P'$ such that $w \in P \setminus P'$ for which either $H(w) + L(P \setminus (P' \cup \{w\})) > H(P' \setminus P)$ or $H(w) + L(P \setminus (P' \cup \{w\})) = H(P' \setminus P)$ and 
	$P' \prec P$.
\end{proof}

The next lemma intuitively says that, if there exist subdomains where witnesses of solutions that are selectable but not strongly selectable can be excluded (included, respectively) when they reveal their type to be the lowest (highest, respectively) possible, then there is no implementation tree which yields an optimal OSP mechanism. Its proof uses ideas similar to that of Lemma \ref{lem:bin2ter_cond1}.

\begin{lemma}
	\label{lem:bin2ter_cond2}
	There is no optimal OSP mechanism for a set system problem
	if there is a subdomain $\tD$ of $\D$
	such that the following properties are both satisfied:
	\begin{itemize}
		\item[(i)] the set $\mathcal{S}$ of selectable solutions  for $\tilde{\D}$ has size $|\mathcal{S}| \geq 2$, and there is at least one $P \in \mathcal{S}$ such that $P$ is not strongly selectable; %,
		\item[(ii)] for every $f$ for which there is at least one selectable solution to which it belongs and at least one selectable to which it does not belong (i.e., $f \in \bigcup_{(P,P') \in \mathcal{S} \times \mathcal{S}} P\setminus P'$) both the following are true:
		%     	{for every $P \in \mathcal{S}$ that is not strongly selectable, and every witness $f$ of $P$} both the following are true:
		\begin{itemize}
			\item there is $\notP_f \in \mathcal{S}$ such that $f \notin \notP_f$ and $\notP_f$ is selectable for $\bar{\D} = (\tD_{-f}, L(f,\tD))$;
			\item there is $\check{P}_f \in \mathcal{S}$ such that $f \in \check{P}_f$ and $\check{P}_f$ is selectable for $\check{\D} = (\tD_{-f}, H(f,\tD))$.
		\end{itemize} 
	\end{itemize}
\end{lemma}

\begin{proof}
	Assume by contradiction that there is a subdomain $\tD$ for which the conditions above are satisfied and yet there is an optimal OSP mechanism $\M$;  let us denote with $\mathcal{T}$ its implementation tree. Fix the subdomain to be this particular $\tD$.
	
	Let $\bar{\mathcal{S}} \subseteq \mathcal{S}$ be
	the subset of solutions that are not strongly selectable,
	and let $W(\bar{\mathcal{S}})$ the set of agents $f$ such that
	$f$ is a witness for some $P \in \bar{\mathcal{S}}$.
	% 	Moreover, for each $f \in W(\bar{\mathcal{S}})$
	% 	let $\notP_f$ and $\check{P}_f$ as above.
	Consider the first node $u \in \mathcal{T}$ in which
	$f \in \bigcup_{(P,P') \in \mathcal{S} \times \mathcal{S}} P\setminus P'$
	% {$f \in W(\bar{\mathcal{S}}) \cup \{e \in (\notP_f \setminus \check{P}_f) \cup (\check{P}_f \setminus \notP_f) \mid f \in W(\bar{\mathcal{S}})\}$}
	diverges between $L(f, \tD)$ and $H(f, \tD)$
	in the subtree compatible with
	the type of every agent $e \in \bigcup_{(P,P') \in \mathcal{S} \times \mathcal{S}} P\setminus P'$ being in $\td_e$
	and the type of every remaining agent $e$ being $H(e,\tilde{\D})$.
	Observe that since there is at least one solution that is not strongly selectable, then there is at least one witness $w$, and thus, 
	since from Observation~\ref{obs:witness} $w \in \bigcup_{(P,P') \in \mathcal{S} \times \mathcal{S}} P\setminus P'$ and $|\td_w| \geq 2$,
	there is at least one agent that can diverge at node $u$.
	We next prove that since the mechanism $\M$ is optimal,
	then such a node must exist.
	Indeed, let $P \in \mathcal{S}$ be a solution that is not strongly selectable,
	let $w$ be its witness, and let $P'$ be the solution whose existence is guaranteed by Observation \ref{obs:witness}.
	Consider  the following two bid profiles $\b$ and $\bar\b$ in $\tD$ defined as
	\begin{align*}
	b_e = \begin{cases}
	L(e,\tD) & \text{if $e \in P \setminus P'$} \\
	H(e,\tD) & \text{otherwise}
	\end{cases} && \text{ and } &&
	\bar b_e = \begin{cases}
	L(e,\tD) & \text{if $e \in P \setminus P'$ and $e \neq w$} \\
	H(e,\tD) & \text{otherwise}
	\end{cases}
	\end{align*}
	Since $P$ is selectable then $\M$ returns $P$ on input profile
	$\b$, but since it is not strongly selectable it returns $P'$ on input profile $P'$. Therefore $P = \M(\b) \neq \M(\bar \b)=P'$.
	Since $\b$ and $\bar \b$ differ only in the bids of agents $w \in W(\bar{\mathcal{S}}) \subseteq \bigcup_{(P,P') \in \mathcal{S} \times \mathcal{S}} P\setminus P'$, there must be some node $u\in \T$ in which $w$ diverges between $L(w,\tD)$ and $H(w,\tD)$.
	
	Now let $f$ be the agent that diverge at $u$ and let $\notP_f$ and $\check{P}_f$ as in the claim.
	We are going to show a negative OSP $2$-cycle for $f$.
	Consider profiles $\b^{(L)}$ and $\b^{(H)}$ as follows:
	\begin{align}
	%	b^{(H)}_g = H(g,\tD), && b^{(L)}_g = L(g,\tD), && \b_e^{(H)} = \b_e^{(L)} =L(e,\tD) \ \  \text{ for all } e \in P^* \cap Q^*
	%	\intertext{and}
	b_e^{(H)} = \begin{cases}
	H(f,\tD) & \text{for $e=f$} \\
	L(e,\tD) & \text{for $e \in \check{P}_f \setminus \notP_f$ and $e\neq f$} \\
	H(e,\tD) & \text{otherwise}
	\end{cases} &&
	b_e^{(L)} = \begin{cases}
	L(f,\tD) & \text{for $e=f$} \\
	L(e,\tD) & \text{for $e \in \notP_f \setminus \check{P}_f$} \\
	H(e,\tD) & \text{otherwise}
	\end{cases}
	\end{align}
	It is immediate to see that both $\b^{(L)}$ and $\b^{(H)}$ are compatible with node $u$.
	However, in $\b^{(L)}$, where agent $f$ has type $L(f, \tD)$,
	an optimal mechanism must select $\notP_f$ (and thus it does not select $f$),
	since it is selectable even if the type of $f$ is $L(f, \tD)$.
	In $\b^{(H)}$, where agent $f$ has type $H(f, \tD) > L(f, \tD)$,
	instead mechanism $\M$ selects $\check{P}_f$ (and thus $f$),
	since it is selectable even if the type of $f$ is $H(f, \tD)$. 
	Therefore, the cycle between $\b^{(L)}$ and $\b^{(H)}$ is negative.
\end{proof}

\subsection{The main result (and the mechanism)}
The two necessary conditions suggest that it is possible to design an optimal OSP mechanism when both the following properties are satisfied for \emph{every} subdomain $\tD$ containing more than one instance. When all selectable solutions are also strongly selectable then there is an $f$ such that every $P'$ with $f \notin P'$ ceases to be selectable if the type of $f$ is $L(f, \tD)$ (this is the negation of Lemma \ref{lem:bin2ter_cond1}).
Moreover, if there is at least one selectable solution that is not strongly selectable for $\tD$,
then there is $f$ such that either every $P'$ with $f \notin P'$ ceases to be selectable if the type of $f$ is $L(f, \tD)$, or every $P'$ with $f \in P'$ ceases to be selectable if the type of $f$ is $H(f, \tD)$ (this is the negation of Lemma \ref{lem:bin2ter_cond2}). We prove that these properties are indeed sufficient by proving that Algorithm~\ref{algo:M_1} can be augmented with payments to give rise to an optimal OSP mechanism, that we call $\Sm$. 

\begin{algorithm}[!t]
	\DontPrintSemicolon
	\let\oldnl\nl% Store \nl in \oldnl
	\newcommand{\nonl}{\renewcommand{\nl}{\let\nl\oldnl}}% Remove line number for one line
	% \SetKwFunction{strong}{strong}
	\SetKwProg{Fn}{}{\string:}{}
	\KwIn{$E, {\mathcal{F}}, \D$}
	\KwOut{An optimal solution}
	%\nonl \Fn{$\Sm(E, \mathcal{F},\D)${}}{
	Initialize $R$ with the set of $P \in \mathcal{F}$ that are not selectable for ${\D}$, $\mathcal{P}=\mathcal{F} \setminus R$ and set $\tD=\D$\;
	\While{$|R| < |\mathcal{F}| - 1$}{
		\While{there is $P \in \mathcal{P}$ that is not strongly selectable for $\tD$ \nllabel{line:strong}}{
			\uIf{$\exists f \in \bigcup_{P  \in \mathcal{P}} P$ s.t. every $P \in \mathcal{P}$, with $f \notin P$, is not selectable for $(\tD_{-f}, L(f, \tD))$}{
				Ask $f$ if her type is $L(f, \tD)$ \nllabel{line:divL1}\;
				\uIf{yes}{
					$\tD=(\tD_{-f}, L(f, \tD))$\; 
					Add to $R$ and remove from $\mathcal{P}$ every $P$ that is not selectable for $\tD$\;
				}
			}\Else{
				Pick $f \in \bigcup_{P  \in \mathcal{P}} P$ s.t. all $P \in \mathcal{P}$, with $f \in P$, are not selectable for $(\tD_{-f}, H(f, \tD))$\;
				Ask $f$ if her type is $H(f, \tD)$ \nllabel{line:divH}\;
				\uIf{yes}{
					$\tD=(\tD_{-f}, H(f, \tD))$\; 
					Add to $R$ and remove from $\mathcal{P}$ every $P$ that is not selectable for $\tD$\;
				}
			}
		}
		\uIf{$|R| < |\mathcal{F}| - 1$}{
			Pick $f \in \bigcup_{P\in\mathcal{P}} P$ s.t. every $P \in \mathcal{P}$, with $f \notin P$, are not selectable for $(\tD_{-f}, L(f, \tD))$\;
			Ask $f$ if her type is $L(f, \tD)$ \nllabel{line:divL2}\;
			\uIf{yes}{
				$\tD=(\tD_{-f}, L(f, \tD))$\; 
				Add to $R$ and remove from $\mathcal{P}$ every $P$ that is not selectable for $\tD$\;
			}
		}
	}
	Return the only solution in $\mathcal{P}$ \nllabel{line:unique}
	%}
	\caption{The implementation tree of the optimal algorithm for mechanism $\Sm$}
	\label{algo:M_1}
\end{algorithm}

\begin{theorem}\label{thm:iff_system}
	There is an optimal OSP mechanism for a set system problem
	with three-value domains
	if and only if there is no subdomain $\tD$
	such that conditions of Lemma~\ref{lem:bin2ter_cond1} or of Lemma~\ref{lem:bin2ter_cond2} hold.
\end{theorem}
\begin{proof}
	The ``only if'' direction follows from Lemma~\ref{lem:bin2ter_cond1} and Lemma~\ref{lem:bin2ter_cond2}.
	
	For the ``if'' direction, assume that
	for every subdomain of $\D$ the conditions of Lemmas~\ref{lem:bin2ter_cond1} and \ref{lem:bin2ter_cond2} {do not} hold. We have discussed above how this leads to the definition of Algorithm \ref{algo:M_1}. The algorithm incrementally removes from $\mathcal{F}$ the solutions that are found to be non-selectable for some subdomain (this is the set $R$ in the algorithm). The construction of the subdomains closely exploits the properties coming from the unsatisfied lemmas. Specifically, the algorithm looks for an agent $f$ we can safely ask for OSP-ness to diverge between their current $L(f)$ and $H(f)$; if $f$ reveals type $L(f)$, then she will be securely selected, or if she reveals type $H(f)$, then she will be never selected. By assumption, such an agent always exists; the algorithm makes the suitable queries. This query will update the current domain of agents, and thus the process can be repeated for the new domains. The algorithm continues until we are left with only one selectable solution, that will be returned. 
	
	Observe at each iteration either the set of $\mathcal{P}=\mathcal{F}\setminus R$ shrinks, or there is at least one agent whose type domain decreases in size. Hence, we eventually reach a point in which we either have one solution therein (and the algorithm stops) or the type of each agent has been revealed, i.e., $\tD$ is such that $\td_e=\{t_e\}$ for all $e \in E$. We now show that, also in this case, $|R|=\mathcal{|F|}-1$. Suppose that this is not the case and that $R$ contains two solutions $P$ and $P'$. Let $H(e) = L(e) = t_e$ for each $e \in E$. Observe that $L(P \setminus P') = H(P \setminus P')$ and $L(P' \setminus P) = H(P' \setminus P)$. Since $P$ and $P'$ are selectable for $\tD$, then 
	\[
	L(P' \setminus P) \leq H(P \setminus P')=L(P \setminus P') \leq H(P' \setminus P)=	L(P' \setminus P)
	\]
	and then the two solutions have the exact same cost. Thus, if $P \prec P'$, then only $P$ is selectable, otherwise, only $P'$ is selectable and one of them should have been in $R$.
	
	Moreover, we will never remove all solutions from $\mathcal{F}$ since, as noted above, for any possible subdomain, there is always a selectable solution. We can then conclude that $\Sm$ always returns a feasible solution. Moreover, it is immediate to check that  the mechanism returns an optimal solution if for each $e$, the type $t_e$ is compatible with the actions taken by agent $e$ during the execution of $\Sm$. Indeed, by Lemma \ref{le:inheritance} we know that the solutions removed for bigger subdomain, because they were not selectable, remain non-selectable for all the smaller domains. 
	
	We finally prove that $\Sm$ satisfies the $2$-cycle monotonicity,
	from which, by Theorem~\ref{thm:2cycle3}, it follows that $\Sm$ is OSP.
	First observe that for every agent $e$ and every node $u$
	at which agent $e$ interacts with the mechanism,
	given that $D_e(u)$ is the domain of $e$ compatible with the action previously taken by $e$,
	either $e$ is asked to diverge among type $H(e, D_e(u))$
	and $D_e(u) \setminus \{H(e, D_e(u))\}$ (at Line~\ref{line:divH}),
	or it is asked to diverge among type $L(e, D_e(u))$
	and $D_e(u) \setminus \{L(e, D_e(u))\}$ (at Line~\ref{line:divL1} or at Line~\ref{line:divL2}).
	In the first case, no $2$-cycle is possible,
	since whenever $e$ takes the action compatible with $H(e, D_e(u))$,
	then, by hypothesis, every solution to which $e$ belongs ceases to be selectable,
	and thus this agent is never selected.
	In the second case, no $2$-cycle is possible,
	since whenever $e$ takes the action compatible with $L(e, D_e(u))$,
	then, by hypothesis, every solution to which $e$ does not belong ceases to be selectable, and thus this agent is always selected.
\end{proof}

%\input{amd-set}

% \subsubsection{Two feasible solutions}
% 
% \subsubsection{Any number of feasible solutions}

% \subsection{Application to optimum for shortest path on simple networks}
% \input{opt-sp-binary-2}

%\section{Conclusions}
%\section{Conclusions}

\section{Conclusions}
We have focused on OSP mechanisms, a compelling and needed notion of incentive compatibility for bounded rationality; \citet{li2015obviously} proves that OSP is the right and only notion of strategyproofness for agents who lack contingent reasoning skills. It is thus paramount to understand the limitations and the power of these mechanisms. 

We have introduced a new technique to look at OSP mechanisms, and shown its power by giving tight results on the approximation and a characterization of these mechanisms for paradigmatic problems in the area. Our contribution highlights how there are two dimensions, algorithms and their implementation, to the design of these mechanisms. The interplay between these dimensions is  encapsulated by OSP CMON and plays a central role, as shown by the limitations of fixing the implementation beforehand (as in DA auctions or direct revelation mechanisms).

We leave a number of open problems. A technical one is about the domain size and the difference between 2-cycles and longer ones; to what extent adding an extra type in the domain can deteriorate the approximation ratio of OSP mechanisms? A second, more conceptual question, is about dealing with multi-parameter agents. It does not seem immediate to characterize the implementation trees for these kind of agents as there is not a concept of relative ordering of types. %\car{Paolo what about comparable agents as in icalp06? Vogliamo dirlo/citarlo?}

%Despite the generality of the theorem above, it is almost immediate to derive from it
%simpler conditions for specific settings. As an example,
%in Appendix~\ref{apx:shortest_path} we show how to apply it to characterize the cases in which an optimal OSP mechanism exists for the shortest path problem when edges
%have two-value domains $D_e=\{L,H\}$ and there are only two parallel paths.

%\begin{acks}
%	
%This is where we put acknowledgements.
%\end{acks}

% Bibliography
\bibliographystyle{abbrvnat}
\bibliography{osp}

\newpage
\appendix

\section{Postponed proofs about OSP CMON}
\subsection{Extending CMON to characterize OSP} 
We here prove Theorem \ref{thm:cmon}.
To this aim, let us first formally prove Observation~\ref{obs:def}.
\begin{proof}[Proof of Observation~\ref{obs:def}]
	According to our definition, a mechanism $\M$ with implementation tree $\T$ is an OSP mechanism for the social choice function $f$
	if and only if there exists a payment function $\p$ such that $\M=(f, \p)$ and, for every agent $i$ with real type $t_i$, and for every vertex $u$ such that $i = S(u)$, it holds that 
	\[
	  p_i(t_i, \bi) - t_i(f(t_i, \bi)) = u_i(t_i, \M(t_i, \bi)) \geq u_i(t_i,\M(b_i, \bi')) =   p_i(b_i, \bi') - t_i(f(b_i, \bi'))
	\]
	for every $\bi, \bi'$ and for every $b_i \in D_i$, with $b_i \neq t_i$,
	such that $(t_i, \bi)$ and $(b_i, \bi')$ are compatible with $u$, but diverge at $u$. Since the true type of $i$ can be any value in $D_i$, then the mechanism is OSP if and only if this is true for any pair $t_i, b_i \in D_i$. 
\end{proof}
Henceforth, we will refer to condition~\eqref{eq:osp_constraint} as \emph{OSP constraint}.
	\begin{proof}[Proof of Theorem~\ref{thm:cmon}]
		We initially show that if $\ver$ has no negative-weight cycles then there exist payments $\p$ such that $\M=(f, \p)$ is OSP with implementation tree $\T$.
		Fix $i$ and $\T$ and consider the graph $\ver$. %Use $P^a$ as a shorthand for $P(a,\bi)$ for any declaration $a \in D_i$. 
		The idea is that an edge 
		$((a,\ai),(b, \bi))$, for $a, b
		\in D_i$, $a \neq b$, and $\ai, \bi \in D_{-i}(u)$,  $u$ being an $ab$-separating vertex for player $i$, in $\ver$ encodes the OSP constraint $P_i(\b) - P_i(\a) \leq a(f(b, \bi))-a(f(a, \ai))$ where, $\a$ ($\b$, respectively) is a shorthand for $(a,\ai)$ ($(b, \bi)$, respectively). 
		
		Augment $\ver$ with a
		node $\omega$ and edges $(\omega, \b)$ for any $\b \in D$ each of
		weight $0$. Observe that $\omega$ does not belong to any
		cycle of the augmented $\ver$ since it has outgoing edges only.
		Therefore we can focus our attention on cycles of $\ver$. For any $\b$ in $D$, we let
		$\shp(\omega, \b)$ be the length of the shortest path from $\omega$ to $\b$ in the
		augmented \vgraph. Since $\ver$ is finite and does not have
		negative-weight cycles then $\shp(\omega,\b)$ is well defined. It suffices, for
		all $\b \in D$, to set $P_i(\b) = \shp(\omega, \b)$. Indeed, consider
		any edge $(\a,\b)$ in $\ver$. The fact that the shortest path from $\omega$ to $\a$ followed by the edge $(\a,\b)$ is a path form $\omega$ to $\b$ implies, by shortest path definition, that $\shp(\omega,\b) \leq \shp(\omega,\a) + a(f(b, \bi))-a(f(a, \ai))$. We can the conclude that for all $a, b
		\in D_i$, $a \neq b$, and $\ai, \bi \in D_{-i}(u)$, $P_i(\b) - P_i(\a) = \shp(\omega,\b) - \shp(\omega,\a) \leq
		a(f(b, \bi))-a(f(a, \ai))$. %The proof of this part of the theorem concludes by looking at the $\ai \in  D_{-i}(u)$ that minimizes $P_i(\a) + a(f(\a))$ and the $\bi \in D_{-i}(u)$ that maximizes $P_i(\b) + a(f(\b))$.
		
		Now we assume that the \vgraph\ $\ver$ contains a negative-weight
		cycle. %, then $f$ is not OSP with money. 
		Let
		the negative-weight cycle be $C=\a^1 \rightarrow \a^2 \rightarrow
		\ldots \rightarrow \a^k \rightarrow \a^{k+1} = \a^1$. Cycle $C$ corresponds to
		the following OSP constraints:\footnote{Note that by the existence of $(\a^{j-1}, \a^{j})$ and $(\a^{j}, \a^{j+1})$, , $1 < j \leq k$, we know that $\ai^{j} \in D_{-i}(u^i_{a_i^{j-1},a_i^j}) \cap D_{-i}(u^i_{a_i^j,a_i^{j+1}})$ for opportune separating vertices $u^i_{a_i^{j-1},a_i^j}$ and $u^i_{a_i^j,a_i^{j+1}}$.}
		\begin{eqnarray*}
			P_i(\a^2) - P_i{(\a^1)} & \leq & a^1_i(f(\a^2)) - a^1_i(f(\a^1)), \\
			P_i(\a^3) - P_i{(\a^2)} & \leq & a^2_i(f(\a^3)) - a^2_i(f(\a^2)), \\
			\ldots \\
			P_i(\a^k) - P_i{(\a^{k-1})} & \leq & a^{k-1}_i(f(\a^{k})) - a^{k-1}_i(f(\a^{k-1})), \\
			P_i(\a^1) - P_i{(\a^{k})} & \leq & a^{k}_i(f(\a^{1})) - a^{k}_i(f(\a^k)).
		\end{eqnarray*}
		Suppose that there is a solution for the $P$'s satisfying each of
		these $k$ inequalities. This solution must also satisfy the
		inequality that results when we sum the $k$ inequalities together.
		If we sum the left-hand sides, each unknown $P^j$ is added in once
		and subtracted out once, so that the sum of the left-hand side is
		$0$. The right-hand side sums to the cycle weight $w(C)$, and thus
		we obtain $0 \leq w(C)$. But since $C$ is a negative-weight cycle,
		$w(C) < 0$, and we obtain the contradiction that $0 \leq w(C) < 0$.
	\end{proof}
	
\subsection{OSP two-cycle monotonicity and OSP CMON}\label{app:CMON-characterizes-OSP}
We here give the full proofs of Theorems \ref{thm:2cycle3} and \ref{thm:neg_res}.

\begin{proof}[Proof of Theorem~\ref{thm:2cycle3}]
 One direction follows from Theorem~\ref{thm:cmon}. As for the other direction,
 we prove that OSP two-cycle monotonicity implies OSP CMON.
 
 Fix agent $i$ and implementation tree $\T$, and consider a cycle $\mathcal{C} = (\x, \x', \x'', \ldots, \x)$ in the graph $\ver$.
 Since $\mathcal{C}$ is a cycle, we can assume without loss of generality that $x_i = L_i$
 if a profile exists in the cycle such that the type of $i$ is $L_i$, and $x_i = M_i$ otherwise.
 Observe that, since no edge exists among two profiles $\x$ and $\y$ such that $x_i = y_i$,
 then it must be the case that every cycle $\mathcal{C}$ can be partitioned in paths $(\mathcal{P}_1, \ldots, \mathcal{P}_t)$ such that
 for every $j = 1, \ldots, t$, $\mathcal{P}_j$ is as one of the following:
 \begin{enumerate}[noitemsep,leftmargin=20pt]
  \item\label{item:len2} $\mathcal{P}_j = (\x^{j-1},\y,\x^{j})$, where $x^{j-1}_i = x^{j}_i < y_i$;
  \item\label{item:len3} $\mathcal{P}_j = (\x^{j-1},\y,\z,\x^{j})$, where $x^{j-1}_i = x^{j}_i = L_i$, but $y_i \notin \{x^j_i, z_i\}$ and $z_i \notin \{x^j_i, y_i\}$;
  \item\label{item:leneven} $\mathcal{P}_j = (\x^{j-1}, \mathcal{P}_j^1, \ldots, \mathcal{P}_j^s, \x^{j})$, where $s \geq 1$,
  $x^{j-1}_i = x^{j}_i = L_i$ and $\mathcal{P}_j^h = (\y^{j,h-1}, \z^{j,h}, \y^{j,h})$, where either $y^{j,h-1}_i = y^{j,h}_i = M_i$ and $z^{j,h}_i = H_i$ for every $h \in \{1, \ldots, s\}$
  or $y^{j,h-1}_i = y^{j,h}_i = H_i$ and $z^{j,h}_i = M_i$ for every $h \in \{1, \ldots, s\}$;
  \item\label{item:lenodd} $\mathcal{P}_j = (\x^{j-1}, \mathcal{P}_j^1, \ldots, \mathcal{P}_j^s, \w^{j}, \x^{j})$, where $s \geq 1$,
  $x^{j-1}_i = x^{j}_i = L_i$, $w^{j}_i \neq L_i$ and, for every $h \in \{1, \ldots, s\}$, $\mathcal{P}_j^h = (\y^{j,h-1}, \z^{j,h}, \y^{j,h})$, where $y^{j,h-1}_i = y^{j,h}_i \notin \{L_i, w^{j}_i\}$ and $z^{j,h}_i = w^{j}_i$.
 \end{enumerate}
 Indeed, if no profile appears in the cycle such that the type of $i$ is $L_i$, then the only way of going from a profile $\x^{j-1}$ with $x^{j-1}_i = M_i$ to a profile $\x^j$ with $x^j_i = M_i$ is through a profile $\y$ such that $y_i = H_i$, and this is considered in case~\ref{item:len2}.
 As for cycles in which there is a profile such that the type of $i$ is $L_i$, then either we have a path of length 2 as considered in case~\ref{item:len2},
 or a path of length 3 as considered in case~\ref{item:len3},
 or a path in which the internal profiles are such that the type of $i$ is alternatively $M_i$ and $H_i$.
 In turn, for this last case we can distinguish two subcases:
 either the type of $i$ in the last internal profile is the same as in the first internal profile, as considered in case~\ref{item:leneven},
 or they are different, as considered instead in case~\ref{item:lenodd}.

 Since we are focusing on single-parameter settings, the cost of a path of type~\ref{item:len2} is
 $
  C(\mathcal{P}_j) = \Big(x^{j-1}_i \cdot f_i(\x^{j-1}) - x^{j-1}_i \cdot f_i(\y)\Big) + \Big(y_i \cdot f_i(\y) - y_i \cdot f_i(\x^j)\Big).
 $
%  for some $f_i \geq 0$.
 This cost can be written as follows
 \begin{align*}
  C(\mathcal{P}_j) & = y_i (f_i(\y) - f_i(\x^{j-1}) + f_i(\x^{j-1}) - f_i(\x^j)) - x^{j-1}_i (f_i(\y) - f_i(\x^{j-1}))\\
  & = (y_i-x^{j-1}_i) (f_i(\y) - f_i(\x^{j-1})) + y_i (f_i(\x^{j-1}) - f_i(\x^j))\\
  & \geq M_i (f_i(\x^{j-1}) - f_i(\x^j)),
 \end{align*}
 where we used that $y_i > x^{j-1}_i \geq L_i$ and that $f_i(\y) \geq f_i(\x^{j-1})$ by two-cycle monotonicity.

 Consider now paths of type~\ref{item:len3}.
%  Since we are focusing on single-parameter settings, 
 The cost of this path is
 $
  C(\mathcal{P}_j) = \Big(x^{j-1}_i \cdot f_i(\x^{j-1}) - x^{j-1}_i \cdot f_i(\y)\Big) + \Big(y_i \cdot f_i(\y) - y_i \cdot f_i(\z)\Big) + \Big(z_i \cdot f_i(\z) - z_i \cdot f_i(\x^j)\Big).
 $
 We can rewrite this cost as follows:
 \begin{align*}
  C(\mathcal{P}_j) & = - x^{j-1}_i(f_i(\y)-f_i(\x^{j-1})) + y_i(f_i(\y)-f_i(\z))\\
  & \qquad + z_i(f_i(\z) - f_i(\y) + f_i(\y) - f_i(\x^{j-1}) + f_i(\x^{j-1}) - f_i(\x^j))\\
  & \geq (z_i - x^{j-1}_i)(f_i(\y)-f_i(\x^{j-1})) + (z_i - y_i)(f_i(\z)-f_i(\y)) + z_i(f_i(\x^{j-1}) - f_i(\x^j))\\
  & \geq M_i (f_i(\x^{j-1}) - f_i(\x^j)),
 \end{align*}
 where we used that $z_i \geq M_i > L_i = x^{j-1}_i$ and, by two-cycle monotonicity, $f_i(\y) \geq f_i(\x^{j-1})$ (since $y_i > x^{j-1}_i$),
 and either $f_i(\z) - f_i(\y) = 0$ or $\mathtt{sign}(f_i(\z) - f_i(\y)) = \mathtt{sign}(z_i - y_i)$.

 For paths of type~\ref{item:leneven}, we have that:
 \begin{equation}
  \label{eq:path}
  C(\mathcal{P}_j) = x^{j-1}_i \left(f_i(\x^{j-1}) - f_i(\y^{j,0})\right) + \sum_{h=1}^s C(\mathcal{P}_j^h) + y^{j,s}_i \left(f_i(\y^{j,s}) - f_i(\x^j)\right),
 \end{equation}
 where
 \begin{equation}
  \label{eq:subpath}
  C(\mathcal{P}_j^h) = y^{j,h-1}_i \left(f_i(\y^{j,h-1}) - f_i(\z^{j,h})\right) + z^{j,h}_i \left(f_i(\z^{j,h}) - f_i(\y^{j,h})\right).
 \end{equation}
 Observe that
%  \begin{equation}
  \begin{align}
   & f_i(\y^{j,s}) - f_i(\x^j) = f_i(\y^{j,s})  - f_i(\x^j) - \sum_{h=1}^s \left(f_i(\z^{j,h}) - f_i(\z^{j,h}) + f_i(\y^{j,h-1}) - f_i(\y^{j,h-1})\right) = \label{eq:last_step}\\
  & \sum_{h=1}^s \left(f_i(\y^{j,h}) - f_i(\z^{j,h})\right) + \sum_{h=1}^s \left(f_i(\z^{j,h}) - f_i(\y^{j,h-1})\right) + \left(f_i(\y^{j,0}) - f_i(\x^{j-1})\right) + \left(f_i(\x^{j-1}) - f_i(\x^j)\right). \nonumber
  \end{align}
%  \end{equation}
 By plugging \eqref{eq:subpath} and \eqref{eq:last_step} into \eqref{eq:path}, we have that
 \begin{align*}
  C(\mathcal{P}_j) & = \sum_{h=1}^s (y_i^{j,s} - z_i^{j,h}) \left(f_i(\y^{j,h}) - f_i(\z^{j,h})\right) + \sum_{h=1}^s (y_i^{j,s} - y_i^{j,h-1}) \left(f_i(\z^{j,h}) - f_i(\y^{j,h-1})\right)\\
  & \qquad + (y_i^{j,s} - x^{j-1}_i) \left(f_i(\y^{j,0}) - f_i(\x^{j-1})\right) + y_i^{j,s} \left(f_i(\x^{j-1}) - f_i(\x^j)\right)\\
  & = \sum_{h=1}^s (y_i^{j,h} - z_i^{j,h}) \left(f_i(\y^{j,h}) - f_i(\z^{j,h})\right) + (y_i^{j,0} - x^{j-1}_i) \left(f_i(\y^{j,0}) - f_i(\x^{j-1})\right)\\
  & \qquad + y_i^{j,s} \left(f_i(\x^{j-1}) - f_i(\x^j)\right),
 \end{align*}
 where we used the fact that the $y_i^{j,h}$'s are the same for every $h$.
 Since by two-cycle monotonicity, we have that $f_i(\y^{j,0}) \geq f_i(\x^{j-1})$ and either $f_i(\y^{j,h}) - f_i(\z^{j,h}) = 0$ or $\mathtt{sign}(f_i(\y^{j,h}) - f_i(\z^{j,h})) = \mathtt{sign}(y_i^{j,h} - z_i^{j,h})$,
 we can conclude that
 $
  C(\mathcal{P}_j) \geq y_i^{j,s} \left(f_i(\x^{j-1}) - f_i(\x^j)\right) \geq M_i \left(f_i(\x^{j-1}) - f_i(\x^j)\right).
 $

 Similarly, for paths of type~\ref{item:lenodd}, we have that:
 $
  C(\mathcal{P}_j) = x^{j-1}_i \left(f_i(\x^{j-1}) - f_i(\y^{j,0})\right) + \sum_{h=1}^s C(\mathcal{P}_j^h) + y^{j,s}_i \left(f_i(\y^{j,s}) - f_i(\w^j)\right) + w^j_i \left(f_i(\w^j) - f_i(\x^j)\right).
 $
 As above, the last step can be rewritten as follows:
 \begin{align*}
  f_i(\w^j) - f_i(\x^j) & = f_i(\w^j)  - f_i(\x^j)\\
  & \qquad - \sum_{h=1}^s \left(f_i(\y^{j,h}) - f_i(\y^{j,h}) + f_i(\z^{j,h}) - f_i(\z^{j,h})\right) - \left(f_i(\y^{j,0}) - f_i(\y^{j,0})\right)\\
  &  = \left(f_i(\w^j) - f_i(\y^{j,s})\right) + \sum_{h=1}^s \left(f_i(\y^{j,h}) - f_i(\z^{j,h})\right) + \sum_{h=1}^s \left(f_i(\z^{j,h}) - f_i(\y^{j,h-1})\right)\\
  & \qquad + \left(f_i(\y^{j,0}) - f_i(\x^{j-1})\right) + \left(f_i(\x^{j-1}) - f_i(\x^j)\right).
  \end{align*}
 Hence,
 \begin{align*}
  C(\mathcal{P}_j) & = \sum_{h=1}^s (w^j_i - z_i^{j,h}) \left(f_i(\y^{j,h}) - f_i(\z^{j,h})\right) + \sum_{h=1}^s (w^j_i - y_i^{j,h-1}) \left(f_i(\z^{j,h}) - f_i(\y^{j,h-1}))\right)\\
  & \quad + (w^j_i - y^{j,s}_i) \left(f_i(\w^j) - f_i(\y^{j,s})\right) + (w^j_i - x^{j-1}_i) \left(f_i(\y^{j,0}) - f_i(\x^{j-1})\right) + w^j_i \left(f_i(\x^{j-1}) - f_i(\x^j)\right)\\
  & = \sum_{h=1}^s (z^{j,h}_i - y_i^{j,h-1}) \left(f_i(\z^{j,h}) - f_i(\y^{j,h-1}))\right) + (w^j_i - y^{j,s}_i) \left(f_i(\w^j) - f_i(\y^{j,s})\right)\\
  & \quad + (y_i^{j,0} - x^{j-1}_i) \left(f_i(\y^{j,0}) - f_i(\x^{j-1})\right) + w^j_i \left(f_i(\x^{j-1}) - f_i(\x^j)\right),
 \end{align*}
 where we used that $z_i^{j,h} = w^j_i$ for every $h$.
 The, by two-cycle monotonicity, we have that
 $
  C(\mathcal{P}_j) \geq w^j_i \left(f_i(\x^{j-1}) - f_i(\x^j)\right) \geq M_i \left(f_i(\x^{j-1}) - f_i(\x^j)\right).
 $

 Finally, the cost of the cycle is the sum of the costs of paths in which this cycle has been partitioned.
 That is,
 $
  C(\mathcal{C}) = \sum_{j=1}^{t} C(\mathcal{P}_j) \geq M_i \sum_{j=1}^{t} (f_i(\x^{j-1}) - f_i(\x^j)) = M_i (f_i(\x^0) - f_i(\x^t)) = 0.
 $
 The theorem then follows.
\end{proof}

\begin{proof}[Proof of Theorem~\ref{thm:neg_res}]
 Let $D_1=\{H,M,B,L\}$ and $D_j=\{H,M,L\}$ for every $j \neq 1$, with $H > M > B > L$.
 and consider a social function $f$ aiming to select the agent $i$ with the smallest type.
 This can be seen as a procurement auction for a single item, or, alternatively, if types are negative,
 as a classical auction setting for a single item, in which we would like to select the agent with the maximum valuation for the item.
 
 Consider the implementation tree $\T$ corresponding to the following algorithm:
 keep a set of candidates, initialized to the set of all agents;
 for every $t \in D_1$ in decreasing order,
 first ask agent $1$ if her type is $t$,
 and then, if $b \neq t$, repeat the question to the remaining agents (in some order);
 if an agent gives a positive answer, then remove this agent from the set of candidates;
 this process continues until only two candidates are left;
 when this occurs, then, for every $t \in D_i$ in increasing order
 ask the candidates left whether their type is $t$, by giving precedence to agent $1$ if she is still a candidate;
 if an agent gives a positive answer, then assigns the item to this agent.
 In other words, we run a descending auction until we are left with only two candidates,
 and then we run an ascending auction to choose the winner among these two candidates.
 
 It is not hard to see that a mechanism that follows this implementation tree satisfies 2-cycle monotonicity.
 Indeed, during the descending phase, if the agent gives a positive answer, then she is never assigned the item.
 Instead, during the ascending phase, if the agent gives a positive answer, then she is always assigned the item.
 (For a more formal proof of 2-cycle monotonicity, we refer the reader to the proof of Proposition~\ref{prop:algo2mon2},
 where this property is proved for a generalization of the mechanism above.)

 We now consider a cycle of length four in the graph ${\mathcal O}_1^{\mathcal{T}}$
 involving the profiles $\x, \y, \z, \w$ defined as follows:
 in $\x$ we have $x_1 = H$ and $x_j < H$ for every $j \neq 1$;
 in $\w$ we have $w_1 = M$ and $w_j = H$ for every $j \neq 1$;
 in $\z$ we have $z_1 = L$ and $z_j \geq M$ for every $j \neq 1$;
 in $\y$ we have $y_1 = B$, there is $k \neq 1$ such that $y_k = L$, and $y_j \geq L$  for every $j \neq 1,k$.

 Observe that when agent $1$ is queried about type $H$, then by answering \emph{yes} the outcome $\x$ can occur,
 whereas $\w, \z, \y$ can all occur if agent $1$ answers \emph{no}. Hence, in the declaration graph of agent $1$
 edges $(\x,\w)$, $(\w,\x)$, $(\x, \z)$, $(\z, \x)$, $(\x, \y)$ and $(\y, \x)$ exist.
 Moreover, if agent $1$ is queried about type $L$, then by answering \emph{yes} the outcome $\z$ can occur,
 whereas $\w$ and $\y$ can occur if agent $1$ answers \emph{no}. Hence, in the declaration graph of agent $1$
 edges $(\z,\w)$, $(\w,\z)$, $(\z, \y)$ and $(\y, \z)$ exist.
 It is immediate to check that no other edge exists among these four profiles.
 Indeed, when agent $1$ is queried about $B$ ($M$) neither $\x, \z$ are possible since agent $1$ is queried about $L$ and $H$ before being queried about $B$ ($M$),
 nor $\w$ since machine $k$ has revealed type $L$ before agent $1$ is queried about $B$ ($M$).

 Observe that the algorithm assigns the item to agent $1$ in $\w$ and $\z$, while agent $1$ does not receive the item in $\x$ and $\y$.
 Hence, edges $(\x,\y)$, $(\y,\x)$, $(\w, \z)$ and $(\z, \w)$ have weight $0$, and weights for other edges are as in Figure~\ref{dec_graph}.
 \begin{figure}[htbp]
 \begin{center}
\begin{tikzpicture}[shorten >=1pt,node distance=3cm]
\tikzstyle{place}=[circle,draw,thick]
\node[place] (x)  		     {$\x$};
\node[place] (y)   [right of=x]  {$\w$};
\node[place] (z) [below of=y]  {$\z$};
\node[place] (w)  [below of=x]  {$\y$};
\path[->] (x.east) edge[bend left] node [label=above:{$H$}] {} (y);
\path[->] (y) edge[bend left] node [label=above:{$-M$}] {} (x.east);
\path[->] (x.south east) edge[bend left] node [label=below:{$H$}] {} (z.north west);
\path[->] (z.north west) edge[bend left] node [label=above:{$-L$}] {} (x.south east);
\path[->] (x.south) edge[bend left] node [label=left:{$0$}] {} (w);
\path[->] (w) edge[bend left] node [label=left:{$0$}] {} (x.south);
\path[->] (z.north) edge[bend left] node [label=right:{$0$}] {} (y);
\path[->] (y) edge[bend left] node [label=right:{$0$}] {} (z.north);
\path[->] (z.west) edge[bend left] node [label=below:{$-L$}] {} (w);
\path[->] (w) edge[bend left] node [label=below:{$B$}] {} (z.west);
\end{tikzpicture}
 \end{center}
\caption{Extract of the declaration graph of $i$}
 \label{dec_graph}
 \end{figure}
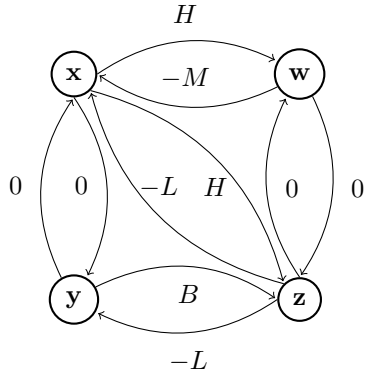

Consider then the cycle $(\x,\y,\z, \w)$. Its weight is $B - M < 0$, from which the claim follows.
\end{proof}

\section{Postponed material about machine scheduling}
\label{apx:scheduling}

\subsection{Characterizing the case of two jobs and two agents}
The next result then gives a complete characterization of the approximability for two jobs and two agents in the three-value domains.
	\begin{theorem}\label{th:schduling:three-values:ub}
		For the machine scheduling problem, with two jobs and two agents with three values domains $D_i=\{L,M,H\}$, there is an optimal OSP mechanism if and only if $M \leq 2L$ or $H \leq 2M$.
% 		the domain satisfies the following condition: 
% 		\begin{align}\label{eq:3valdom-easy}
% 		M \leq 2L && \text{ or } && H \leq  2M\ . 
% 		\end{align}
	\end{theorem}
	\begin{proof}%[Proof of Theorem~\ref{th:schduling:three-values:ub}]
% 		We  show that if \eqref{eq:3valdom-easy} holds, then there is a simple exact OSP  mechanism.
		In the sequel, we write $x\approx y$ if $x<y\leq 2x$, and $x\prec y$ if $x< 2x < y$. With this notation, we require that $L\approx M$ or $M \approx H$. 
		The optimal schedule is as follows, depending on each possible case:
		\begin{itemize}
			\item %($M < 2L$ and  $H > 2M$) 
			($L \approx M  \prec H$):
			If one agent has high cost and the other does not, schedule both jobs on the agent with lower cost. The OSP mechanism is depicted in Figure~\ref{fig:OSP_three_val_case3}, where leaf nodes show the allocation and the payment of agent $1$: The constants $p_{-}$ and $p_{+}$ satisfy
			$L < p_- < M < p_+ < H$
			and $2p_- \geq L + p_+$. As for agent $2$,  we pay $p_+$ for each assigned job.
			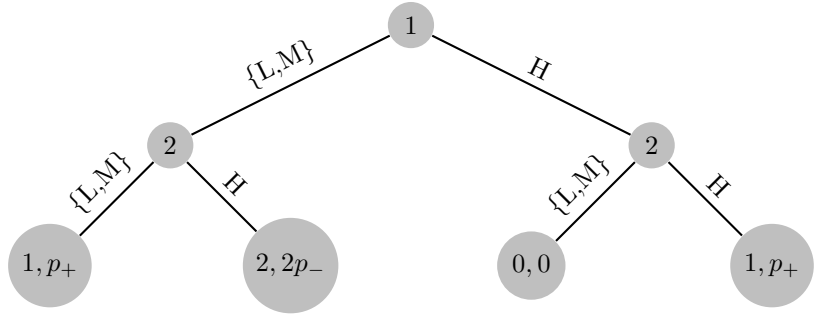
\begin{figure}[htbp]
	\centering
	\begin{tikzpicture}[scale=.8]
	\treeedges{\{L,M\}}{H}{\{L,M\}}{H}{\{L,M\}}{H}
	\tree{1, p_{+}}{2, 2p_{-}}{0,0}{1,p_{+}}
	\end{tikzpicture}
	\caption{An exact OSP mechanism for scheduling two identical jobs when $M \leq 2L$ and $H > 2M$. 
		%$L > \frac{M}{2} \land M \leq \frac{H}{2}$.
		}
	\label{fig:OSP_three_val_case3}
\end{figure}
			
			\item %($L \leq \frac{M}{2}$ and $M > \frac{H}{2}$) 
			($L \prec M \approx H$): If one agent has low cost and the other does not, schedule both jobs on the agent with low cost. The OSP mechanism is shown in Figure~\ref{fig:OSP_three_val_case2}, with payments and allocations being described as in the previous case, except that for agent $2$  we pay $p_-$ for each assigned job.
			
			\begin{figure}[htbp]
	\centering
	\begin{tikzpicture}[scale=.8]
	\treeedges{L}{\{M,H\}}{L}{\{M,H\}}{L}{\{M,H\}}
	\tree{1, p_{+}}{2, 2p_{-}}{0,0}{1,p_{+}}
	\end{tikzpicture}
	\caption{An exact OSP mechanism for scheduling two identical jobs when $M > 2L$ and $H \leq 2M$.}\label{fig:OSP_three_val_case2}\label{fig:scheduling:three:1}
\end{figure}
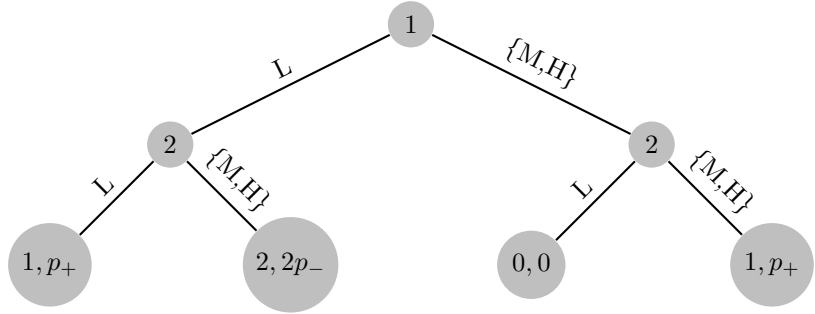

			\item %$(M < 2L) \land (M > \frac{H}{2}) \land (L \leq \frac{H}{2}$) 
			($L \approx M  \approx H$): Here we distinguish two subcases. If $2L \geq H$, then we always schedule one job per machine. Otherwise, for $2L < H$, we only need to distinguish if one agent has cost $L$ and the other has cost $H$ (in which case both jobs go to the low cost agent). The mechanism is shown in Figure~\ref{fig:OSP_three_val_case4}, with payments and allocations for agent 1 being described as in the previous cases. As for agent 2, her payment per unit of work is $p_+$ when separating $\{L,M\}$ from $H$ (left subtree), and $p_-$ when separating $L$ from $\{M,H\}$ (right subtree).			
			\begin{figure}[htpb]
	\centering
	\begin{tikzpicture}[scale=.8]
	\treeedges{L}{H}{\{L,M\}}{H}{L}{\{M,H\}}
	\tree{1, p_{+}}{2, 2p_{-}}{0,0}{1,p_{+}}
	\node[circle, fill=gray!50] (nodeM) at (6,2) {$1,p_+$};
	\draw[thick] (R) -- (nodeM) node[pos=.5,sloped,above] {M};
	\end{tikzpicture}
	\caption{An exact OSP mechanism for scheduling two identical jobs when $M \leq 2L$, $H \leq 2M$, and $H > 2L$.}
	\label{fig:OSP_three_val_case4}
\end{figure}
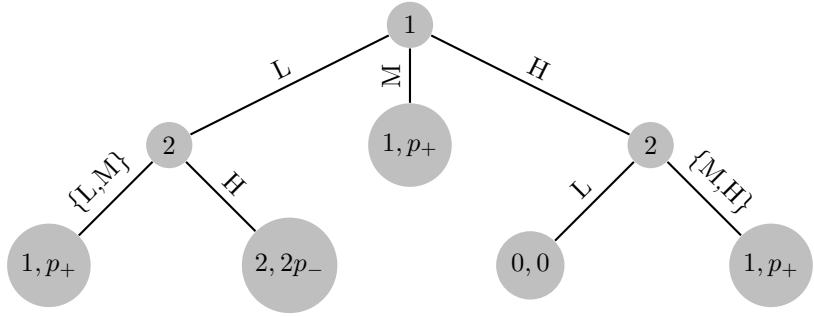 \qedhere
		\end{itemize}
% 		This completes the proof. 
	\end{proof}

\renewcommand{\L}{\vett{L}}
\renewcommand{\H}{\vett{H}}
\newcommand{\Opt}{{GR}}

\subsection{The case of heterogeneous two-value domains (proof of Theorem \ref{th:scheduling:many-agents:identical-jobs-})}
\label{apx:2het}
In this section, we consider machine scheduling on \emph{heterogeneous} two-value domains $D_i = \{L_i,H_i\}$, where $L_i\leq H_i$ for each agent $i$. In particular, we give an optimal polynomial-time OSP mechanism for this case (thus proving Theorem~\ref{th:scheduling:many-agents:identical-jobs-}). %thus extending the result in Theorem~\ref{prop:scheduling:many-agents:identical-jobs}. 
\paragraph{Notation.}
We denote the \emph{makespan} of a solution (scheduling) $\x$ with respect to a profile (machines reported costs) $\b$ by $MS(\x,\b) := \max_{i} b_i \cdot x_i$, that is, the maximum agent cost (machine completion time). For any subset $S$ of agents, and any profile $\b$, we denote by $\b_S$ the profile restricted to the agents in $S$ only, and by $\b_{-S}$ the profile restricted to the agents not in $S$. Moreover, $\L_S$ ($\H_S$, respectively) denotes the profile in which all agents in $S$ bid $L_i$ ($H_i$, respectively). Analogously, $\L_{-S}$ and $\H_{-S}$ are the profiles in which all agents \emph{not} in $S$ bid $L_i$ and $H_i$, respectively.

\paragraph{An optimal (greedy) algorithm.}
We next describe a simple optimal greedy algorithm $\Opt$ which 
will lead to an optimal OSP mechanism. 
\block{Algorithm $\Opt$: Allocate jobs one by one greedily, breaking \emph{ties} in favor of faster machines or, in case of same speed, in favor of lower index machines.}

It is well known that the greedy algorithm returns the optimal scheduling for {identical jobs}. 
In addition, this particular  tie-breaking rule will ensure that the algorithm returns the ``most balanced'' allocation. As stated in Fact~\ref{fac:balanced} below, if we  move a job from a 
machine to
another
one,
the makespan relative to these two machines only (the maximum cost between them) will increase.

%Consider a type profile $\b$ and assume that machines are sorted in increasing order of type,
%i.e.
%Let $\Opt$ be the algorithm, that given the bid of each machine, returns the optimal scheduling.
%We assume that $\Opt$, in case of ties, returns the most balanced outcome:
%i.e., for every two machines $i$ and $j$, 
%with $b_i < b_j$ or $b_i = b_j$ and $i < j$,
%we require that\footnote{The tie breaking rule here is specific for unitary jobs.}
%$b_i \cdot (\Opt_i(\b)+1) > b_j \cdot \Opt_j(\b)$
%% $\Opt_i(\b) \geq \Opt_j(\b)$
%and $b_j \cdot (\Opt_j(\b)+1) \geq b_i \cdot \Opt_i(\b)$.
%In other words,
%% faster machines take more jobs than slower ones, but
%it must not be possible to move a job from a 
%% faster
%machine to
%% a slower 
%another
%one,
%without increasing the makespan relative to these two machines.

\begin{fact}\label{fac:balanced}
	Algorithm $\Opt$ satisfies the following condition. For any two machines $i$ and $j$, 
	such that $b_i < b_j$, or $b_i = b_j$ and $i < j$,
	it holds that 
	%$b_i (\Opt_i(\b)+1) > b_j \cdot \Opt_j(\b)$
	%and $b_j \cdot (\Opt_j(\b)+1) \geq b_i \cdot \Opt_i(\b)$.
	\begin{align}
	\label{eq:greedy_cond}
	b_i \cdot (\Opt_i(\b)+1) > b_j \cdot \Opt_j(\b) && \text{
		and } && b_j \cdot (\Opt_j(\b)+1) \geq b_i \cdot \Opt_i(\b) \ .
	\end{align}
\end{fact}
\begin{proof}
	For every pair of machines $i$ and $j$ as above,
	$b_j \cdot (\Opt_j(\b)+1) \geq b_j \cdot (\Opt^t_j(\b)+1) \geq b_i \cdot (\Opt^t_i(\b) + 1) = b_i \cdot \Opt_i(\b)$,
	where $\Opt^t_j(\b)$ ($\Opt^t_i(\b)$, respectively) denotes the number of jobs assigned to machine $j$ ($i$, respectively) before step $t$ in which the greedy algorithm assigns the last job to $i$,
	and 
	$b_i \cdot (\Opt_i(\b)+1) \geq b_i \cdot (\Opt^\tau_i(\b)+1) > b_j \cdot (\Opt^\tau_j(\b) + 1) = b_j \cdot \Opt_j(\b)$, where $\tau$ is the time in which the greedy algorithm assigns the last job to $j$.
\end{proof}

Note also that it is impossible that both conditions in \eqref{eq:greedy_cond} are unsatisfied,
otherwise we have that $b_i \cdot (\Opt_i(\b)+1) \leq b_j \cdot \Opt_j(\b) < b_j \cdot (\Opt_j(\b)+1) < b_i \cdot \Opt_i(\b)$, thus  a contradiction (recall that $b_i>0$).

Moreover, if $b_j \cdot (\Opt_j(\b)+1) \geq b_i \cdot \Opt_i(\b)$,
then it is impossible to satisfy both conditions
even for the assignment $\x$ achieved from $\Opt(\b)$
by moving a job from $i$ to $j$: indeed, we have that
$b_i \cdot (x_i + 1) = b_i \cdot \Opt_i(\b) \leq b_j \cdot (\Opt_j(\b)+1) = b_j \cdot x_j$, so that the first condition is not satisfied.

Similarly, if $b_i \cdot (\Opt_i(\b)+1) > b_j \cdot \Opt_j(\b)$,
then it is impossible to satisfy both the conditions
even for the assignment $\y$ achieved from $\Opt(\b)$
by moving a job from $j$ to $i$: indeed, we have that
$b_j \cdot (y_j + 1) = b_j \cdot \Opt_j(\b) < b_i \cdot (\Opt_i(\b)+1) = b_i \cdot y_i$, so that the second condition is not satisfied.

\bigskip
The classical monotonicity condition says that, whenever an agent increases her own cost, her load  will either decrease or stay the same. The next lemma says that, in our algorithm $\Opt$, the load of any \emph{every} other agent will either increase or stay the same.

\begin{lemma}
	\label{lem:opt_prop}
	For any profile $\b$, and for any two distinct agents (machines) $i$ and $ j$, it holds that
	$$
	\Opt_i(\b_{-j},L_j) \leq \Opt_i(\b_{-j}, H_j).
	$$
\end{lemma}
\begin{proof}
	Let $\b^{L} = (\b_{-j},L_j)$ and $\b^{H}= (\b_{-j},H_j)$,
	and $\y^{L} = \Opt(\b_{-j},L_j)$ and $\y^{H} = \Opt(\b_{-j},H_j)$.
	Using standard argument (see e.g.  \citep{AT01}), we first prove that the algorithm is monotone, that is, $y^{L}_j \geq y^{H}_j$. 
	Suppose by contradiction that $y^{L}_j < y^{H}_j$.
	Observe that 
	\begin{align}\label{eq:makspan:inequalitites}
	MS(\y^{L}, \b^{L}) \le MS(\y^{H}, \b^{L})  \le MS(\y^{H}, \b^{H}) \le MS(\y^{L}, \b^{H}) \ ,
	\end{align}
	where the first and the last inequalities are due to the optimality of the algorithm, and the second inequality follows from the fact that $\b^{H}$ is obtained from $\b^{L}$ by increasing one machine cost. 
	We next distinguish two cases according to which machine determines the makespan $MS(\y^{L}, \b^{H})$. If this is machine $j$, meaning that $MS(\y^{L}, \b^{H}) = y^L_j H_j$, then \eqref{eq:makspan:inequalitites} and the hypothesis $y^{L}_j < y^{H}_j$ imply  
	\[MS(\y^{H}, \b^{H}) \le MS(\y^{L}, \b^{H}) = y^L_j H_j < y^H_j H_j\le  MS(\y^{H}, \b^{H})\]
	where the last inequality follows by definition of makespan, thus a contradiction. Otherwise, $MS(\y^{L}, \b^{H}) = y^L_i b^H_i$ for some machine $i\neq j$, which means that $b^H_i = b^L_i$ as $\b^L$ and $\b^H$ differ only in machine $j$. Therefore, $MS(\y^{L}, \b^{H}) = y^L_i b^H_i = y^L_i b^L_i \le MS(\y^{L}, \b^{L})$. This inequality,  combined with \eqref{eq:makspan:inequalitites}, implies that all inequalities in \eqref{eq:makspan:inequalitites} hold with `='. We  argue that this contradicts the fact that the algorithm breaks ties in a fixed manner, since 
	\begin{align}
	MS(\y^L,\b^H) = MS(\y^H,\b^H) && MS(\y^H,\b^L) = MS(\y^L,\b^L)
	\end{align}
	and thus $\y^L$ would also be optimal for $\b^H$ and $\y^H$ would also be optimal for $\b^L$.

	Since we assumed $y^{H}_j > y^{L}_j$, there must be another machine $k$ such that
	$y^{H}_k < y^{L}_k$.
	On input $\b^H$, by the definition of greedy algorithm implies
	\begin{equation}
	\label{eq:tie0}
	b_k \cdot (y^{H}_k + 1) \geq H_j \cdot y^{H}_j
	\end{equation}
	because otherwise the algorithm would have assigned the last job to machine $j$ to machine $k$ instead. 
	%However, \eqref{eq:4}, \eqref{eq:5}, and \eqref{eq:6}, implies also that
	%$MS(\y^{L},\b^{L}) = MS(\y^{H},\b^{L})$.
	Similarly, on input $\b^{L}$, the greedy algorithm must satisfy
	$L_j \cdot (y^{L}_j + 1) \geq b_k \cdot y^{L}_k$.
	Since $H_j \cdot y^{H}_j > L_j (y^{L}_j + 1)$ and $b_k \cdot y^{L}_k \geq b_k \cdot (y^{H}_k+1)$, we have that $b_k \cdot (y^{H}_k+1) \leq b_k y_k^L \leq L_j (y_j^L+1)< H_j \cdot y^{H}_j$, which contradicts \eqref{eq:tie0}.
	
	% That is, both $\y^{L}$ and $\y^{H}$ are optimal for both $\b^{L}$ and $\b^{H}$.
	% However, this is not possible, since in this case $\Opt$ would break the tie always in favor of the same scheduling.
	%If instead $y^{L}_j \cdot H_j > MS(\y^{L},\b^{L})$, then it must be the case that \pao{Useless (?), we already have a contradiction}
	%$$
	% MS(\y^{L}, \b^{H}) = y^{L}_j \cdot H_j < y^{H}_j \cdot H_j \leq MS(\y^{H}, \b^{H}),
	%$$
	%that contradicts the optimality of $\y^{H}$ with respect to $\b^{H}$.
	
	\medskip
	The lemma then follows from our choice for the tie-breaking rule.
	Suppose indeed that there is a machine $i \neq j$ that receives less jobs in $\y^{H}$ than in $\y^{L}$.
	Then it must exist also another machine $k \neq j$ that receives more jobs in $\y^{H}$ than in $\y^{L}$.
	Observe that it must be the case that, starting from $\y^H$, moving a job from $k$ to $i$ does not increase the makespan, i.e., $b_i (y^{H}_i + 1) \leq MS(\y^{H},\b^{H})$,
	otherwise $MS(\y^{L},\b^{L}) \geq b_i y_i^L \geq b_i (y^{H}_i + 1) > MS(\y^{H},\b^{H}) \geq MS(\y^{H},\b^{L})$, which contradicts the optimality of $\y^{L}$ with respect to $\b^{L}$.
	Hence, both $\y^{H}$ and the assignment $\bar{\y}^{H}$ achieved from $\y^{H}$ by moving a job from machine $k$ to machine $i$ are optimal for $\b^{H}$.
	Since $\y^{H}$ has been returned, it must be because
	it satisfies the tie-breaking condition for each pair of machines.
	In particular it must hold that
	\begin{equation}
	\label{eq:tie}
	b_i \cdot (y^{H}_i + 1) \geq b_k \cdot y^{H}_k
	\end{equation}
	with strict inequality if $b_k > b_i$, or $b_k = b_i$ and $k > i$.
	
	Moreover, as showed above, it must be the case that
	$b_i \cdot (\bar{y}^{H}_i + 1) \leq b_k \cdot \bar{y}^{H}_k$
	(with strict inequality if $b_k < b_i$ or $b_k = b_i$ and $k < i$).
	Since only one of the two tie-breaking conditions may fail,  we have that
	$b_k \cdot y^{H}_k = b_k \cdot (\bar{y}^{H}_k + 1) \geq b_i \cdot \bar{y}^{H}_i = b_i \cdot (y^{H}_i + 1)$
	(with strict inequality if $b_k < b_i$ or $b_k = b_i$ and $k < i$),
	that contradicts \eqref{eq:tie}.
\end{proof}

Two next  corollaries  follow from a repeated application of Lemma~\ref{lem:opt_prop}.
\begin{corollary}
	\label{cor:1}
	For any subset  $S$ of agents (machines), for any $\b_S$, and for any $i \in S$,  it holds that $\Opt_i(\b_S, \L_{-S}) \leq \Opt_i(\b_S,\H_{-S})$.
\end{corollary}

\begin{corollary}
	\label{cor:2}
	For any partition of the agents (machines) into three subsets $(S,T,\{k\})$, and for every $\b_S$, $\b_T$,  and $b_k$,
	it holds that \[\Opt_k(\b_S, b_k, \L_T) \leq \Opt_k(\b_S, b_k, \b_T) \leq \Opt_k(\b_S, b_k, \H_T). \]
\end{corollary}

\begin{lemma}
	For any subset $S \subset [n]$ of agents (machines), and for every $\b_S$, there is an agent $i_S \notin S$ such that
	\begin{equation}
	\label{eq:pending}
	\Opt_{i_S}(\b_S, \L_{-S}) \geq \Opt_{i_S}(\b_S, \H_{-S}).
	\end{equation}
\end{lemma}
\begin{proof}
	We use Corollary~\ref{cor:1} to compare the total number of jobs allocated to $S$ in the two inputs, namely, we have that $\sum_{i \in S}\Opt_i(\b_S,\L_{-S}) \le \sum_{i \in S}\Opt_i(\b_S,\H_{-S})$. As all jobs must be allocated to some machine, there must be a machine $i_S \not \in S$ such that  the reverse inequality holds for this machine, i.e., $\Opt_{i_S}(\b_S,\L_{-S}) \ge \Opt_{i_S}(\b_S,\H_{-S})$.
\end{proof}

Consider now the following mechanism:

\block{
	Mechanism  $\M_{\Opt}$:
	\begin{enumerate}
		\item Initially set $S = \emptyset$ and $\b_S = \emptyset$;
		\item While $S \subset [n]$ do the following:
		\begin{itemize}
			\item Pick an arbitrary agent  $i_S\not \in S$ satisfying \eqref{eq:pending}; 
			\item Ask to $i_S$ her bid $b_{i_S}$;
			\item Add $i_S$ to $S$ and update $\b_S$ accordingly;
		\end{itemize}
		\item Return $\Opt(\b)$.
	\end{enumerate}
}

\begin{theorem}\label{th:scheduling:two-value:heterogeneous}
	For any two-value domain $D_i = \{L_i, H_i\}$, the above mechanism $\M_{\Opt}$ is OSP. 
\end{theorem}
\begin{proof}
	We show that $\Opt$ satisfies OSP 2-CMON. By definition of $\M_{\Opt}$, each node where some agent diverges is a node $u_{i_S}$, and the agent $i_S$ that diverges satisfies \eqref{eq:pending}. We  next show that  $i_S$ will receive at least $\Opt_{i_S}(\b_S,\L_{-S})$ jobs
	when her  bid is $b_{i_S} = L_{i_S}$, and at most $\Opt_{i_S}(\b_S,\H_{-S})$ jobs
	when her bid is $b_{i_S} = H_{i_S}$:
	\begin{align}
	\Opt_{i_S}(\b_S,L_{i_S},\b_T) \geq & \  \Opt_{i_S}(\b_S,L_{i_S},\L_T)   = \Opt_{i_S}(\b_S,\L_{-S}) \\
	\Opt_{i_S}(\b_S,H_{i_S},\b'_T) \leq & \  \Opt_{i_S}(\b_S,H_{i_S},\H_T)   = \Opt_{i_S}(\b_S,\H_{-S})
	\end{align} 
	where the inequalities are due to Corollary~\ref{cor:2}. Since $i_S$ satisfies \eqref{eq:pending}, i.e., $\Opt_{i_S}(\b_S,\L_{-S}) \geq \Opt_{i_S}(\b_S,\H_{-S})$, the above inequalities imply $\Opt_{i_S}(\b_S,L_{i_S},\b_T) \geq \Opt_{i_S}(\b_S,L_{i_S},\b'_T)$ for all $\b_S$, $\b_T$ and $\b'_T$.
\end{proof}

%\begin{example}[counterexample to different jobs]\label{exa:scheduling}\pao{check}
%	Consider three machines and three jobs of size $10$, $6$, and $6$, and a two-values domain $\{L_i,H_i\}=\{L,H\}$ with $L=1$ and $H=2+\epsilon$,  for some small $\epsilon$ to be specified (below). The algorithm $OPT$ which minimizes the makespan, breaking ties in favor of the machines with smaller index, produces the following two allocations: $OPT(L,L,H) = (6+6,10,0)$ and $OPT(L,H,H)=(10,6,6)$.
%	Though algorithm $OPT$ is a natural generalization of algorithm $\Opt$ for identical jobs, the condition of Lemma~\ref{lem:opt_prop} is no longer satisfied for $j=2$ and $i=1$.
%\end{example}

%\input{asc_desc_sched}

\subsection{Ascending and Descending Auctions have linear approximation ratio for machine scheduling}
\label{apx:asc_desc_sched}%\car{Aggiunta def di $u_i$ e aggiornato}
Consider $m = n$ and,
for every $i$, $L_i = L$, $M_i = M \geq n^2 \cdot L$, and $H_i = H \geq n^2 \cdot M$.

Consider a mechanism $\M$ and let $\T$ be its implementation tree. We restrict our attention to \emph{full} mechanisms $\M$, i.e., mechanisms where all agents diverge on the path from the root of $\T$ where players play according to $H$. 

\begin{lemma}
	\label{lem:nonfull}
	Every non-full mechanism for the machine scheduling problem has approximation ratio at least $n$.
\end{lemma}
\begin{proof}
	Suppose $\M$ is a non-full mechanism for machine scheduling. This means that there exists at least one agent, say $i$, who will not diverge between $H$ and $L$ when all other agents have played according to $H$. But then the mechanism must return the same outcome on both the profiles $\x=(L, \bi)$ and $\y=(H, \bi)$ where $b_j=H$, for all $j \neq i$. Now if the outcome does not assign all jobs to $i$ then the approximation of $\M$ on $\x$ is $n$; on the contrary, if $\M$ assigns all jobs to $i$ then the approximation of $\M$ on $\y$ is $n$.
\end{proof}

Let us now rename the agents as follows: Agent $1$ is the $1$st agent that diverges in $\T$; note that agent $1$ is well defined for all non-trivial mechanisms. Agent $2$ is the $2$nd agent that diverges, different from agent $1$, in the subtree of $\T$ defined by agent $1$ taking an action compatible with type $H$; more generally, agent $i$ is the $i$th agent that diverges, different from the agents $1, 2, \ldots, i-1$, in the subtree of $\T$ in which the actions taken by agents that previously diverged are compatible with their type being $H$. We let $u_i$ denote the node of $\T$ in which $i$ diverges as from above. Moreover, we let $c_i = M$ if, at node $u_i$, $i$ diverges between $L$ and $M$, and $c_i = H$ otherwise.

We have the following results, whose
%\begin{theorem}
%\label{thm:special}
% Every full OSP mechanism for the machine scheduling problem such that either $c_1 =M$ or $c_i = H$ for every $i < n$ has approximation ratio at least $n$.
%\end{theorem}
proofs follow the same ideas of the one given in Section~\ref{sec:sched3_lb}.

\begin{theorem}
	\label{lem:c1_low}
	Every non-trivial OSP mechanism for the machine scheduling problem such that $c_1 = M$ has approximation ratio at least $n$.
\end{theorem}
\begin{proof}
	Suppose that there is a $k$-approximate OSP mechanism $\M$ with $k < n$ --- $\M$ is clearly non-trivial.  Consider the type profile $\x$ such that $x_1 = M$, and $x_j = H$ for every $j \neq 1$.
	Observe that $\x$ is compatible with node $u_1$.
	The optimal allocation for the type profile $\x$ assigns all jobs to machine $1$,
	with cost $OPT(\x) = n \cdot M$.
	Since $\M$ is $k$-approximate, then it also assigns all jobs to machine $1$.
	Indeed, if a job is assigned to a machine $j \neq 1$, then the cost of the mechanism would be at least
	$H \geq n^2 \cdot M > k \cdot OPT(\x)$, that contradicts the approximation bound.
	
	Consider now the profile $\y$ such that $y_j = L$ for every $j$.
	Observe that also $\y$ is compatible with node $u_1$.
	Clearly, $OPT(\y) = L$.
	Since $\M$ is $k$-approximate, then it cannot assign all jobs to machine $1$.
	Indeed, in this case the cost of the mechanism would be
	$n L > k \cdot OPT(\y)$, that contradicts the approximation bound.
	
	Hence, we have that if agent $1$ takes actions compatible with $M$, then there exists a type profile compatible with $u_1$ such that $1$ receives $n$ jobs,
	whereas, if $1$ takes a different action compatible with a lower type, then there exists a type profile compatible with $u_1$ such that $1$ receives less than $n$ jobs.
	However, this contradicts that the mechanism is OSP.
\end{proof}

\begin{lemma}
	\label{lem:ci_high}
	For every $i < n$, and every full OSP $k$-approximate mechanism $\M$, with $k < n$,
	if $c_i = H$ and $i$ at node $u_i$ takes an action compatible with her type being $H$, then $\M$ does not assign any job to $i$,
	regardless of the actions taken by the other machines.
\end{lemma}
\begin{proof}
	Suppose that there is $i < n$ and $\x_{-i}$ compatible with $u_i$ such that if $i$ takes an action compatible with type $H$, then $\M$ assigns a job to $i$.
	According to the definition of $c_i$, machine $i$ diverges at node $u_i$ on $H$ and $M$.
	
	Consider then the profile $\y$ such that $y_i = M$, $y_j = H$ for $j < i$, and $y_j = L$ for $j > i$.
	It is easy to see that the optimal allocation has cost $OPT(\y)=\left\lceil\frac{n}{n-i}\right\rceil \cdot L$.
	Since $\M$ is $k$-approximate, then it does not assign any job to machine $i$.
	Indeed, in this case the cost of the mechanism would be at least
	$M \geq n^2 L \geq n \cdot \left\lceil\frac{n}{n-i}\right\rceil L > k \cdot OPT(\x),$
	that contradicts the approximation bound.
	
	Hence, we have that if $i$ takes actions compatible with $H$, then there exists a type profile compatible with $u_i$ such that $i$ receives one job,
	whereas, if $i$ takes a different action compatible with a lower type, then there exists a type profile compatible with $u_i$ such that $i$ receives zero jobs.
	However, this contradicts that the mechanism is OSP.
\end{proof}

\begin{theorem}
	\label{thm:special}
	Every full OSP mechanism for the machine scheduling problem such that $c_i = H$ for every $i < n$ has approximation ratio at least $n$.
\end{theorem}

\begin{proof}
	Suppose now that there is an OSP $k$-approximate mechanism $\M$, for some $k < n$, such that $c_i = H$ for every $i < n$.
	Consider $\x$ such that $x_i = H$ for every $i$.
	Observe that $\x$ is compatible with $u_i$ for every $i$.
	The optimal allocation consists in assigning each job to a different machine, and has cost $OPT(\x) = H$.
	
	According to Lemma~\ref{lem:ci_high}, if machines take actions compatible with $\x$,
	then the mechanism $\M$ does not assign any job to machine $i$ for every $i < n$.
	Hence, the outcome that $\M$ returns for type profile $\x$ consists in assigning all jobs to the remaining machine.
	Therefore, the cost of $M$ is $n \cdot H > k OPT(\x)$, that contradicts the approximation bound.
\end{proof}

This result implies a lower bound of $n$ for the approximation of ascending and descending auctions.

\begin{corollary}
	Every ascending or descending auction has approximation ratio at least $n$ for the machine scheduling problem.
\end{corollary}
\begin{proof}
	An ascending mechanism has $c_1=M$ and, therefore, the lower bound follows from Theorem~\ref{lem:c1_low}.
	
	A full descending auction, instead, has $c_i=H$ for all $i$ and then in particular for all $i < n$.  The lower bound then follows from Lemma \ref{lem:nonfull} and Theorem~\ref{thm:special}.
\end{proof}

\subsection{Proof of Proposition~\ref{prop:algoapprox}}

\begin{proof}[Proof of Proposition~\ref{prop:algoapprox}]
	We show that the mechanism $\Mmany$ returns  $\left(\frac{m+\left\lceil\sqrt{n}\right\rceil-1}{m} \cdot \left\lceil\sqrt{n}\right\rceil\right)$-approx\-imate allocation. % with respect to the input.
 We denote with $OPT(\x)$ the makespan of the optimal assignment when machines have type profile $\x$.
 We will use the same notation both if the optimal assignment is computed on a set of $n$ machines
 and if it is computed and on a set of $\left\lceil\sqrt{n}\right\rceil$ machines, since these two cases
 can be distinguished through the input profile.

 Fix a type profile $\x$. %Let $OPT(\x)$ be the makespan of the optimal assignment on type profile $\x$.
 Let $A$ be as at the beginning of the ascending phase.
 Let $w$ be the last query done in this phase and let $i$ be the last queried agent, i.e., the one that answered yes.
 Moreover, let $t = \min_{j \notin A} t_j$.
 It is immediate to see that $OPT(\x) \geq OPT(\y)$ where $\y$ is such that $y_j = w$ for every $j \in A$, and $y_j = t$, otherwise.
 Moreover, let $\M(\x)$ be the makespan of the assignment returned by our mechanism on the same input.
 Then, $\M(\x)$ is equivalent to $OPT(\hat{\z})$, where $\hat{\z}$ is such that $\hat{z}_i = w$ and $\hat{z}_j = t$ for each remaining machine $j$.
 Hence, the theorem follows by proving that $\frac{OPT(\hat{\z})}{OPT(\y)} \leq \frac{m+\left\lceil\sqrt{n}\right\rceil-1}{m} \cdot \left\lceil\sqrt{n}\right\rceil$.

 Observe that the optimal assignment on input $\hat{\z}$ assigns $\alpha$ jobs to machine $i$ and the remaining $m-\alpha$ jobs to the remaining $\left\lceil\sqrt{n}\right\rceil - 1$ machines, for some $\alpha \in \{1, \ldots, m\}$.
 In case of multiple assignment achieving the optimal makespan, we pick the one that maximizes $\alpha$.
 Hence, $OPT(\hat{\z}) = \max\left\{\alpha \cdot w, \left\lceil\frac{m-\alpha}{\left\lceil\sqrt{n}\right\rceil - 1}\right\rceil \cdot t\right\}$.
 We next show that $OPT(\hat{\z}) \leq \frac{(m+\left\lceil\sqrt{n}\right\rceil-1)tw}{t + (\left\lceil\sqrt{n}\right\rceil-1)w}$.

 Suppose first that $OPT(\hat{\z}) = \alpha \cdot w$. Since this is the optimal assignment,
 then it must be the case that
 $$\alpha \cdot w \leq \max\left\{(\alpha-1) \cdot w, \left\lceil\frac{m-\alpha+1}{\left\lceil\sqrt{n}\right\rceil - 1}\right\rceil \cdot t\right\} = \left\lceil\frac{m-\alpha+1}{\left\lceil\sqrt{n}\right\rceil - 1}\right\rceil \cdot t = \frac{m-\alpha+1+\delta}{\left\lceil\sqrt{n}\right\rceil - 1} \cdot t,$$
 where $\delta = 0$ if $r = (m-\alpha+1) \mod \left(\left\lceil\sqrt{n}\right\rceil - 1\right) = 0$, and
 $\delta = \left\lceil\sqrt{n}\right\rceil - 1 - r$, otherwise.
 Hence, we have that
 $\alpha \leq \frac{(m+1+\delta)t}{t + (\left\lceil\sqrt{n}\right\rceil-1)w}$, and, consequently, $OPT(\hat{\z}) \leq \frac{(m+1+\delta)tw}{t + (\left\lceil\sqrt{n}\right\rceil-1)w} \leq \frac{(m+\left\lceil\sqrt{n}\right\rceil-1)tw}{t + (\left\lceil\sqrt{n}\right\rceil-1)w}$.

 If instead $OPT(\hat{\z}) = \left\lceil\frac{m-\alpha}{\left\lceil\sqrt{n}\right\rceil - 1}\right\rceil \cdot t = \frac{m-\alpha+\delta'}{\left\lceil\sqrt{n}\right\rceil - 1}$,
 where $\delta' = \left\lceil\sqrt{n}\right\rceil - 2$ if $\delta=0$ and $\delta' = \delta-1$, otherwise,
 then it must be the case that, by our tie-breaking rule among optimal assignments,
 $\frac{m-\alpha+\delta'}{\left\lceil\sqrt{n}\right\rceil - 1} \cdot t < (\alpha + 1)w$,
 from which we have $\alpha > \frac{(m+\delta')t - (\left\lceil\sqrt{n}\right\rceil-1)w}{t + (\left\lceil\sqrt{n}\right\rceil-1)w}$, and, consequently,
 $$OPT(\hat{\z}) < \frac{m + \delta' - \frac{(m+\delta')t - (\left\lceil\sqrt{n}\right\rceil-1)w}{t + (\left\lceil\sqrt{n}\right\rceil-1)w}}{\left\lceil\sqrt{n}\right\rceil - 1} \cdot t = \frac{(m+1+\delta')tw}{t + (\left\lceil\sqrt{n}\right\rceil-1)w} \leq \frac{(m+\left\lceil\sqrt{n}\right\rceil-1)tw}{t + (\left\lceil\sqrt{n}\right\rceil-1)w}.$$

 The optimal assignment on input $\y$ assigns $\beta$ jobs to the $\left\lceil\sqrt{n}\right\rceil$ machines in $A$,
 for some $\beta \in \{1, \ldots, m\}$,
 and
%  the remaining 
 $m-\beta$ jobs to the remaining $n - \left\lceil\sqrt{n}\right\rceil$ machines.
 Hence, $OPT(\y) =\max\left\{\left\lceil\frac{\beta}{\left\lceil\sqrt{n}\right\rceil}\right\rceil \cdot w, \left\lceil\frac{m-\beta}{n- \left\lceil\sqrt{n}\right\rceil}\right\rceil \cdot t\right\}$.
 We next show that $OPT(\y) \geq \frac{mtw}{\left\lceil\sqrt{n}\right\rceil t + (n-\left\lceil\sqrt{n}\right\rceil)w}$.

 If $OPT(\y) = \left\lceil\frac{\beta}{\left\lceil\sqrt{n}\right\rceil}\right\rceil \cdot w = \frac{\beta + \gamma}{\left\lceil\sqrt{n}\right\rceil} \cdot w$,
 where $\gamma = 0$ if $r = \beta \mod \left\lceil\sqrt{n}\right\rceil = 0$, and
 $\gamma = \left\lceil\sqrt{n}\right\rceil - r$, otherwise,
 then $\frac{\beta + \gamma}{\left\lceil\sqrt{n}\right\rceil} \cdot w \geq \left\lceil\frac{m-\beta}{\left\lceil\sqrt{n}\right\rceil - 1}\right\rceil \cdot t = \frac{m-\beta+\lambda}{n-\left\lceil\sqrt{n}\right\rceil} \cdot t$,
 where $\lambda = 0$ if $\rho = (m-\beta) \mod (n-\left\lceil\sqrt{n}\right\rceil) = 0$, and
 $\lambda = n - \left\lceil\sqrt{n}\right\rceil - \rho$, otherwise.
 Hence, we have $\beta \geq \frac{(m+\lambda)\left\lceil\sqrt{n}\right\rceil t - (n-\left\lceil\sqrt{n}\right\rceil)\gamma w}{\left\lceil\sqrt{n}\right\rceil t + (n-\left\lceil\sqrt{n}\right\rceil)w}$, and, consequently,
 $OPT(\y) \geq \frac{(m+\gamma+\lambda)tw}{\left\lceil\sqrt{n}\right\rceil t + (n-\left\lceil\sqrt{n}\right\rceil)w} \geq \frac{mtw}{\left\lceil\sqrt{n}\right\rceil t + (n-\left\lceil\sqrt{n}\right\rceil)w}$.

 If $OPT(\y) = \left\lceil\frac{m-\beta}{n- \left\lceil\sqrt{n}\right\rceil}\right\rceil \cdot t = \frac{m-\beta+\lambda}{n-\left\lceil\sqrt{n}\right\rceil}$, then,
 $\frac{\beta + \gamma}{\left\lceil\sqrt{n}\right\rceil} \cdot w \leq \frac{m-\beta+\lambda}{n-\left\lceil\sqrt{n}\right\rceil} \cdot t$.
 Hence, $\beta \leq \frac{(m+\lambda)\left\lceil\sqrt{n}\right\rceil t - (n-\left\lceil\sqrt{n}\right\rceil)\gamma w}{\left\lceil\sqrt{n}\right\rceil t + (n-\left\lceil\sqrt{n}\right\rceil)w}$, and, consequently,
 $OPT(\y) \geq \frac{(m+\gamma+\lambda)tw}{\left\lceil\sqrt{n}\right\rceil t + (n-\left\lceil\sqrt{n}\right\rceil)w} \geq \frac{mtw}{\left\lceil\sqrt{n}\right\rceil t + (n-\left\lceil\sqrt{n}\right\rceil)w}$.

 Thus, $\frac{OPT(\hat{\z})}{OPT(\y)} \leq \frac{m+\left\lceil\sqrt{n}\right\rceil-1}{m} \cdot \frac{\left\lceil\sqrt{n}\right\rceil t + (n-\left\lceil\sqrt{n}\right\rceil)w}{t + (\left\lceil\sqrt{n}\right\rceil-1)w} \leq \frac{m+\left\lceil\sqrt{n}\right\rceil-1}{m} \cdot \left\lceil\sqrt{n}\right\rceil$, where we used that $n-\left\lceil\sqrt{n}\right\rceil \leq \left\lceil\sqrt{n}\right\rceil^2 - \left\lceil\sqrt{n}\right\rceil = \left\lceil\sqrt{n}\right\rceil (\left\lceil\sqrt{n}\right\rceil - 1)$.
\end{proof}

\end{document}